\documentclass[10pt]{article} 
\usepackage{amsthm}
\theoremstyle{definition}
\newtheorem{defi}{Definition}[section]
\newtheorem{remark}[defi]{Remark}

\newtheorem{example}[defi]{Example}

\theoremstyle{plain}
\newtheorem{theorem}[defi]{Theorem}

\newtheorem{lemma}[defi]{Lemma}
\newtheorem{prop}[defi]{Proposition}

\newtheorem{asum}[defi]{Assumption}
\usepackage{cite}
\usepackage{amsmath, amsfonts}
\usepackage[english]{babel}
\usepackage{enumerate}
\usepackage [T1]{fontenc}
\usepackage{endnotes}
\usepackage{xcolor}
\usepackage{abstract}
\usepackage{enumitem,blindtext}
\usepackage[paper=a4paper,left=30mm,right=30mm,top=30mm,bottom=30mm]{geometry}
\allowdisplaybreaks

\numberwithin{equation}{section}

\newcommand{\red}[1]{\textcolor{black}{#1}}

\newcommand*\samethanks[1][\value{footnote}]{\footnotemark[#1]}

\DeclareMathOperator*{\esssup}{ess\,su{p}^{\textit{P}}}
\DeclareMathOperator*{\esssupT}{ess\,su{p}^{\textit{$\tilde{P}$}}}

\DeclareMathOperator*{\essinf}{ess\,inf^{\textit{$\tilde{P}$}}}

\title{Reduced-form framework for multiple ordered default times under model uncertainty}
\author{Francesca Biagini\thanks{Workgroup Financial and Insurance Mathematics, Department of Mathematics, Ludwig-Maximilians Universit{\"a}t, Theresienstrasse 39, 80333 Munich, Germany. Emails: mazzon@math.lmu.de, biagini@math.lmu.de.} \and Andrea Mazzon\samethanks[1] \and Katharina Oberpriller\thanks{Department of Mathematics, University of Freiburg, Ernst-Zermelo-Stra\ss e 1, 79104 Freiburg, Germany}}
\begin{document}
\maketitle
\begin{abstract}
	In this paper we introduce a sublinear conditional operator with respect to a family of possibly nondominated probability measures in presence of multiple ordered default times. In this way we generalize the results of \cite{bz_2019}, where a reduced-form framework under model uncertainty for a single default time is developed. Moreover, we use this operator for the valuation of credit portfolio derivatives under model uncertainty.
\end{abstract}


\section{Introduction}

In this paper we introduce a reduced-form framework for multiple \red{ordered} default times under model uncertainty. To this purpose we define a sublinear conditional operator with respect to a family of probability measures possibly mutually singular to each other in presence of multiple ordered default times. In this way we extend the classical literature on credit risk in presence of multiple defaults, see for example \cite{Ehlers_Schoenbucher_2009}, \cite{El_Karoui_Jeanblanc_Jioa_2015},  \cite{el_karoui_jeanblanc_jiao_2018} and \cite{jeanblanc_li_song_2018} to the case of a setting where many different probability models can be taken into account.  \\
Over the last years, several different approaches have been developed in order to establish such robust settings which are independent of the underlying probability distribution, see among others \cite{acciaio_larsson_2017}, \cite{bbkn_2017}, \cite{denis_hu_peng_2010}, \cite{denis_kervarec_2013}, \cite{denis_martini_2006}, \cite{guo_pan_peng_2017}, \cite{hu_peng_2013}, \cite{matoussi_2013}, \cite{neufeld_nutz_2016}, \cite{peng_2007}, \cite{peng_song_zhang_2014}, \cite{soner_touzi_zhang_MRP} and \cite{tevzadze_toronjadze_uzunashvili_2012}. However, the above results hold only on the canonical space endowed with the natural filtration. In credit risk and insurance modeling it is fundamental to model multiple random events occurring as a surprise, such as defaults in a network of financial institutions or the loss occurrences of a portfolio of policy holders. This requires to consider filtrations with a dependence structure. Such a problem is mentioned in \cite{aksamit_hou_obloj_2012} and solved for an initial enlarged filtration. In \cite{bz_2019} they define a sublinear conditional operator with respect to a filtration which is progressively enlarged by one random time.\\
In this paper we extend the approach in \cite{bz_2019} and define a sublinear conditional operator with respect to a filtration progressively enlarged by multiple ordered stopping times. Such an extension is connected to several additional technical challenges with respect to the construction in \cite{bz_2019}.\\
First, we cannot consider default times in all generality, but we need to focus on a family of ordered stopping times. In particular, we work in the setting of the top-down model for increasing default times introduced \red{in \cite{Ehlers_Schoenbucher_2009}, in order to model the loss of CDOs}, as a generalization to the well known Cox model in \cite{lando}.
More specifically, we start with a reference filtration $\mathbb{F}$  and define a family of ordered stopping times $\tau_1,...,\tau_N$, in a similar way as done in \cite{Ehlers_Schoenbucher_2009}. We then progressively enlarge $\mathbb{F}$  with the filtrations $\mathbb{H}^n$ generated by $(\textbf{1}_{\lbrace \tau_{n} \le t \rbrace})_{t \ge 0}$, $n=1,...,N$, and define $\mathbb{G}^{(n)}:=\mathbb{F} \vee \mathbb{H}^1 \vee ... \vee \mathbb{H}^n$, $n=1,...,N$. {In our case, we construct $\tau_1 <... <\tau_N$ in such a way that $\tilde{\tau}_n:=\tau_n-\tau_{n-1}$ is independent of $\mathcal{H}_t^{n-1}$ for any $n=2,...,N,t \geq 0$ conditionally on $\mathcal{F}_{\infty}$. In particular, the intensities \red{of the stopping times} are driven by $\mathbb{F}$-adapted stochastic processes which may be used to model dependence structures driven by common risk factors and also contagion effects.} 
 We first address the problem of computing $\mathbb{G}^{(N)}$-conditional expectations of a given random variable under one given prior in terms of a sum of $\mathbb{F}$-conditional expectations depending on how many defaults have happened before time $t$. 
This is also a new contribution to the literature on ordered multiple default times in the classical case, i.e., in presence of only one probability measure. For an analogous result following the density approach for modeling successive default times, we refer to \cite{El_Karoui_Jeanblanc_Jioa_2015}. The main technical issue in our setting is to compute conditional expectations when a strictly positive number of defaults, but not \emph{all} the $N$ defaults, have happened. 
 Already under a fixed prior the results for multiple ordered default times are not a trivial extension of the ones in a single default setting.\\
 We then use this representation to define the sublinear conditional operator $\tilde{\mathcal{E}}^N$ under model uncertainty with respect to the progressively enlarged filtration $\mathbb{G}^{(N)}$. As in \cite{bz_2019}, our definition makes use of the sublinear conditional operator introduced by Nutz and van Handel in \cite{nh_2013} with respect to $\mathbb{F}$. To this purpose we assume that $\mathbb{F}$ is given by the canonical filtration. In particular, we show that our construction is consistent with the ones in \cite{nh_2013} in presence of no default and in \cite{bz_2019} for $N=1$, respectively. The main technical challenge is to prove a weak dynamic programming principle for the operator as done in \cite{bz_2019} for the single default setting, as it requires to exchange the order of integration between the operator and expectations under a given prior. This problem also appears in \cite{bz_2019}, however it becomes significantly more complex in the case when some defaults have already happened, but not all of them. In this situation we cannot any longer rely on the techniques used in \cite{bz_2019} but have to come up with additional steps, see Lemma \ref{lem:secondlemmasecondpiece} and Proposition \ref{prop:dynamicsecond}.
 We then use the conditional sublinear operator to evaluate credit portfolio derivatives under model uncertainty. In particular, we focus on the valuation of the so called $i$-th to default contingent claims CCT$^{(i)}$, for $i=1,...,N$.
   This subject has been widely studied under one reference probability measure, see for example \cite{arnsdorf_halperin_2008}, \cite{cont_minca_2011}, \cite{Ehlers_Schoenbucher_2009}, \cite{El_Karoui_Jeanblanc_Jioa_2015}, \cite{filipovic_overbeck_schmidt_2010}, \cite{giesecke_goldberg_ding_2009},\cite{sidenius_piterbarg_andersen_2008} and the survey in \cite{frey_mcneil_2003}. Especially, intensity models are often used in option pricing under credit risk \red{and in particular in CDOs valuation}, see e.g., references in \cite{bielecki_crepey_jeanblanc_2010}, \cite{brigo_pallavicini_torresetti_2007}, \cite{El_Karoui_Jeanblanc_Jioa_2015}, \cite{Ehlers_Schoenbucher_2009}, \cite{filipovic_overbeck_schmidt_2010}, \cite{giesecke_goldberg_ding_2009} and \cite{laurent_cousin_fermanian_2010}. \red{Namely, the construction of ordered defaults stems from modelling the loss of CDOs. Furthermore, ordered default times are important for the valuation of basket credit derivatives, see e.g. \cite{chen2008fast}, \cite{choe2011k}, \cite{mortensen2006semi}, \cite{bujok2012numerical} and the modelling of surrender times, see \cite{le2013surrender}.}
   Here we extend these approaches to a robust framework, by using $\tilde{\mathcal{E}}^N$ to evaluate payoffs under a worst case scenario. In particular, for the $i$-th to default contingent claims we are able to derive sufficient conditions under which the strong tower property holds. As done in \cite{bz_2019} for the single default case, we can also establish a relation between the sublinear conditional operator and the superhedging problem in a multiple default setting for a generic payment streams under given conditions.
   
The paper is organized as follows. In Section \ref{sec:coxmodel} we present the generalized Cox model for successive default times. In Section \ref{sec:twotimes} we evaluate conditional expectations in presence of multiple ordered stopping times on a classical probability space. In Section \ref{sec:MultipleDefaultModelUncertainty} we provide a reduced-form setting under model uncertainty for multiple default times. In particular, we define the sublinear conditional operator $\tilde{\mathcal{E}}^{N}$ and analyze its properties. In Section \ref{section:DynamicProgramming} we derive the weak dynamic programming principle for the operator $\tilde{\mathcal{E}}^{N}$. In Section \ref{sec:valuation} we apply the operator to the valuation of credit portfolio derivatives. In Section \ref{sec:Superhedging} we {briefly} analyze superhedging for payment streams in a robust setting with multiple defaults.

\section{Multiple ordered default times in the Cox model} \label{sec:coxmodel}
\red{We} let $(\Omega, \mathcal{F}, P)$ be a probability space \red{where $\mathcal{F}:=\mathcal{B}(\Omega)$ denotes the Borel-$\sigma$-algebra and equip this space with} the reference filtration $\mathbb{F}=(\mathcal{F}_t)_{t \geq 0}$. Moreover, we consider an additional probability space $(\hat{\Omega}, \mathcal{B}(\hat{\Omega}), \hat{P})$. We then define the product probability space 
\begin{equation} \label{eq:ProductSpace}
	(\tilde{\Omega}, \mathcal{G},\tilde{P}):= (\Omega \times \hat{\Omega}, \mathcal{B}(\Omega) \otimes \mathcal{B}(\hat{\Omega}), P \otimes \hat{P}).
\end{equation}
We use the notation $\tilde{\omega}=(\omega,\hat{\omega})$ for $\omega \in \Omega$ and $\hat{\omega}\in\hat{\Omega}$. For every function $X$ on $(\Omega,\mathcal{B}(\Omega))$ we denote its natural immersion into the product space by $X(\tilde{\omega}):=X(\omega)$ for all $\omega \in \Omega$, and similarly for functions on $(\hat{\Omega}, \mathcal{B}(\hat{\Omega}))$. Furthermore, for every sub-$\sigma$-algebra $\mathcal{A}$ of $\mathcal{B}(\Omega)$, its natural extension as a sub-$\sigma$-algebra of $\mathcal{G}$ on $(\tilde{\Omega}, \mathcal{G})$ is given by $\mathcal{A} \otimes \lbrace \emptyset, \tilde{\Omega} \rbrace$, and similarly for sub-$\sigma$-algebras of $\mathcal{B}(\hat{\Omega})$. \red{In the sequel, for simplicity  we  indicate $\mathcal{A}\otimes \lbrace \emptyset, \tilde{\Omega}\rbrace$ with $\mathcal{A}$ for any sub-$\sigma$-algebra $\mathcal{A}$ of $\mathcal{B}(\Omega)$, and analogously for any sub-$\sigma$-algebra of $\mathcal{B}(\hat{\Omega})$.} In the following, we construct on $(\tilde{\Omega}, \mathcal{G}, \tilde{P})$ an increasing sequence of $N$ random times $0 <\tau_1 <...<\tau_N < \infty$, as done in \cite{Ehlers_Schoenbucher_2009}, giving a generalization of the Cox model.

\red{We} now introduce the random times construction on $(\tilde{\Omega}, \mathcal{G}, \tilde{P})$, along with the associated enlarged filtrations.
\begin{defi}\label{def:stoppingtimes} 
Let \red{$E_1,...,E_N$ be positive random variables on $(\hat{\Omega}, \mathcal{B}(\hat{\Omega}),\hat{P})$.} 
 For each $n=1,\dots,N$, iteratively proceed along the following steps:
	\begin{enumerate}
		\item Choose a non-negative $ \mathbb{G}^{(n-1)}$-adapted process $\lambda^n$ such that $\int_0^t \lambda_s^n ds < \infty$ $\tilde{P}$-a.s. for all $t < \infty$ and $\int_0^{\infty} \lambda_s^n ds=\red{+}\infty$ $\tilde{P}$-a.s. with $\mathbb{G}^{(0)}:= \mathbb{F}$.
		\item Define $\tau_{n}:= \inf \lbrace t > \tau_{n-1}: \int_{\tau_{\red{n-1}}}^t \lambda_s^n ds \geq E_{n} \rbrace$ with $\tau_0:=0$.
		\item Define the process $H^{n}=(H^{n}_t)_{t \geq 0}$ by $H^{n}_t:=\textbf{1}_{\lbrace \tau_{n} \leq t \rbrace}$ for $t \geq 0$, and denote by $\mathbb{H}^{n}\red{:=(\mathcal{H}^n_t)_{t \geq 0}}$ the filtration generated by $H^{n}$.
		\item Define the enlarged filtration $\mathbb{G}^{(n)}\red{:=(\mathcal{G}_t^n)_{t \geq 0}}$ by $\mathbb{G}^{(n)}:=\mathbb{G}^{(n-1)} \vee \mathbb{H}^{n}$, again using the convention $\mathbb{G}^{(0)}:= \mathbb{F}$.
	\end{enumerate}
\end{defi}
\begin{remark}
	Our construction slightly differs from the one given in Definition 4.2 of \cite{Ehlers_Schoenbucher_2009} in the following two points. First, we do not assume that the initial filtration $\mathbb{F}$ satisfies the usual hypothesis. Second, in \cite{Ehlers_Schoenbucher_2009}, the enlarged filtration corresponding to the first $n$ random times is defined as the smallest filtration which contains $\mathbb{F}$, satisfies the usual hypothesis and makes $\tau_1,...,\tau_n$ stopping times. Such a filtration is not equal to $\mathbb{G}^{(n)}$, as it can be shown that $\mathbb{G}^{(n)}$ is in general not right-continuous. However, the results of \cite{Ehlers_Schoenbucher_2009} that are useful for our purposes also hold when $\mathbb{F}$ does not satisfy the usual hypothesis and for our choice of the enlarged filtrations $\mathbb{G}^{(n)}$, $n=1,\dots,N$.
\end{remark}

We here summarize some results proved in Proposition 4.3 in \cite{Ehlers_Schoenbucher_2009}.
\begin{prop}\label{prop:properties}
	The stopping times introduced in Definition \ref{def:stoppingtimes} satisfy the following properties:
	\begin{enumerate}
		\item $\tau_0 < \tau_1<...<\tau_N< \infty$ $\tilde{P}$-a.s.
		\item \red{If the random variables $E_1,...,E_N$ are unit-exponentially distributed under $\hat{P}$, then }\linebreak$\tilde{P}\left(\tau_{n+1} >t \vert \mathcal{G}_{\infty}^{(n)}\right)=e^{-\int_{t \wedge \tau_n}^t \lambda_s^n ds}$, $t \ge 0$.
		\item $\tau_{n+1}$ avoids the $\mathbb{G}^{(n)}$-stopping times, i.e., for every $\mathbb{G}^{(n)}$-stopping time $\Theta$, $\tilde{P}(\tau_{n+1}=\Theta)=0$. 
	\end{enumerate}
\end{prop}
Note that $\tilde{P}\left(\tau_{n+1} >t \vert \mathcal{G}_{t}^{(n)}\right)= \tilde{P}\left(\tau_{n+1} >t \vert \mathcal{G}_{\infty}^{(n)}\right)=e^{-\int_{t \wedge \tau_n}^t \lambda_s^n ds}$, $t \ge 0$. \\

We now give a construction of $N$ random times $0 <\tau_1 <...<\tau_N < \infty$ as above \red{under the following assumption.}
\red{ \begin{asum} \label{asum:ExistenceUniformlyRV}
\begin{enumerate}
	\item On $(\hat{\Omega}, \mathcal{B}(\hat{\Omega}))$ there are $N$ random variables $E_1,...,E_N$ which are i.i.d. unit-exponentially distributed under $\hat{P}$.
	\item On $(\tilde{\Omega}, \mathcal{G})$ the random variables $E_1,...,E_N$ are independent of $\mathcal{F}_{\infty}$ under $\tilde{P}$.
\end{enumerate}
\end{asum}}
Consider the i.i.d. unit-exponentially distributed random variables $E_1,...,E_N$ introduced in Assumption \ref{asum:ExistenceUniformlyRV}. Moreover, let {$\tilde{\lambda}^1,...,\tilde{\lambda}^N$ be $\mathbb{F}$-adapted non-negative processes} such that for $n=1,...,N$ and $0 \leq t < \infty$ it holds
\begin{equation} \label{eq:AssumptionsLambda}
	\int_0^t \tilde{\lambda}_s^n ds < \infty, \quad \int_0^{\infty} \tilde{\lambda}^n_s ds = \red{+} \infty \quad \red{{P}\text{-a.s.}}.
\end{equation}
We now define for each $n=1,...,N$ the random variable $\tilde{\tau}_n$ by
\begin{equation} \label{eq:DefiTauFixedProbability}
	\tilde{\tau}_n=\inf\left\{ t > 0: \int_0^t \tilde{\lambda}^n_s ds \geq E_n \right\}
\end{equation}
	and the random variable $\tau_n$ by
\begin{equation} \label{eq:DefiTildeTauFixedProbability}
	\tau_n=\sum_{k=1}^n 	\tilde{\tau}_k.
\end{equation}
{By construction and by Assumption \ref{asum:ExistenceUniformlyRV}, we have that $\tilde{\tau}_n=\tau_n-\tau_{n-1}$ is independent of $\mathcal{H}^{n-1}_t$ given $\mathcal{F}_{\infty}$, for any $n=2,\dots,N$, $t \geq 0$.}
\\
\red{
\begin{lemma}
Let Assumption \ref{asum:ExistenceUniformlyRV} be satisfied and the random times $(\tau_n)_{n =0,...,N}$ be defined as in \eqref{eq:DefiTildeTauFixedProbability}. Then, for all $n=1,...,N$ there exists a $\mathbb{G}^{(n-1)}$-adapted process $\lambda^n$ such that for all $0 \leq t < \infty$
\begin{equation*}
	\int_0^t {\lambda}_s^n ds < \infty, \quad \int_0^{\infty} {\lambda}^n_s ds = \red{+}\infty \quad \red{\tilde{P}\text{-a.s.}}
\end{equation*}
and 
\begin{equation*}
	\tau_n=\inf \left\{ t > \tau_{n-1}: \int_{\tau_{n-1}}^t \lambda_s^{n}ds \geq E_n \right\}
\end{equation*}
setting $\tau_0:=0$.
\end{lemma}}
\red{
\begin{proof}
It is clear that $\tau_1=\tilde{\tau}_1$, so that we first choose $\lambda^1:=\tilde{\lambda}^1$. For a general $n=2,\dots,N$, define
\begin{equation} \label{eq:DefinitionLambdaRevision}
\lambda_t^n:=\textbf{1}_{\lbrace t \geq {\tau}_{n-1} \rbrace}\tilde{\lambda}^n_{t-\tau_{n-1}}, \quad t \ge 0.
\end{equation}
The process $\lambda^n$ defined in this way is clearly $\mathbb{G}^{(n-1)}$-adapted. Moreover, we have
\begin{align*}
	\tau_n&=\sum_{k=1}^n \tilde{\tau}_k = \sum_{k=1}^{n-1} \tilde{\tau}_k + \tilde{\tau}_n \notag \\ 
	&
	=\tau_{n-1} +\inf \left\lbrace t >0: \int_0^t \tilde{\lambda}^n_s ds \geq E_n\right\rbrace \notag \\ &
	= \tau_{n-1} +\inf \left\lbrace t >0: \int_{{\tau}_{n-1}}^{{\tau}_{n-1}+t} {\lambda}^n_s ds \geq E_n\right\rbrace \notag \\ &
	= \inf \left\lbrace s >\tau_{n-1}: \int_{\tau_{n-1}}^{s} {\lambda}^n_u du \geq E_n\right\rbrace.
\end{align*}
\end{proof}}

\red{
\begin{example} \label{example:SelfExciting}
Consider the family of stopping times $(\tau_i)_{i=0,1,\dots,N}$ defined in \eqref{eq:DefiTildeTauFixedProbability} where $(\tilde{\tau}_n)_{n =1,...,N}$ are given by \eqref{eq:DefiTauFixedProbability} with $\tilde{\lambda}^n=(\tilde{\lambda}^n_t)_{t \ge 0}$ defined by
\begin{equation} \label{eq:TildeLambdaSelfExciting}
\tilde{\lambda}_t^n:=\mu_t + n e^{- \gamma t}, \quad t \ge 0.
\end{equation}
Here $\gamma>0$ is a constant and $\mu=(\mu_t)_{t \ge 0}$ is an $\mathbb{F}$-adapted and positive stochastic process such that
\begin{equation*}
	\int_0^t \mu_s ds < \infty, \quad \int_0^{\infty} \mu_s ds = \red{+}\infty \quad \red{P\text{-a.s.}}
\end{equation*}
Note that $\tilde{\lambda}^n$ is $\mathbb{F}$-adapted and satisfies  \eqref{eq:AssumptionsLambda}.
With this construction the $n$-th default time $\tau_n$ has intensity $\lambda^n=(\lambda^n_t)_{t \ge 0}$ given by
\begin{equation}\label{eq:lambdanself}
\lambda_t^n:=\textbf{1}_{\lbrace t \geq {\tau}_{n-1} \rbrace}\tilde{\lambda}^n_{t-\tau_{n-1}}= \textbf{1}_{\lbrace t \geq \tau_{n-1} \rbrace}\left(\mu_{t- {\tau}_{n-1}} + n e^{- \gamma (t- {\tau}_{n-1})}\right), \quad t \ge 0.
\end{equation}
 The stochastic process $N=(N_t)_{t \ge 0}$ defined by 
\begin{equation}\notag
N_t:=\sum_{i=1}^{N} \mathbf{1}_{\{{\tau}_i \le t\}}, \quad t \ge 0
\end{equation}
has jump intensity $\bar{\lambda}=(\bar{\lambda}_t)_{t \ge 0}$ with 
\begin{equation}\label{eq:selfexcitingcounting}
\bar{\lambda}_t := \mu_{t-\bar\tau_{N_t}}+(N_t+1) e^{-\gamma(t-\bar{\tau}_{N_t})}, \quad t \ge 0,
\end{equation}
which is indeed equal to \eqref{eq:lambdanself} conditional to $\mathbf{1}_{\{\tau_{n-1} \le t <\tau_n\}}$, as in this case $N_t=n-1$. \\
The counting process $N$ obtained in this way has a self-exciting tendency since by \eqref{eq:TildeLambdaSelfExciting} we have $\tilde{\lambda}^n \geq \tilde{\lambda}^{n-1}$ $P$-a.s. and thus
\begin{equation} \notag
P \left(\tilde{\tau}_{n}>t \vert \mathcal{F}_t \vee \sigma (E_{n})\right)  = e^{-\int_0^t \tilde{\lambda}^n_s ds} < e^{-\int_0^t \tilde{\lambda}^{n-1}_sds }= P \left(\tilde{\tau}_{n-1}>t \vert \mathcal{F}_t \vee \sigma (E_{n-1})\right).
\end{equation}
If in addition $\mu$ is decreasing, the jump intensity of $N$ increases further after any jump because of the second term in the right-hand side of \eqref{eq:selfexcitingcounting}. In this case, $N$ is \emph{self-exciting} in the sense that
$$
Cov(N_t-N_s, N_s - N_u) > 0 \text{ for any $t>s>u$.}
$$
\end{example}
}

In the rest of the paper, we \red{work under Assumption \ref{asum:ExistenceUniformlyRV} and} use the default times $(\tau_i)_{1 \leq i \leq N}$ constructed as in \eqref{eq:DefiTauFixedProbability} and \eqref{eq:DefiTildeTauFixedProbability} with corresponding filtrations $\mathbb{H}^i, \mathbb{G}^{(i)}$ for $i=1,...,N$, $\mathbb{G}^{(0)}=\mathbb{F}$, as introduced in Definition \ref{def:stoppingtimes}. 
\begin{remark}
	With this construction it holds 
\begin{align}
\mathcal{G}=\mathcal{G}^{(N)}&=\mathcal{F}_{\infty} \otimes \red{(}\sigma(E_1)\red{\vee} ... \red{\vee} \sigma(E_N)\red{)} \notag \\ &
= \mathcal{F}_{\infty} \vee \mathcal{H}_{\infty}^1\vee...\vee \mathcal{H}_{\infty}^N\notag \\&
=\mathcal{F}_{\infty}\vee \sigma(\tau_1)\vee ... \vee \sigma(\tau_N)\notag \\&
=\mathcal{F}_{\infty} \vee \sigma(\tilde{\tau}_1)\vee ... \vee \sigma(\tilde{\tau}_N).\notag
\end{align}
Moreover, on the event $\lbrace \tau_1 \leq t \rbrace$ we have 
\begin{equation} \label{eq:RemarkFiltrations}
	\mathcal{F}_t \vee \sigma(E_1) =\mathcal{F}_t \vee \sigma(\tau_1), \quad t \ge 0.
\end{equation}
This follows as
$$
 \mathcal{F}_t \vee \sigma(E_1) \subseteq \mathcal{F}_t \vee \sigma(\tau_1), \quad t \ge 0
$$
{because $E_1=\int_0^{\tau_1} \tilde{\lambda}_s^1 ds $ and}
$$
 \mathcal{F}_t \vee \sigma(\tau_1) \subseteq \mathcal{F}_t \vee \sigma(E_1), \quad t \ge 0
$$
by the definition of $\tau_1$ in \eqref{eq:DefiTauFixedProbability} and $\eqref{eq:DefiTildeTauFixedProbability}$.
\end{remark}

Finally, we introduce the notation $\boldsymbol{\tau}:=(\tau_1,...,\tau_{N})$ and the convention $\tau_0:=0$, $\tau_{N+1}:=+\infty$. Given an $N$-dimensional vector $\boldsymbol{u}=(u_1,...,u_N)$ we denote by $\boldsymbol{u}_{(k)}$, for $k \leq N$, the $k$-dimensional vector containing the first $k$-entries of $\boldsymbol{u}$. 
Furthermore, for every $P \in \mathcal{P} $ let $\mathcal{G}^P:=\mathcal{G} \vee \mathcal{N}_{\infty
}^{P}$, where $\mathcal{N}_{\infty}^{P}$ is the collection of sets which are $(P, \mathcal{F}_{\infty})$-null. We define the set
\begin{align} \label{eq:defL1P}
	L^1_{\tilde{P}}({\tilde{\Omega}}):=\lbrace \tilde{X} \vert  \tilde{X}: (\tilde{\Omega}, \mathcal{G}^{P}) \to (\mathbb{R}, \mathcal{B}(\mathbb{R})) \text{ measurable function such that  }
	\mathbb{E}^{\tilde{P}}[\vert \tilde{X} \vert] < \infty \rbrace.
\end{align}

\section{General pricing results for successive defaults on a given probability space }\label{sec:twotimes}
In the setting outlined in Section \ref{sec:coxmodel} we now extend the sublinear conditional expectation of \cite{bz_2019} to the case of multiple defaults. To this purpose we need to express the conditional expectation of a random variable $Y \in L^1_{\tilde{P}}(\tilde{\Omega})$ with respect to $\mathcal{G}_t^{(N)}$ in terms of conditional expectations with respect to $\mathcal{F}_t$, for any $t \ge 0$. In a first step we split the conditional expectation $\mathbb{E}^{\tilde{P}}[Y \vert \mathcal{G}_t^{(N)}]$ as
\begin{align} \label{eq:SplitConditionalExpectationsGeneral}
	\mathbb{E}^{\tilde{P}} \left [ Y \vert \mathcal{G}_t^{(N)} \right]= \mathbb{E}^{\tilde{P}} \left [ \textbf{1}_{\lbrace t < \tau_1\rbrace} Y \vert \mathcal{G}_t^{(N)}\right] + \sum_{k=1}^{N-1}\mathbb{E}^{\tilde{P}}\left [ \textbf{1}_{\lbrace \tau_{k} \leq t < \tau_{k+1}\rbrace} Y \vert \mathcal{G}_t^{(N)} \right]+ \mathbb{E} ^{\tilde{P}}\left [ \textbf{1}_{\lbrace \tau_N \leq t \rbrace} Y \vert \mathcal{G}_t^{(N)}\right]
\end{align}
for any $t \ge 0$. We now analyze the terms on the right-hand side of \eqref{eq:SplitConditionalExpectationsGeneral}.
We start by the following lemma which is a generalization of Lemma 2.13 in \cite{bz_2019} and which can also be seen as an extension of a result in Section 2.2.2 in \cite{El_Karoui_Jeanblanc_Jioa_2015}. 
\begin{lemma} \label{lemma:DecompositionGeneraln}
For any $\mathcal{G}^{(N)}_{\infty}$-measurable random variable $Y$ there exists a unique measurable function 
$$
\varphi:\left(\mathbb{R}_+^{N}\times\Omega,\mathcal{B}(\mathbb{R}_+^{N})\otimes \mathcal{F}_{\infty}\right) \to \left(\mathbb{R},\mathcal{B}(\mathbb{R})\right)
$$
such that 
\begin{equation} \label{eq:Decomposition}
	Y(\omega, \hat{\omega}) =  \varphi(\tau_1(\omega, \hat{\omega}), \dots,\tau_{N}(\omega, \hat{\omega}), \omega), \quad (\omega,\hat{\omega}) \in \tilde{\Omega}.
\end{equation}
\end{lemma}
\begin{proof}
	The result follows by the same arguments used in the proof of Lemma 2.13 in \cite{bz_2019}.
\end{proof}

We now provide some preliminary lemmas. 

\begin{lemma} \label{lemma:SecondPartGeneral}
	Let $Y \in L^1_{\tilde{P}}(\tilde{\Omega})$. Then
	\begin{equation} \label{eq:SecondPartGeneralLemma}
		\textbf{1}_{\lbrace \tau_k \leq t \rbrace} \mathbb{E}^{\tilde{P}}\left[ Y \vert \mathcal{F}_t \vee \mathcal{H}_t^1 \vee ... \vee \mathcal{H}_t^k\right]= \textbf{1}_{\lbrace \tau_k \leq t \rbrace} \mathbb{E}^{\tilde{P}}\left[ Y \vert \mathcal{F}_t \vee \sigma(E_1) \vee ... \vee \sigma(E_k)\right] \quad \tilde{P} \text{-a.s.}
	\end{equation}
	 for any $k=1,...,N$ and $t \geq 0$.
\end{lemma}
\begin{proof}
	We first first consider the case $k=1$ to show the main arguments of the proof in a simpler setting.  For any $t \ge 0$ we have
	
\begin{align}
\textbf{1}_{\lbrace \tau_1 \leq t \rbrace} \mathbb{E}^{\tilde{P}}[Y \vert \mathcal{F}_t \vee \mathcal{H}_t^1]&= \textbf{1}_{\lbrace \tau_1 \leq t \rbrace} \mathbb{E}^{\tilde{P}}[Y \vert \mathcal{F}_t \vee \mathcal{H}_{\infty}^1] \nonumber\\
&= \textbf{1}_{\lbrace \tau_1 \leq t \rbrace} \mathbb{E}^{\tilde{P}}[Y \vert \mathcal{F}_t \vee \sigma(\tau_1)] \nonumber\\
&=\textbf{1}_{\lbrace \tau_1 \leq t \rbrace}\mathbb{E}^{\tilde{P}}[Y \vert\mathcal{F}_t \vee \sigma(E_1)].\label{eq:SecondTermGeneralK=1}
\end{align}	

Here the first equality follows by a generalization of Corollary 5.1.2 in \cite{bielecki_rutkowski_2004}, and the second one by the definition of $\mathbb{H}^1$. 
This together with \eqref{eq:RemarkFiltrations} proves equation \eqref{eq:SecondTermGeneralK=1}. Let now $k=2,...,N-1$. For any $t \ge 0$, we get
\begin{align}
	&\textbf{1}_{\lbrace \tau_k \leq t \rbrace} \mathbb{E}^{\tilde{P}}\left[ Y \vert \mathcal{F}_t \vee \mathcal{H}_t^1 \vee ... \vee \mathcal{H}_t^k\right] \nonumber \\
	&=\textbf{1}_{\lbrace \tau_k \leq t \rbrace} \mathbb{E}^{\tilde{P}}\left[ Y \vert \mathcal{F}_t \vee \mathcal{H}_t^1 \vee ... \vee  \mathcal{H}_t^{k-1} \vee \sigma(E_k)\right] \label{eq:SecondTermGeneral1} \\
	&=\textbf{1}_{\lbrace \tau_k \leq t \rbrace} \mathbb{E}^{\tilde{P}}\left[ Y (\textbf{1}_{\lbrace \tau_{k-1} \leq t \rbrace}+ \textbf{1}_{\lbrace \tau_{k-1} > t \rbrace} )\vert \mathcal{F}_t \vee \mathcal{H}_t^1 \vee ... \vee  \mathcal{H}_t^{k-1} \vee \sigma(E_k)\right] \nonumber\\
	&=\textbf{1}_{\lbrace \tau_k \leq t \rbrace} \textbf{1}_{\lbrace \tau_{k-1} \leq t \rbrace} \mathbb{E}^{\tilde{P}}\left[ Y \vert \mathcal{F}_t \vee \mathcal{H}_t^1 \vee ... \vee  \mathcal{H}_t^{k-1} \vee \sigma(E_k)\right] \nonumber\\
	&=\textbf{1}_{\lbrace \tau_k \leq t \rbrace} \textbf{1}_{\lbrace \tau_{k-1} \leq t \rbrace} \mathbb{E}^{\tilde{P}}\left[ Y \vert \mathcal{F}_t \vee \mathcal{H}_t^1 \vee ... \vee \mathcal{H}_{t}^{k-2} \vee  \sigma(E_k)  \vee  \mathcal{H}_t^{k-1}\right] \nonumber \\
	&=\textbf{1}_{\lbrace \tau_k \leq t \rbrace} \textbf{1}_{\lbrace \tau_{k-1} \leq t \rbrace} \mathbb{E}^{\tilde{P}}\left[ Y \vert \mathcal{F}_t \vee \mathcal{H}_t^1 \vee ... \vee \mathcal{H}_{t}^{k-2} \vee  \sigma(E_k)  \vee  \sigma(E_{k-1})\right] \label{eq:SecondTermGeneral2} \\
	&=\textbf{1}_{\lbrace \tau_k \leq t \rbrace} \mathbb{E}^{\tilde{P}}\left[ Y \vert \mathcal{F}_t \vee \sigma(E_1) \vee ... \vee \sigma(E_k)\right]	\label{eq:SecondTermGeneral3}.
\end{align}
Here \eqref{eq:SecondTermGeneral1} and \eqref{eq:SecondTermGeneral2} follow by the same arguments used to prove \eqref{eq:SecondTermGeneralK=1} applied to $\mathcal{F}_t \vee \mathcal{H}_t^1 \vee ... \vee  \mathcal{H}_t^{k-1}$ on the event $\lbrace \tau_k \leq t \rbrace$ and to $\mathcal{F}_t  \vee \mathcal{H}_t^1  \vee ... \vee \mathcal{H}_{t}^{k-2} \vee  \sigma(E_k)$ on the event $\lbrace \tau_{k-1} \leq t \rbrace$, respectively. By recursively repeating the same procedure we get \eqref{eq:SecondTermGeneral3}.
\end{proof}

\begin{lemma} \label{prop:SecondPartGeneral}
For any $\psi:(\mathbb{R}_+^N \times \Omega, \mathcal{B}(\mathbb{R}_+^N) \otimes \mathcal{F}_{\infty}) \to (\mathbb{R}, \mathcal{B}(\mathbb{R}))$, $t \geq 0$  and $k=1,...,N$, it holds
\begin{align}
	&\textbf{1}_{\lbrace \tau_{k} \leq t  \rbrace} \mathbb{E}^{\tilde{P}}\left[ \psi(\boldsymbol{\tau}_{(k)}, \tilde{\tau}_{k+1},...,\tilde{\tau}_N, \cdot) \vert \mathcal{F}_t \vee \sigma(E_1) \vee ... \vee \sigma(E_k)\right] \nonumber \\
	&=\textbf{1}_{\lbrace \tau_{k} \leq t  \rbrace} \mathbb{E}^{\tilde{P}}\left[ \psi(\boldsymbol{u}_{(k)}, \tilde{\tau}_{k+1},...,\tilde{\tau}_N, \cdot) \vert \mathcal{F}_t\right] \big\vert_{\boldsymbol{u}_{(k)}=\boldsymbol{\tau}_{(k)}}. \label{eq:SecondPartGeneralProp}
\end{align}
\end{lemma}
\begin{proof}
	We use a monotone class argument. Let $A \in \mathcal{F}_{\infty}$ and $s_l,\tilde{s}_j>0$ for $l=1,...,k$, $j=k+1,...,N$. Then we get
	\red{\begin{align}
		&\textbf{1}_{\lbrace \tau_{k} \leq t \rbrace} \mathbb{E}^{\tilde{P}}\left[ \textbf{1}_A \prod_{l=1}^k\textbf{1}_{\lbrace \tau_l \leq s_l \rbrace} \prod_{j=k+1}^N\textbf{1}_{\lbrace \tilde{\tau}_j \leq \tilde{s}_j \rbrace} \bigg\vert \mathcal{F}_t  \vee \sigma(E_1) \vee ... \vee \sigma(E_k)\right] \nonumber \\
		&=\prod_{l=1}^k\textbf{1}_{\lbrace \tau_l \leq s_l \wedge t\rbrace} \mathbb{E}^{\tilde{P}}\left[ \textbf{1}_A \prod_{j=k+1}^N\textbf{1}_{\lbrace \tilde{\tau}_j \leq \tilde{s}_j \rbrace} \bigg \vert \mathcal{F}_t \vee \sigma(E_1) \vee ... \vee \sigma(E_k)\right] \nonumber \\
		&=\prod_{l=1}^k\textbf{1}_{\lbrace \tau_l \leq s_l \wedge t\rbrace} \mathbb{E}^{\tilde{P}}\left[ \textbf{1}_A \prod_{j=k+1}^N\textbf{1}_{\lbrace \int_0^{\tilde{s}_j} \tilde{\lambda}^j_u du \leq E_j\rbrace} \bigg \vert \mathcal{F}_t \right] \label{eq:SecondPartGeneralProp1} \\ 
		&=\textbf{1}_{\lbrace \tau_{k} \leq t \rbrace} \left(\prod_{l=1}^k\textbf{1}_{\lbrace u_l \leq s_l \rbrace} \mathbb{E}^{\tilde{P}}\left[ \textbf{1}_A \prod_{j=k+1}^N\textbf{1}_{\lbrace \tilde{\tau}_j \leq \tilde{s}_j \rbrace}\bigg \vert \mathcal{F}_t\right]  \right) \bigg\vert_{\boldsymbol{u}_{(k)}=\boldsymbol{\tau}_{(k)}}\nonumber \\
		&=\textbf{1}_{\lbrace \tau_{k} \leq t \rbrace}  \mathbb{E}^{\tilde{P}}\left[ \prod_{l=1}^k\textbf{1}_{\lbrace u_l \leq s_l\rbrace}  \textbf{1}_A \prod_{j=k+1}^N\textbf{1}_{\lbrace \tilde{\tau}_j \leq \tilde{s}_j \rbrace}\bigg \vert \mathcal{F}_t\right] \bigg\vert_{\boldsymbol{u}_{(k)}=\boldsymbol{\tau}_{(k)}}\nonumber.
	\end{align}}
	Here \eqref{eq:SecondPartGeneralProp1} \red{comes from} Assumption \ref{asum:ExistenceUniformlyRV}.
	Equation \eqref{eq:SecondPartGeneralProp} then follows by using a monotone class argument as in the proof of Lemma 2.15 in \cite{bz_2019}.
\end{proof}

\begin{lemma} \label{lemma:GeneralMiddlePartAuxiliary}
Let $Y \in L^1_{\tilde{P}}(\tilde{\Omega})$. Then 
\begin{align}
	\mathbb{E}^{\tilde{P}}\left[ \textbf{1}_{\lbrace \tau_{k} \leq t < \tau_{k+1} \rbrace}Y \vert \mathcal{G}_t^{(N)} \right]= \textbf{1}_{\lbrace \tau_k \leq t < \tau_{k+1}\rbrace} \frac{\mathbb{E}^{\tilde{P}}\left[\textbf{1}_{\lbrace \tilde{\tau}_{k+1}>t-\tau_k\rbrace}Y\vert \mathcal{G}_t^{(k)}\right]}{{\tilde{P}}(\tau_{k+1}>t \vert \mathcal{G}_t^{(k)})}. \label{eq:AuxilaryLemmaReductionFiltration}
	\end{align}
	for any $t \geq 0$ and $k=0,...,N$.
\end{lemma}
\begin{proof}
	We have 
	\red{\begin{align}
	&\mathbb{E}^{\tilde{P}} \left[ \textbf{1}_{\lbrace \tau_{k} \leq t < \tau_{k+1} \rbrace}Y \vert \mathcal{G}_t^{(N)} \right]\nonumber \\
	&= \mathbb{E}^{\tilde{P}} \left[ \textbf{1}_{\lbrace \tau_{k} \leq t < \tau_{k+1} \rbrace}Y \prod_{j=k+2}^N \textbf{1}_{\lbrace t < \tau_j\rbrace}\bigg\vert \mathcal{G}_t^{(k)} \vee \mathcal{H}_t^{k+1} \vee ... \vee \mathcal{H}_t^N\right] \nonumber \\
	&=\textbf{1}_{\lbrace t < \tau_N \rbrace}\frac{\mathbb{E}^{\tilde{P}} \left[ \textbf{1}_{\lbrace \tau_{k} \leq t < \tau_{k+1} \rbrace}Y \prod_{j=k+2}^{N} \textbf{1}_{\lbrace t < \tau_j\rbrace}\bigg\vert \mathcal{G}_t^{(k)} \vee \mathcal{H}_t^{k+1} \vee ... \vee \mathcal{H}_t^{N-1} \right]}{{\tilde{P}}(\tau_N >t\vert \mathcal{G}_t^{(k)} \vee \mathcal{H}_t^{k+1} \vee ... \vee \mathcal{H}_t^{N-1}  )} \label{eq:SecondPartGeneralII1}\\
	&=\textbf{1}_{\lbrace t < \tau_N-1 \rbrace}\mathbb{E}^{\tilde{P}}\left[ \textbf{1}_{\lbrace \tau_{k} \leq t < \tau_{k+1} \rbrace}Y \prod_{j=k+2}^{N-1} \textbf{1}_{\lbrace t < \tau_j\rbrace}\bigg\vert \mathcal{G}_t^{(k)} \vee \mathcal{H}_t^{k+1} \vee ... \vee \mathcal{H}_t^{N-1} \right] \label{eq:SecondPartGeneralII2} \\
	&= \textbf{1}_{\lbrace \tau_k \leq t < \tau_{k+1}\rbrace} \frac{\mathbb{E}^{\tilde{P}}\left[\textbf{1}_{\lbrace \tilde{\tau}_{k+1}>t-\tau_k\rbrace}Y\vert \mathcal{G}_t^{(k)}\right]}{\tilde{P}(\tau_{k+1}>t \vert \mathcal{G}_t^{(k)})}.\label{eq:SecondPartGeneralII3}
	\end{align}}
	Equation \eqref{eq:SecondPartGeneralII1} follows from Lemma 5.1.2 (ii) in \cite{bielecki_rutkowski_2004}, and we get \eqref{eq:SecondPartGeneralII2} because
	\begin{equation} \label{eq:ConditionalExpectationOne}
		\tilde{P}\left( \tau_l > t \vert \mathcal{G}_t^{(k)} \vee \mathcal{H}_t^{k+1} \vee ... \vee \mathcal{H}_t^{l-1} \right)=1 \quad \text{ on } \quad \lbrace t < \tau_{l-1 }\rbrace
	\end{equation}
	for $l=k+1,...,N$, respectively.
	By recursively applying these arguments we get \eqref{eq:SecondPartGeneralII3}.
\end{proof}

\begin{remark}
In $\eqref{eq:SecondPartGeneralII3}$ we apply a slightly more general version of Lemma 5.1.2 (ii) in \cite{bielecki_rutkowski_2004}, since $Y$ is not $\mathcal{G}_{\infty}^{(l)}$-measurable for $l=k+1,...,N$ but $\mathcal{G}_{\infty}^{(N)}$-measurable. The proof of the result in \cite{bielecki_rutkowski_2004} can be adapted without significant changes to our case.
\end{remark}
We are now able to provide the following proposition, that reduces the $\mathbb{G}^{(N)}$-conditional expectations in \eqref{eq:SplitConditionalExpectationsGeneral} to $\mathbb{F}$-conditional expectations.
\begin{prop} \label{prop:SecondPartGeneralIII}
Let $Y \in L^1_{\tilde{P}}(\tilde{\Omega})$. Then 
	\begin{align*}
		\mathbb{E}^{\tilde{P}} \left[ \textbf{1}_{\lbrace \tau_k \leq t < \tau_{k+1}\rbrace} Y \vert \mathcal{G}_t^{(N)} \right]=  \textbf{1}_{\lbrace \tau_k \leq t < \tau_{k+1}\rbrace} e^{\int_0^{t-\tau_k}\tilde{\lambda}_s^{k+1} ds} \mathbb{E}^{\tilde{P}}\left [ \textbf{1}_{\lbrace\tilde{\tau}_{k+1} >t-u_k\rbrace}  \varphi(\bold{u}_{(k)}, \bold{u}_{k}^{\tilde{\tau}},\cdot) \vert \mathcal{F}_t \right]\big\vert_{\bold{u}_{(k)}=\boldsymbol{\tau}_{(k)}}, 
	\end{align*}
$\tilde{P}$-a.s. for any $t \geq 0$ and $k=1,...,N$, where $\varphi$ is the function introduced in Lemma \ref{lemma:DecompositionGeneraln}, i.e.,
\begin{align*}
	\varphi: (\mathbb{R}_+^{N}  \times  \Omega, \mathcal{B}(\mathbb{R}_+^N)  \otimes \mathcal{F}_{\infty}) &\to (\mathbb{R}, \mathcal{B}(\mathbb{R}))
\end{align*}
such that  
\begin{align*}
	Y(\omega, \hat{\omega})=\varphi(\boldsymbol{\tau}_{(N)}(\omega,\hat{\omega}),\omega), \quad { (\omega,\hat{\omega}) \in \tilde{\Omega},}
\end{align*}
and $\bold{u}_{k}^{\tilde{\tau}}=(u_{k,k+1}^{\tilde{\tau}},...,u_{k,N}^{\tilde{\tau}}) \in \mathbb{R}^{N-k}$ with $u_{k,l}^{\tilde{\tau}}:=u_k+ \sum_{m=k+1}^l \tilde{\tau}_m$ for $l=k+1,...,N$.
\end{prop}
\begin{proof}
For any $t \ge 0$ and $k=1,\dots,N$ we get
	\red{\begin{align}
		&\mathbb{E}^{\tilde{P}} \left[ \textbf{1}_{\lbrace \tau_k \leq t < \tau_{k+1}\rbrace} Y \vert \mathcal{G}_t^{(N)} \right] \nonumber \\
		&= \textbf{1}_{\lbrace \tau_k \leq t < \tau_{k+1}\rbrace} \frac{\mathbb{E}^{\tilde{P}}\left[\textbf{1}_{\lbrace \tilde{\tau}_{k+1}>t-\tau_k\rbrace}Y\vert \mathcal{G}_t^{(k)}\right]}{\tilde{P}(\tau_{k+1}>t \vert \mathcal{G}_t^{(k)})} \label{eq:secondPartGeneralIII1}\\
		&= \textbf{1}_{\lbrace \tau_k \leq t < \tau_{k+1}\rbrace} \frac{\mathbb{E}^{\tilde{P}}\left[\textbf{1}_{\lbrace \tilde{\tau}_{k+1}>t-\tau_k\rbrace}\varphi(\boldsymbol{\tau}_{(N)}, \cdot)\vert \mathcal{F}_t \vee \sigma(E_1) \vee ... \vee \sigma(E_k)\right]}{\tilde{P}(\tilde{\tau}_{k+1}>t-\tau_k \vert \mathcal{F}_t \vee \sigma(E_1) \vee ... \vee \sigma(E_k))} \label{eq:secondPartGeneralIII2} \\
		&= \textbf{1}_{\lbrace \tau_k \leq t < \tau_{k+1}\rbrace}  e^{\int_0^{t-\tau_k}\tilde{\lambda}_s^{k+1} ds}\mathbb{E}^{\tilde{P}}\left[\textbf{1}_{\lbrace \tilde{\tau}_{k+1}>t-\tau_k\rbrace}\varphi(\boldsymbol{u}_{(k)}, \boldsymbol{u}_k^{\tilde{\tau}} , \cdot)\vert \mathcal{F}_t\right]]\vert_{\boldsymbol{u}_{(k)}=\boldsymbol{\tau}_{(k)}}, \label{eq:secondPartGeneralIII4}
	\end{align}}
	where  \eqref{eq:secondPartGeneralIII1} follows directly from Lemma \ref{lemma:GeneralMiddlePartAuxiliary}. Moreover, we use Lemma \ref{lemma:SecondPartGeneral} to get \eqref{eq:secondPartGeneralIII2}, together with the fact that by Lemma \ref{lemma:DecompositionGeneraln} there exists a unique measurable 
	function $\varphi:(\mathbb{R}_+^N \times \Omega, \mathcal{B}(\mathbb{R}_+^N) \otimes \mathcal{F}_{\infty}) \to (\mathbb{R}, \mathcal{B}(\mathbb{R}))$ such that
	\begin{equation*}
		Y(\omega, \hat{\omega})= \varphi(\boldsymbol{\tau}_{(N)}(\omega, \hat{\omega}),\omega), \quad (\omega, \hat{\omega}) \in \tilde{\Omega}.
	\end{equation*}
	\red{Equation \eqref{eq:secondPartGeneralIII4} follows applying Lemma \ref{prop:SecondPartGeneral} to the functions
	\begin{align*}
		\psi_1: (\mathbb{R}_+^N \times \Omega, \mathcal{B}(\mathbb{R}_+^N) \otimes \mathcal{F}_{\infty}) \to (\mathbb{R}, \mathcal{B}(\mathbb{R})) \\
		\psi_2: (\mathbb{R}_+^2 \times \Omega, \mathcal{B}(\mathbb{R}_+^2) \otimes \mathcal{F}_{\infty}) \to (\mathbb{R}, \mathcal{B}(\mathbb{R}))
	\end{align*}
	defined by
	\begin{equation*}
		\psi_1(\boldsymbol{\tau}_{(k)}, \tilde{\tau}_{k+1},...,\tilde{\tau}_{N}, \omega):= \textbf{1}_{\lbrace \tilde{\tau}_{k+1} > t- \tau_k \rbrace} \varphi\left(\boldsymbol{\tau}_{(k)}, \tau_k + \tilde{\tau}_{k+1},...,\tau_k+ \sum_{j=k+1}^N \tilde{\tau}_{j},\omega \right)
	\end{equation*}
	and
	\begin{equation*}
		\psi_2 (\tau_k, \tilde{\tau}_{k+1},\omega):=\textbf{1}_{\lbrace \tilde{\tau}_{k+1}> t - \tau_k \rbrace,}
	\end{equation*}
and since
	\begin{align}
		\tilde{P}(\tilde{\tau}_{k+1} > t- u_k \vert \mathcal{F}_t )\vert_{u_k=\tau_k}&= \tilde{P}\left( E_{k+1} > \int_0^{t-u_k} \tilde{\lambda}_s^{k+1} ds \ \bigg \vert \ \mathcal{F}_t  \right) \bigg \vert_{u_k=\tau_k} \nonumber \\ 
		&=\tilde{P}\left( E_{k+1} > \int_0^{t-u_k} \tilde{\lambda}_s^{k+1} ds  \right) \bigg \vert_{u_k=\tau_k} \nonumber \\
		&= e^{-\int_0^{t-\tau_k} \tilde{\lambda}_s^{k+1} ds}\quad \tilde{P} \text{-a.s.}, 	\label{eq:ConditionalProbabilityWRTF} 
	\end{align}}
	because for any $t \ge 0$ we have that  $ \int_0^{t-u_k} \tilde{\lambda}_s^{k+1}ds$ is measurable with respect to $\mathcal{F}_t$  and $E_{k+1}$ is independent of $\mathcal{F}_t$ by Assumption \ref{asum:ExistenceUniformlyRV}.
\end{proof}

Let $\tilde{X}: (\tilde{\Omega}, \mathcal{G}) \to (\mathbb{R}, \mathcal{B}(\mathbb{R}))$ be a measurable function. Analogously to the setting of \cite{bz_2019}, and with a slight notational abuse, we introduce the notation
\begin{equation} \label{eq:NotationI}
	\mathbb{E}^{\hat{P}}[\tilde{X}](\omega):=\int_{\hat{\Omega}}\tilde{X}(\omega,\hat{\omega})\hat{P}(d\hat{\omega}), \quad \omega\in\Omega. 
\end{equation}
  When not needed, we do not make the dependence on $\omega$ explicit.

\begin{theorem} \label{prop:NTimesModelUncertainty}
Let $t \geq 0$. If $Y \in L_{\tilde{P}}^1(\tilde{\Omega})$, then 
\begin{align}
	\mathbb{E}^{\tilde{P}} \left [  Y \vert \mathcal{G}_t^{(N)} \right]&= \textbf{1}_{\lbrace t < \tau_1\rbrace} e^{\int_0^{t}\tilde{\lambda}_s^1 ds}\mathbb{E}^{P}\left[  \mathbb{E}^{\hat{P}} \left [ \textbf{1}_{\lbrace t < \tau_1\rbrace}  Y \right]\bigg\vert \mathcal{F}_t \right] \nonumber \\
	&+\sum_{k=1}^{N-1} \textbf{1}_{\lbrace \tau_k \leq t < \tau_{k+1}\rbrace} e^{\int_0^{t-\tau_k}\tilde{\lambda}_s^{k+1} ds} \mathbb{E}^{P} \left [ \mathbb{E}^{\hat{P}}\left[\textbf{1}_{\lbrace\tilde{\tau}_{k+1} >t-u_k\rbrace}  \varphi(\bold{u}_{(k)}, \bold{u}_{k}^{\tilde{\tau}},\cdot) \right]\bigg\vert \mathcal{F}_t \right]\bigg\vert_{\bold{u}_{(k)}=\boldsymbol{\tau}_{(k)}} \nonumber \\
	&+\textbf{1}_{\lbrace \tau_{N} \leq t \rbrace} \mathbb{E}^{P}[\varphi(\bold{u}_{(N)},\cdot)\vert \mathcal{F}_t] \big \vert_{\bold{u}_{(N)}=\boldsymbol{\tau}}, \quad \tilde{P} \text{-a.s.} \label{eq:RepresentationMultipleDefault}
\end{align}
where $\varphi$ is the function introduced in Lemma \ref{lemma:DecompositionGeneraln}, i.e.,
\begin{align}
	\varphi: (\mathbb{R}_+^{N}  \times  {\Omega}, \mathcal{B}(\mathbb{R}_+^N)  \otimes \mathcal{F}_{\infty}) &\to (\mathbb{R}, \mathcal{B}(\mathbb{R})) \label{eq:DefiVarphi1}
\end{align} 
such that  
\begin{align}
	Y(\omega, \hat{\omega})=\varphi(\boldsymbol{\tau}_{(N)}(\omega,\hat{\omega}),\omega), \label{eq:DefiVarphi2}
\end{align}
and $\bold{u}_{k}^{\tilde{\tau}}=(u_{k,k+1}^{\tilde{\tau}},...,u_{k,N}^{\tilde{\tau}}) \in \mathbb{R}^{N-k}$ with $u_{k,l}^{\tilde{\tau}}:=u_k+ \sum_{m=k+1}^l \tilde{\tau}_m$ for $l=k+1,...,N$.	
\end{theorem}

\begin{proof}
	Fix $t \geq 0$ {and} $Y \in L_{\tilde{P}}^1(\tilde{\Omega})$. By applying Lemma \ref{lemma:GeneralMiddlePartAuxiliary} to the case $k=0$ and Proposition \ref{prop:SecondPartGeneralIII} to the other terms in \eqref{eq:SplitConditionalExpectationsGeneral}, we get
	\begin{align}
	\mathbb{E}^{\tilde{P}} \left [  Y \vert \mathcal{G}_t^{(N)} \right]&= \textbf{1}_{\lbrace t < \tau_1\rbrace} e^{\int_0^{t}\tilde{\lambda}_s^1 ds}\mathbb{E}^{\tilde{P}} \left [ \textbf{1}_{\lbrace t < \tau_1\rbrace}  Y \vert \mathcal{F}_t \right] \nonumber \\
	&+\sum_{k=1}^{N-1} \textbf{1}_{\lbrace \tau_k \leq t < \tau_{k+1}\rbrace} e^{\int_0^{t-\tau_k}\tilde{\lambda}_s^{k+1} ds} \mathbb{E}^{\tilde{P}} \left [ \textbf{1}_{\lbrace\tilde{\tau}_{k+1} >t-u_k\rbrace}  \varphi(\bold{u}_{(k)}, \bold{u}_{k}^{\tilde{\tau}},\cdot) \big\vert \mathcal{F}_t \right]\big\vert_{\bold{u}_{(k)}=\boldsymbol{\tau}_{(k)}} \nonumber \\
	&+\textbf{1}_{\lbrace \tau_{N} \leq t \rbrace} \mathbb{E}^{\tilde{P}}[\varphi(\bold{u}_{(N)},\cdot)\vert \mathcal{F}_t] \big \vert_{\bold{u}_{(N)}=\boldsymbol{\tau}} \quad {\tilde{P}} \text{-a.s,} \nonumber
\end{align}
where $\varphi$ is the function introduced in Lemma \ref{lemma:DecompositionGeneraln}, i.e.,
\begin{align}
	\varphi: (\mathbb{R}_+^{N}  \times  {\Omega}, \mathcal{B}(\mathbb{R}_+^N)  \otimes \mathcal{F}_{\infty}) &\to (\mathbb{R}, \mathcal{B}(\mathbb{R})) \nonumber
\end{align} 
such that  
\begin{align}
	Y(\omega, \hat{\omega})=\varphi(\boldsymbol{\tau}_{(N)}(\omega,\hat{\omega}),\omega), \nonumber
\end{align}
and $\bold{u}_{k}^{\tilde{\tau}}=(u_{k,k+1}^{\tilde{\tau}},...,u_{k,N}^{\tilde{\tau}}) \in \mathbb{R}^{N-k}$ with $u_{k,l}^{\tilde{\tau}}:=u_k+ \sum_{m=k+1}^l \tilde{\tau}_m$ for $l=k+1,...,N$.

By Lemma 2.12 in \cite{bz_2019}, for $X \in L^1_{\tilde{P}}(\tilde{\Omega})$ and for any $t \geq 0$ we have
\begin{equation} \label{eq:ConditionalExpectationFrancescaYinglin}
	\mathbb{E}^{\tilde{P}}[X \vert \mathcal{F}_t]= \mathbb{E}^P[\mathbb{E}^{\hat{P}}[X] \vert \mathcal{F}_t] \quad \tilde{P} \text{-a.s.,}
\end{equation}
where $\mathbb{E}^{\hat{P}}[X]$ is introduced in \eqref{eq:NotationI}.
Therefore, we get 
\begin{equation}
	\mathbb{E}^{\tilde{P}} \left [ \textbf{1}_{\lbrace t < \tau_1\rbrace}  Y \vert \mathcal{F}_t \right] = \mathbb{E}^{P}\left[  \mathbb{E}^{\hat{P}} \left [ \textbf{1}_{\lbrace t < \tau_1\rbrace}  Y \right]\bigg\vert \mathcal{F}_t \right] \quad \tilde{P} \text{-a.s.} \label{eq:FirstTermProductProbability}
\end{equation}
and
\begin{small}
\begin{equation}
	\mathbb{E}^{\tilde{P}} \left [ \textbf{1}_{\lbrace\tilde{\tau}_{k+1} >t-u_k\rbrace}  \varphi(\bold{u}_{(k)}, \bold{u}_{k}^{\tilde{\tau}},\cdot) \vert \mathcal{F}_t \right]\big\vert_{\bold{u}_{(k)}=\boldsymbol{\tau}_{(k)}} = \mathbb{E}^{P} \left [ \mathbb{E}^{\hat{P}}\left[\textbf{1}_{\lbrace\tilde{\tau}_{k+1} >t-u_k\rbrace}  \varphi(\bold{u}_{(k)}, \bold{u}_{k}^{\tilde{\tau}},\cdot) \right]\bigg\vert \mathcal{F}_t \right]\bigg\vert_{\bold{u}_{(k)}=\boldsymbol{\tau}_{(k)}}, \label{eq:SecondTermProductProbability}
\end{equation}
\end{small}
$\tilde{P}$-a.s. for every $t\ge 0$. Moreover, 
\begin{equation}
	\mathbb{E}^{\tilde{P}}[\varphi(\bold{u}_{(N)},\cdot)\vert \mathcal{F}_t] \big \vert_{\bold{u}_{(N)}=\boldsymbol{\tau}}= \mathbb{E}^{P}[\varphi(\bold{u}_{(N)},\cdot)\vert \mathcal{F}_t] \big \vert_{\bold{u}_{(N)}=\boldsymbol{\tau}}, \label{eq:ThirdTermProductProbability}
\end{equation}
 $P$ a.s., for any $t \ge 0$.
Putting together \eqref{eq:FirstTermProductProbability}, \eqref{eq:SecondTermProductProbability} and \eqref{eq:ThirdTermProductProbability} we finish the proof.  
\end{proof}

\section{Sublinear conditional operator for multiple default times} \label{sec:MultipleDefaultModelUncertainty}
In this section we introduce a sublinear conditional operator with respect to a family of possibly nondominated probability measures, on a filtration enlarged with $N$ stopping times. The definition of this operator is based on the construction of $N$ ordered random times as in Section \ref{sec:coxmodel} and on the representation of the $\mathbb{G}^{(N)}$-conditional expectation in \eqref{eq:RepresentationMultipleDefault} from Theorem \ref{prop:NTimesModelUncertainty}. Moreover, the following construction is a generalization of the definition of sublinear conditional operator under model uncertainty with respect to a filtration enlarged by one random time, as introduced in \cite{bz_2019}.

Let $\Omega=D_0(\mathbb{R}_+, \mathbb{R})$ be the space of c\`{a}dl\`{a}g  functions $\omega=(\omega_t)_{t \geq 0}$ in $\mathbb{R}$ starting from zero, which is equipped with the metric induced by the Skorokhod topology. We consider the measurable space $(\Omega, \mathcal{F})$, where $\mathcal{F}:=\mathcal{B}(\Omega)$ is the Borel $\sigma$-algebra on $\Omega$. The set of probability measures on $(\Omega, \mathcal{F})$ is given by $\mathcal{P}(\Omega)$. We assume that $\mathcal{P}(\Omega)$ is endowed with the topology of weak convergence. Furthermore, we denote by $B:=(B_t)_{t \ge 0}$ the canonical process, i.e., $B_t(\omega)=\omega_t, \ t \geq 0$, and its corresponding raw filtration by $\mathbb{F}:=(\mathcal{F}_t)_{t \geq 0}$ with $\mathcal{F}_0=\lbrace \emptyset, \Omega \rbrace$ and $\mathcal{F}_{\infty}:=\bigvee_{t \geq 0} \mathcal{F}_t=\mathcal{F}$. For every given $P \in \mathcal{P}(\Omega)$ and $t \ge 0$, we define $\mathcal{N}_t^P$ as the collection of sets which are $(P, \mathcal{F}_t)$-null, and  consider the filtration $\mathbb{F}^*:=(\mathcal{F}^*_t)_{t \geq 0}$ defined by
\begin{equation}
	\mathcal{F}_t^*:=\mathcal{F}_t \vee \mathcal{N}_t^*, \quad \mathcal{N}_t^*:=\bigcap_{P \in \mathcal{P}(\Omega)} \mathcal{N}_t^P. \label{filtration}
\end{equation}
For a given family of probability measures $\mathcal{P}$ on $\Omega$ we define the $\sigma$-algebra $\mathcal{F}^{\mathcal{P}}$ by 
\begin{equation} \label{defiNullsets}
	\mathcal{F}^{\mathcal{P}}:=\mathcal{F} \vee \mathcal{N}_{\infty}^{\mathcal{P}}, \quad \mathcal{N}_{\infty}^{\mathcal{P}}:=\bigcap_{P \in \mathcal{P}} N_{\infty}^P.
\end{equation}

We follow the approach of \cite{nh_2013} for defining sublinear expectations and introduce the following notation. Let $\tau$ be a finite-valued $\mathbb{F}$-stopping time and $\omega \in \Omega$. For every $\red{\bar{\omega}} \in \Omega$, the concatenation process $\omega \otimes_{\tau} \red{\bar{\omega}}:=((\omega \otimes_{\tau} \red{\bar{\omega}})_t)_{t \geq 0}$ of $(\omega, \red{\bar{\omega}})$ at $\tau$ is given by
\begin{equation}
	(\omega \otimes_{\tau} \red{\bar{\omega}})_t:= \omega_t \textbf{1}_{[0,\tau(\omega))}(t)+(\omega_{\tau(\omega)}+ \red{\bar{\omega}}_{t-\tau(\omega)})\textbf{1}_{[\tau(\omega),+\infty)}(t), \quad t \geq 0. \label{concatenation}
\end{equation}
Furthermore, for every function $X$ on $\Omega$, define the function $X^{\tau,\omega}$ on $\Omega$ by
\begin{equation}
	X^{\tau,\omega}(\red{\bar{\omega}}):=X(\omega \otimes_{\tau} \red{\bar{\omega}}), \quad \red{\bar{\omega}} \in \Omega. \label{concaRV}
\end{equation}
Given a probability measure $P \in \mathcal{P}(\Omega)$ and the regular conditional probability distribution $P^{\omega}_{\tau}$ of $P$ given $\mathcal{F}_{\tau}$, introduce the probability measure $P^{\tau, \omega} \in \mathcal{P}(\Omega)$ by
\begin{equation}
	P^{\tau, \omega}(A):=P_{\tau}^{\omega}(\omega \otimes_{\tau} A), \quad A \in \mathcal{F}, \label{condProbability} 
\end{equation}
with $\omega \otimes_{\tau} A=\lbrace \omega \otimes_{\tau} \red{\bar{\omega}}: \red{\bar{\omega}} \in A \rbrace$. Note that $P^{\omega}_{\tau}$ is concentrated on the paths which coincide with $\omega$ up to time $\tau(\omega)$. 

From now on we fix a family $\mathcal{P} \subseteq \mathcal{P}(\Omega)$ such that the following assumption holds, see \cite{nh_2013}.

\begin{asum} \label{assumptionnutzNew}
	For every finite-valued $\mathbb{F}$-stopping time $\tau$, the family $\mathcal{P}$ satisfies the following conditions:
	\begin{enumerate}
	\itemsep0pt
		\item \emph{Measurability:} The set $\mathcal{P} \in \mathcal{P}(\Omega)$ is analytic.
		\item \emph{Invariance:} $P^{\tau, \omega} \in \mathcal{P}$ for $P$-a.e. $\omega \in \Omega$.
		\item \emph{Stability under Pasting:} For every $\mathcal{F}_{\tau}$-measurable kernel $\kappa: \Omega \to \mathcal{P}(\Omega)$ such that $\kappa(\omega) \in \mathcal{P}$ for $P$-a.e. $\omega \in \Omega$, the following measure 
			\begin{equation}
				\overline{P}(A)=\int \int (\textbf{1}_A)^{\tau,\omega}(\omega')\kappa(d\omega';\omega)P(d\omega), \quad A \in \mathcal{F}, \label{OverlineP}
			\end{equation}
			is an element of $\mathcal{P}$.
	\end{enumerate}
\end{asum}
We say that a property holds $\mathcal{P}$-quasi-surely ($\mathcal{P}$-q.s.) on $\Omega$ if this property holds $P$-a.s. for all $P \in \mathcal{P}$.\\
\red{Before introducing the main result  in \cite[Theorem 2.3]{nh_2013} we briefly recall the definition of an upper semianalytic function. We say that a $\overline{\mathbb{R}}$-valued function $f$ on a Polish space is \emph{upper semianalytic} if the set $\lbrace f > c \rbrace$ is analytic for all $c \in \mathbb{R}$, i.e. it is the image of a Borel set of another Polish space under a Borel-measurable mapping.  }

\begin{prop} \label{SublinearNutz}
	Let Assumption \ref{assumptionnutzNew} hold true, $\sigma \leq \tau$ be finite-valued $\mathbb{F}$-stopping times and $X: \Omega \to \mathbb{R} \cup \{-\infty\} \cup \{+\infty\}$ be an upper semianalytic function on $\Omega$. Then the function $\mathcal{E}_{\tau}(X)$ defined by
		\begin{equation}
			\mathcal{E}_{\tau}(X)(\omega):=\sup_{P \in \mathcal{P}} \mathbb{E}^P[X^{\tau,\omega}], \quad \omega \in \Omega, \label{definitionOperator}
		\end{equation}
		is $\mathcal{F}^*_{\tau}$-measurable and upper semianalytic.
		Moreover
		\begin{equation}
			\mathcal{E}_{\sigma}(X)(\omega)= \mathcal{E}_{\sigma}(\mathcal{E}_{\tau}(X))(\omega) \quad \text{ for all } \omega \in \Omega. \label{towerNutz}
		\end{equation}
	Furthermore, the following consistency condition is fulfilled, i.e.,
		\begin{equation}
			\mathcal{E}_{\tau}(X)=\esssup_{P' \in \mathcal{P}(\tau;P)} \mathbb{E}^{P'}[X| \mathcal{F}_{\tau}] \quad P\text{-a.s. for all } P \in \mathcal{P}, \label{repesssup}
		\end{equation}
	where $\mathcal{P}(\tau;P)=\lbrace P' \in \mathcal{P}: P'=P \text{ on } \mathcal{F}_{\tau}\rbrace$.
\end{prop}
The family of sublinear conditional expectations $(\mathcal{E}_t)_{t \in [0,T]}$ is called $(\mathcal{P}, \mathbb{F})$-conditional expectation. \\

We now enlarge the underlying space $(\Omega, \mathcal{F}, \mathcal{P})$ to introduce $N$ random times $0 <\tau_1<...<\tau_N$, as a generalization under model uncertainty of the construction in Section \ref{sec:coxmodel}. To do so, let $(\hat{\Omega},\mathcal{B}(\hat{\Omega}),\hat{P})$ be another probability space satisfying the first point of Assumption \ref{asum:ExistenceUniformlyRV} as in Section \ref{sec:coxmodel}, and consider the product space $(\tilde{\Omega}, \mathcal{G}):=(\Omega \times \hat{\Omega}, \mathcal{B}(\Omega) \otimes \mathcal{B}(\hat{\Omega})).$
The family of all probability measures on $(\tilde{\Omega}, \mathcal{G})$ is denoted by $\mathcal{P}(\tilde{\Omega})$. In particular we are interested in the family of probability measures $\tilde{\mathcal{P}}$ given by
	\begin{equation}
			\tilde{\mathcal{P}}:=\lbrace \tilde{P} \in \mathcal{P}(\tilde{\Omega}): \tilde{P} = P \otimes \hat{P}, P \in \mathcal{P} \rbrace. \label{probExtendedDependence}
	\end{equation}
In addition, we assume that the second point of Assumption \ref{asum:ExistenceUniformlyRV} holds for each $\tilde{P} \in \tilde{\mathcal{P}}$.
Moreover, on $(\Omega,\mathcal{F},\mathcal{P})$ we consider non-negative $\mathbb{F}$-adapted processes $\tilde{\lambda}^1,...,\tilde{\lambda}^N$ such that for $n=1,...N$ and $0 \leq t < \infty$ \eqref{eq:AssumptionsLambda} holds $\mathcal{P}$-q.s. Furthermore, for each $n=1,...,N$ the random times $\tilde{\tau}_n$ and $\tau_n$ on $\tilde{\Omega}$ are defined as in \eqref{eq:DefiTauFixedProbability} and \eqref{eq:DefiTildeTauFixedProbability} with associated filtrations $\mathbb{H}^n$ and $\mathbb{G}^{(n)}$ as in Definition \ref{def:stoppingtimes}. \begin{remark}
	As a slight generalization to Proposition \ref{prop:properties}, it can be seen that for any $n=0,...,N-1$ the following properties hold:
	\begin{enumerate}
	\item $\tilde{\tau}_n:=\tau_n-\tau_{n-1}$ is $\tilde{\mathcal{P}}$-q.s. independent of $\mathcal{H}^{(n-1)}_t$ given $\mathcal{F}_{\infty}$ for any $t \geq 0$, i.e., for each $\tilde{P} \in \tilde{\mathcal{P}}$, $\tilde{\tau}_n=\tau_n-\tau_{n-1}$ is independent of $\mathcal{H}^{(n-1)}_t$ given $\mathcal{F}_{\infty}$ for any $t \geq 0$ with respect to $\tilde{P}$;
	\item $\tau_{n+1}$ avoids $\mathbb{G}^{(n)}$-stopping times $\tilde{\mathcal{P}}$-q.s.
	\end{enumerate}
\end{remark}
For $n=1,...,N$ we denote by $\mathbb{G}^{(n),*}$ the corresponding universally completed filtration as in \eqref{filtration}. Moreover, let $\mathcal{G}^{\mathcal{P}}:=\mathcal{G} \vee \mathcal{N}_{\infty}^{\mathcal{P}}$ \red{and $\mathcal{G}^{\mathcal{P},(n)}:=\mathcal{G}^{n} \vee \mathcal{N}_{\infty}^{\mathcal{P}}$ for $n=1,...,N$} with $\mathcal{N}_{\infty}^{\mathcal{P}}$ defined\footnote{Note that it is important to consider here the $\sigma$-algebra $\mathcal{G}^{\mathcal{P}}\red{, \mathcal{G}^{\mathcal{P},(n)}}$ and not $\mathcal{G}^{\tilde{\mathcal{P}}}:=\mathcal{G} \vee \mathcal{N}_{\infty}^{\tilde{\mathcal{P}}}\red{, \mathcal{G}^{\tilde{\mathcal{P}},(n)}:=\mathcal{G}^{n} \vee \mathcal{N}_{\infty}^{\tilde{\mathcal{P}}},}$ \red{respectively}. For the same reason, the results in Section \ref{sec:twotimes} are only stated for a $\mathcal{G}^P$-measurable random variable and not one which is measurable with respect to $\mathcal{G}^{\tilde{P}}:=\mathcal{G} \vee \mathcal{N}_{\infty}^{\tilde{P}}$. A detailed explanation for this is given in Remark 2.14 in \cite{bz_2019}.} in \eqref{defiNullsets}.
In addition, we define $L^0(\tilde{\Omega})$ as the space of all $\mathbb{R}$-valued $\mathcal{G}^{\mathcal{P}}$-measurable functions, where we use the following convention. For every $\tilde{P} \in \mathcal{P}(\tilde{\Omega})$, we set $\mathbb{E}^{\tilde{P}}[X]:=\mathbb{E}^{\tilde{P}}[X^+]-\mathbb{E}^{\tilde{P}}[X^-]$ if $\mathbb{E}^{\tilde{P}}[X^+]$ or $\mathbb{E}^{\tilde{P}}[X^-]$ is finite and $\mathbb{E}^{\tilde{P}}[X]:= -\infty$ if $\mathbb{E}^{\tilde{P}}[X^+]=E^{\tilde{P}}[X^-]=+\infty$. For each $\tilde{P} \in \tilde{\mathcal{P}}$, the set $L^1_{\tilde{P}}(\tilde{\Omega})$ is given by \eqref{eq:defL1P}. Furthermore, we introduce the set\red{s}
\begin{align*}
	L^1({\tilde{\Omega}}):=\lbrace \tilde{X} \vert  \tilde{X}: (\tilde{\Omega}, \mathcal{G}^{\mathcal{P}}) \to (\mathbb{R}, \mathcal{B}(\mathbb{R})) \text{ measurable function such that  }
	\tilde{\mathcal{E}}(\vert \tilde{X} \vert ) < \infty \rbrace
\end{align*}
\red{and 
\begin{align*}
	L^{1,(n)}({\tilde{\Omega}}):=\lbrace \tilde{X} \vert  \tilde{X}: (\tilde{\Omega}, \mathcal{G}^{\mathcal{P},(n)}) \to (\mathbb{R}, \mathcal{B}(\mathbb{R})) \text{ measurable function such that  }
	\tilde{\mathcal{E}}(\vert \tilde{X} \vert ) < \infty \rbrace,
\end{align*}
for $n=1,...,N$.}
where $\tilde{\mathcal{E}}$ denotes the upper expectation associated to $\tilde{\mathcal{P}}$ defined as
\begin{equation*}
	\tilde{\mathcal{E}}(\tilde{X}):=\sup_{\tilde{P} \in \tilde{\mathcal{P}}} \mathbb{E}^{\tilde{P}} [\tilde{X}],  \quad \tilde{X} \in L^0(\tilde{\Omega}).
\end{equation*}

We now use Theorem \ref{prop:NTimesModelUncertainty} to define the sublinear conditional operator $\tilde{\mathcal{E}}^N$. 
\begin{defi}\label{def:defiSublinearOperatorMulti}
Let Assumption \ref{assumptionnutzNew} hold for $\mathcal{P}$ and consider an upper semianalytic function $Y$ on $\tilde{\Omega}$ such that $Y \in L^1(\tilde{\Omega})$ or $Y$ is $\mathcal{G}^{\mathcal{P}}$-measurable and non-negative. For $t \geq 0$ we define the following function
\begin{align}
	\tilde{\mathcal{E}}^N_t(Y)&:= \textbf{1}_{\lbrace t < \tau_1\rbrace} \mathcal{E}_t\left( e^{\int_0^{t}\tilde{\lambda}_s^1 ds} \mathbb{E}^{\hat{P}} \left [ \textbf{1}_{\lbrace t < \tau_1\rbrace}  Y \right] \right) \nonumber \\
	&+\sum_{k=1}^{N-1} \textbf{1}_{\lbrace \tau_k \leq t < \tau_{k+1}\rbrace}  \mathcal{E}_t\left ( e^{\int_0^{t-u_k}\tilde{\lambda}_s^{k+1} ds} \mathbb{E}^{\hat{P}}\left[\textbf{1}_{\lbrace\tilde{\tau}_{k+1} >t-u_k\rbrace}  \varphi(\bold{u}_{(k)}, \bold{u}_{k}^{\tilde{\tau}},\cdot) \right]\right)\bigg\vert_{\bold{u}_{(k)}=\boldsymbol{\tau}_{(k)}} \nonumber \\
	&+\textbf{1}_{\lbrace \tau_{N} \leq t \rbrace} \mathcal{E}_t(\varphi(\bold{u}_{(N)},\cdot)) \big \vert_{\bold{u}_{(N)}=\boldsymbol{\tau}_{(N)}}, \label{eq:defiSublinearOperatorMulti}
\end{align}
where $\varphi$ is the measurable function
\begin{equation} \label{eq:FunctionVarphiDefinition1}
	\varphi:(\mathbb{R}_+^N \times \Omega, \mathcal{B}(\mathbb{R}_+^N) \otimes \mathcal{F}_{\infty}^{\mathcal{P}}) \to (\mathbb{R}, \mathcal{B}(\mathbb{R}))
\end{equation}
such that
\begin{equation} \label{eq:FunctionVarphiDefinition2}
	Y(\omega,\hat{\omega}) =  \varphi(\tau_1(\omega, \hat{\omega}), \dots,\tau_{N}(\omega, \hat{\omega}), \omega), \quad (\omega,\hat{\omega}) \in \tilde{\Omega},
\end{equation}
 and $\boldsymbol{u}_k^{\tilde{\tau}}$ is given in Theorem \ref{prop:NTimesModelUncertainty} and $(\mathcal{E}_t)_{t \geq 0}$ denotes the conditional sublinear operator introduced in Proposition \ref{SublinearNutz}.
\end{defi}

\begin{remark} \label{remark:MeasurabilityII}
	Note that Theorem \ref{prop:NTimesModelUncertainty} also holds for $Y$ which is in $L^1(\tilde{\Omega})$ or $Y$ which is $\mathcal{G}^{\mathcal{P}}$-measurable and non-negative. This follows as Lemma \ref{lemma:DecompositionGeneraln} can also be proved for a $\mathcal{G}^{\mathcal{P}}$-measurable random variable, see also Remark 2.14 in \cite{bz_2019}. In this case $\varphi: (\mathbb{R}^N_+ \times \Omega, \mathcal{B}(\mathbb{R}_+^N) \otimes \mathcal{F}_{\infty}^{\mathcal{P}}) \to (\mathbb{R}, \mathcal{B}(\mathbb{R}))$, which we use in \eqref{eq:DefiVarphi1} and \eqref{eq:DefiVarphi2}.
\end{remark}
We now recall Lemma 2.17 in \cite{bz_2019} which lists some properties of upper semianalytic functions.
\begin{lemma} \label{lemma:propertiesUpperSemianalytic}
	Let $X, Y$ be two Polish spaces.
	\begin{enumerate}
		\item If $f:X \to Y$ is a Borel-measurable function and a set $A \subseteq X$ is analytic, then $f(A)$ is analytic. If a set $B \subseteq Y$ is analytic, then $f^{-1}(B)$ is analytic.
		\item If $f_n:X \to \bar{\mathbb{R}}, n \in \mathbb{N},$ is a sequence of upper semianalytic functions and $f_n \to f$, then $f$ is upper semianalytic.
		\item If $f:X \to Y$ is a Borel-measurable function and $g: Y \to \bar{\mathbb{R}}$ is upper semianalytic, then the composition $g \circ f$ is also upper semianalytic. If $f:X \to Y$ is a surjective Borel-measurable function and there is a function $g: Y \to \bar{\mathbb{R}}$ such that $g \circ f$ is upper semianalytic, then $g$ is upper semianalytic.
		\item If $f,g: X \to \bar{\mathbb{R}}$ are two upper semianalytic functions, then $f+g$ is upper semianalytic.
		\item If $f: X \to \bar{\mathbb{R}}$ is an upper semianalytic function, $g: X \to \bar{\mathbb{R}}$ is Borel-measurable such that $g \geq 0$, then the product $f \cdot g$ is upper semianalytic.
		\item If $f: X \times Y \to \bar{\mathbb{R}}$ is upper semianalytic and $\kappa(dy;x)$ is a Borel-measurable stochastic kernel on $Y$ given $X$, then the function $g: X \to \bar{\mathbb{R}}$ defined by
		\begin{equation*}
			g(x)=\int f(x,y) \kappa (dy; x), \quad x \in X,
		\end{equation*}
		is upper semianalytic.
	\end{enumerate}
\end{lemma}

\begin{prop} \label{prop:WellDefinedConsistency}
	Let Assumption \ref{assumptionnutzNew} hold for $\mathcal{P}$ and consider an upper semianalytic function $Y$ on $\tilde{\Omega}$ such that $Y \in L^1(\tilde{\Omega})$ or $Y$ is $\mathcal{G}^{\mathcal{P}}$-measurable and non-negative. Then the sublinear conditional operator $(\tilde{\mathcal{E}}_t^{N}(Y))_{t \geq 0}$  in \eqref{eq:defiSublinearOperatorMulti} is well-defined. Furthermore, $\tilde{\mathcal{E}}^N(Y)$ satisfies the consistency condition 
	\begin{equation}
			\tilde{\mathcal{E}}_t^N (Y)=\esssupT_{\tilde{P}' \in \tilde{\mathcal{P}}(t;\tilde{P})} \mathbb{E}^{\tilde{P}'}\left[Y| \mathcal{G}_{t}^{(N)}\right] \quad \tilde{P}\text{-a.s. for all } \tilde{P} \in \tilde{\mathcal{P}}, \label{eq:ConsistencyCondition}
	\end{equation}	
	for any $t \geq 0$, where $\mathcal{\tilde{P}}(t;\tilde{P})=\lbrace \tilde{P}' \in \tilde{\mathcal{P}}: \tilde{P}'=\tilde{P} \text{ on } \mathcal{G}_{t}^{(N)}\rbrace$.

\end{prop}
\begin{proof}
	By the same arguments as in the proof of Lemma 2.18 in \cite{bz_2019}, $e^{\int_0^{t}\tilde{\lambda}_s^1 ds} \mathbb{E}^{\hat{P}} \left [ \textbf{1}_{\lbrace t < \tau_1\rbrace}  Y \right] $ is an upper {semianalytic} function on $\Omega$ for any $t \ge 0$. 
	
We now prove that 
\begin{equation}
	e^{\int_0^{t-u_k}\tilde{\lambda}_s^{k+1} ds} \mathbb{E}^{\hat{P}}\left[\textbf{1}_{\lbrace\tilde{\tau}_{k+1} >t-u_k\rbrace}  \varphi(\bold{u}_{(k)}, \bold{u}_{k}^{\tilde{\tau}},\cdot) \right] \label{eq:UpperSemianalyticMiddlePart}
\end{equation}
is an upper seminanalytic function on $\Omega$  for $k=1,...,N-1$ and fixed $\boldsymbol{u}_{(k)}$, for every $t \ge 0$.

Note that by Theorem \ref{prop:NTimesModelUncertainty} and Remark \ref{remark:MeasurabilityII}
there exists a function 
\begin{align}
	\tilde{\varphi}: (\mathbb{R}_+^{N}  \times  {\Omega}, \mathcal{B}(\mathbb{R}_+^N)  \otimes \mathcal{F}_{\infty}^{{\mathcal{P}}}) &\to (\mathbb{R}, \mathcal{B}(\mathbb{R})) \label{eq:PresentationTilde}
\end{align} 
such that  
\begin{align*}
	Y(\omega, \hat{\omega})=\varphi(\boldsymbol{\tau}_{(N)}(\omega,\hat{\omega}),\omega)=\tilde{\varphi}(\tilde{\boldsymbol{\tau}}_{(N)}(\omega,\hat{\omega}),\omega), 
\end{align*}
with $\tilde{\boldsymbol{\tau}}_{(N)}=(\tilde{\tau}_1,...,\tilde{\tau}_N)$. Therefore, we can rewrite \eqref{eq:UpperSemianalyticMiddlePart} as 
\begin{equation*}
	e^{\int_0^{t-\sum_{l=1}^k u_l}\tilde{\lambda}_s^{k+1} ds} \mathbb{E}^{\hat{P}}\left[\textbf{1}_{\lbrace\tilde{\tau}_{k+1} >t- \sum_{l=1}^k u_l\rbrace}  \tilde{\varphi}(\tilde{\bold{u}}_{(k)}, \tilde{\boldsymbol{\tau}}_{(k+1:N)},\cdot) \right].
\end{equation*}
where $\tilde{\boldsymbol{\tau}}_{(k+1:N)}:=(\tilde{\tau}_{k+1},...,\tilde{\tau}_N)$. 

From now on we fix $t \ge 0$, $k=1,...,N-1$ and  $\tilde{\bold{u}}_{(k)}$.  By {point 5 of Lemma \ref{lemma:propertiesUpperSemianalytic}} it is enough to prove that $\mathbb{E}^{\hat{P}}\left[\textbf{1}_{\lbrace\tilde{\tau}_{k+1} >t- \sum_{l=1}^k u_l\rbrace}  \tilde{\varphi}(\tilde{\bold{u}}_{(k)}, \tilde{\boldsymbol{\tau}}_{(k+1:N)},\cdot) \right]$ is an upper semianalytic function on $\Omega$, as $e^{\int_0^{t-u_k}\tilde{\lambda}_s^{k+1} ds} $ is a non-negative Borel-measurable function. 

Next, we show that $Z:=\textbf{1}_{\lbrace\tilde{\tau}_{k+1} >t- \sum_{l=1}^k u_l\rbrace}  \tilde{\varphi}(\tilde{\bold{u}}_{(k)}, \tilde{\boldsymbol{\tau}}_{(k+1:N)},\cdot)$ is an upper semianalytic function on $\Omega \times \hat{\Omega}$. {Point 6 of Lemma \ref{lemma:propertiesUpperSemianalytic}} will then imply that \eqref{eq:UpperSemianalyticMiddlePart} is an upper semianalytic function on $\Omega$. As $\textbf{1}_{\lbrace\tilde{\tau}_{k+1} >t- \sum_{l=1}^k u_l\rbrace} $ is Borel-measurable on $\Omega \times \hat{\Omega}$ and non-negative, it can be seen that $Z$ is upper semianalytic applying again {point 5 of Lemma \ref{lemma:propertiesUpperSemianalytic}.} Therefore, it remains to show that for fixed $k=1,...,N-1$ and $\tilde{\bold{u}}_{(k)}$, the function 
\begin{equation*}
	(\omega, \hat{\omega})  \mapsto \tilde{\varphi}(\tilde{u}_1,...,\tilde{u}_k, \tilde{\tau}_{k+1}(\omega, \hat{\omega}),...,\tilde{\tau}_N(\omega, \hat{\omega}), \omega)
\end{equation*}
is upper semianalytic. We first prove that the function $g:\mathbb{R}_+^N \times \Omega \to \mathbb{R}$ given by
\begin{equation*}
	g(\tilde{\bold{u}}_{(N)},\omega)=\tilde{\varphi}(\tilde{\bold{u}}_{(N)},\omega)
\end{equation*}
is upper semianalytic.

To do so, define the function $f: \Omega \times \hat{\Omega} \to \mathbb{R}_+^N \times \Omega$ with $f(\omega, \hat{\omega})=(\tilde{\tau}_1(\omega,\hat{\omega}),...,\tilde{\tau}_N(\omega,\hat{\omega}), \omega)$. {Then $f$} is surjective by the definition of $(\tilde{\tau}_n)_{n=1}^N$.
Note that $Y(\omega, \hat{\omega})= \tilde\varphi \circ f(\omega, \hat{\omega})= g \circ f(\omega, \hat{\omega}) $ is an upper semianalytic function on $\Omega \times \hat{\Omega}$, which implies by the second part of {point 3 of Lemma \ref{lemma:propertiesUpperSemianalytic}} that the function $g$ is upper semianalytic on $\mathbb{R}_+^N \times\Omega $. \\
We now define the function $\tilde{f}: \mathbb{R}_+^k \times \Omega \times \hat{\Omega} \to \mathbb{R}_+^N \times \Omega$ by 
$$
\tilde{f}(\tilde{\boldsymbol{u}}_{(k)},\omega, \hat{\omega})=(\tilde{\boldsymbol{u}}_{(k)},\tilde{\boldsymbol{\tau}}_{(k+1:N)}(\omega, \hat{\omega}), \omega),
$$
which is a Borel-measurable function on $\mathbb{R}_+^k \times \Omega \times \hat{\Omega}$. Then it holds 
$$
h(\tilde{\boldsymbol{u}}_{(k)},\omega, \hat{\omega}):=g \circ \tilde{f} (\tilde{\boldsymbol{u}}_{(k)},\omega, \hat{\omega})=\tilde{\varphi}( \tilde{\boldsymbol{u}}_{(k)},\tilde{\boldsymbol{\tau}}_{(k+1:N)}(\omega, \hat{\omega}), \omega),
$$ which is again upper semianalytic on $\mathbb{R}^k_+ \times \Omega \times \hat{\Omega}$ by the first part of {point 3 of Lemma \ref{lemma:propertiesUpperSemianalytic}}. 

For fixed $\tilde{\boldsymbol{u}}_{(k)}$ define now $\Psi_{\tilde{\bold{u}}_{(k)}}:\Omega \times \hat{\Omega} \to \mathbb{R}^k_+ \times \Omega \times \hat{\Omega}$ by $\Psi_{\tilde{\bold{u}}_{(k)}}(\omega, \hat{\omega})=(\tilde{\bold{u}}_{(k)},\omega, \hat{\omega})$, which is Borel-measurable. Then we have 
$$
\tilde{\varphi}(\tilde{\boldsymbol{u}}_{(k)},\tilde{\boldsymbol{\tau}}_{(k+1:N)}(\omega, \hat{\omega}), \omega) = h \circ \Psi_{\tilde{\bold{u}}_{(k)}}(\omega, \hat{\omega})=h(\tilde{\bold{u}}_{(k)},\omega, \hat{\omega}),
$$
which is upper semianalytic on $\Omega \times \hat{\Omega}$ by the first part of {point 3 of Lemma \ref{lemma:propertiesUpperSemianalytic}}.
By using similar arguments and the fact that $g$ is upper semianalytic, we can also show that $\tilde{\varphi}(\tilde{\bold{u}}_{(N)},\omega)$ is upper semianalytic on $\Omega$. \\
The consistency condition in \eqref{eq:ConsistencyCondition} follows by the same arguments as Theorem 2.18 in \cite{bz_2019}. This is possible as ${\lbrace \tau_1 > t\rbrace}$,$(\lbrace \tau_k \leq t < \tau_{k+1}\rbrace)_{k=1}^{N-1}$, $\lbrace \tau_N \leq t \rbrace$ are disjoint events and 
\begin{equation*}
	\mathcal{\tilde{P}}(t;\tilde{P})=\lbrace \tilde{P}' \in \tilde{\mathcal{P}}: {P}' \otimes \hat{P}={P} \otimes \hat{P} \text{ on } \mathcal{G}_{t}^{(N)}\rbrace= \lbrace \tilde{P}' \in \tilde{\mathcal{P}}: {P}'={P} \text{ on } \mathcal{F}_{t}\rbrace.
\end{equation*}
\end{proof}

\begin{prop} \label{prop:EvaluatedOperatorUpperSemianalytic}
 Let $t \geq 0$ and $Y$ satisfy the assumptions in Proposition \ref{prop:WellDefinedConsistency}. The function $\tilde{\mathcal{E}}_t^N(Y)$ defined in \eqref{eq:defiSublinearOperatorMulti} is upper semianalytic and measurable with respect to $\mathcal{G}_t^{(N),*}$ and $\mathcal{G}^{\mathcal{P}}$. 
\end{prop}
\begin{proof}
	This follows by the same arguments as Proposition 2.21 in \cite{bz_2019}.
\end{proof}
For $m=1,...,\red{N}$ and $Y$ \red{which is upper semianalytic on $\tilde{\Omega}$ such that $Y \in L^{1,(m)}(\tilde{\Omega})$ or $Y$ is $\mathcal{G}^{\mathcal{P},(m)}$-measurable and non-negative}, we denote by $\tilde{\mathcal{E}}^m$ the following operator
\begin{align}
	\tilde{\mathcal{E}}^m_t(Y)&:=\textbf{1}_{\lbrace t < \tau_1\rbrace} \mathcal{E}_t\left( e^{\int_0^{t}\tilde{\lambda}_s^1 ds} \mathbb{E}^{\hat{P}} \left [ \textbf{1}_{\lbrace t < \tau_1\rbrace}  Y \right] \right) \nonumber \\
	&\ +\sum_{k=1}^{m-1} \textbf{1}_{\lbrace \tau_k \leq t < \tau_{k+1}\rbrace}  \mathcal{E}_t\left ( e^{\int_0^{t-u_k}\tilde{\lambda}_s^{k+1} ds} \mathbb{E}^{\hat{P}}\left[\textbf{1}_{\lbrace\tilde{\tau}_{k+1} >t-u_k\rbrace}  \varphi(\textbf{u}_{(k)},u_{k,{k+1}}^{\tilde{\tau}},...,u_{k,m}^{\tilde{\tau}},\cdot)\right]\right)\bigg\vert_{\textbf{u}_{(k)}=\bold{\tau}_{(k)}} \nonumber \\
	&\ + \textbf{1}_{\lbrace \tau_{m} \leq t \rbrace} \mathcal{E}_t(\varphi(\bold{u}_{(m)},\cdot))\big \vert_{\bold{u}_{(m)}=\bold{\tau}_{(m)}}. \label{eq:DefinitionOperatorm}
\end{align} 
\begin{prop} \label{prop:OperatorCoincides}
Let $Y$ satisfy the assumptions in Proposition \ref{prop:WellDefinedConsistency}. Then:
\begin{enumerate}
	\item If $Y$ is $\mathcal{F}_{\infty}^{{\mathcal{P}}}$-measurable, then $\tilde{\mathcal{E}}_t^N(Y)$ coincides with $\mathcal{E}_t(Y)$ for $t \geq 0$.
	\item If $Y$ is $\mathcal{G}_{\infty}^{{{\mathcal{P}}},\red{(m)}}$-measurable, then $\tilde{\mathcal{E}}_t^{N}(Y)$ coincides with {$\tilde{\mathcal{E}}^m_t(Y)$} for $t \geq 0$ and $m=1,...,N$. For $m=1$, the operator $\tilde{\mathcal{E}}^1$ is the one defined in Theorem 2.18 in \cite{bz_2019}. 
	\end{enumerate}
	\end{prop}

	\begin{proof}
1. If $Y$ is $\mathcal{F}_{\infty}^{{\mathcal{P}}}$-measurable, then $Y(\tilde{\omega})=Y(\omega, \hat{\omega})=Y(\omega)$. Therefore, by Definition \ref{def:defiSublinearOperatorMulti} for any $t \ge 0$ we get
	\allowdisplaybreaks{
	\red{\begin{align*}
	\tilde{\mathcal{E}}^N_t(Y)&= \textbf{1}_{\lbrace t < \tau_1\rbrace} \mathcal{E}_t\left( e^{\int_0^{t}\tilde{\lambda}_s^1 ds} \mathbb{E}^{\hat{P}} \left [ \textbf{1}_{\lbrace t < \tau_1\rbrace}  Y\right] \right) \nonumber \\
	&\ +\sum_{k=1}^{N-1} \textbf{1}_{\lbrace \tau_k \leq t < \tau_{k+1}\rbrace}  \mathcal{E}_t\left ( e^{\int_0^{t-u_k}\tilde{\lambda}_s^{k+1} ds} \mathbb{E}^{\hat{P}}\left[\textbf{1}_{\lbrace\tilde{\tau}_{k+1} >t-u_k\rbrace}  Y\right]\right)\bigg\vert_{{u}_{k}={\tau}_{k}} \nonumber \\
	&\ +\textbf{1}_{\lbrace \tau_{N} \leq t \rbrace} \mathcal{E}_t(Y) \\
	&= \textbf{1}_{\lbrace t < \tau_1\rbrace} \mathcal{E}_t\left( e^{\int_0^{t}\tilde{\lambda}_s^1 ds} e^{-\int_0^{t}\tilde{\lambda}_s^1 ds} Y \right) \nonumber \\
	&\ +\sum_{k=1}^{N-1} \textbf{1}_{\lbrace \tau_k \leq t < \tau_{k+1}\rbrace}  \mathcal{E}_t\left ( e^{\int_0^{t-u_k}\tilde{\lambda}_s^{k+1} ds} e^{-\int_0^{t-u_k}\tilde{\lambda}_s^{k+1} ds}  Y\right)\bigg\vert_{{u}_{k}={\tau}_{k}} \nonumber \\
	&\ +\textbf{1}_{\lbrace \tau_{N} \leq t \rbrace} \mathcal{E}_t(Y) \\
	&= \textbf{1}_{\lbrace t < \tau_1\rbrace} \mathcal{E}_t\left(Y  \right) +\sum_{k=1}^{N-1} \textbf{1}_{\lbrace \tau_k \leq t < \tau_{k+1}\rbrace}  \mathcal{E}_t\left (   Y\right) +\textbf{1}_{\lbrace \tau_{N} \leq t \rbrace} \mathcal{E}_t(Y)\\
	&=\mathcal{E}_t(Y).
	\end{align*}}}
2. Let $m=1,...,N$. If $Y$ is $\mathcal{G}_{\infty}^{\mathcal{P},\red{(m)}}$-measurable, we can write  
	
	\begingroup
\allowdisplaybreaks
	$$
	Y(\omega, \hat{\omega})=\varphi(\tau_{1}(\omega, \hat{\omega}),...,\tau_m(\omega, \hat{\omega}), \omega)
	$$ 
	with $\varphi: (\mathbb{R}_+^m \times \Omega, \mathcal{B}(\mathbb{R}_+^m)\otimes \mathcal{F}_{\infty}^{{\mathcal{P}}}) \to (\mathbb{R},\mathcal{B}(\mathbb{R}))$.  Then for any $t \ge 0$ it follows
	\red{\small{
	\begin{align*}
	\tilde{\mathcal{E}}^N_t(Y)&= \textbf{1}_{\lbrace t < \tau_1\rbrace} \mathcal{E}_t\left( e^{\int_0^{t}\tilde{\lambda}_s^1 ds} \mathbb{E}^{\hat{P}} \left [ \textbf{1}_{\lbrace t < \tau_1\rbrace}  Y \right] \right) \nonumber \\
	&\ +\sum_{k=1}^{m-1} \textbf{1}_{\lbrace \tau_k \leq t < \tau_{k+1}\rbrace}  \mathcal{E}_t\left ( e^{\int_0^{t-u_k}\tilde{\lambda}_s^{k+1} ds} \mathbb{E}^{\hat{P}}\left[\textbf{1}_{\lbrace\tilde{\tau}_{k+1} >t-u_k\rbrace}  \varphi(\textbf{u}_{(k)},u_{k,{k+1}}^{\tilde{\tau}},...,u_{k,m}^{\tilde{\tau}},\cdot)\right]\right)\bigg\vert_{\textbf{u}_{(k)}=\boldsymbol{\tau}_{(k)}} \nonumber \\
	&\ +\sum_{k=m}^{N-1} \textbf{1}_{\lbrace \tau_k \leq t < \tau_{k+1}\rbrace}  \mathcal{E}_t\left ( e^{\int_0^{t-u_k}\tilde{\lambda}_s^{k+1} ds} \mathbb{E}^{\hat{P}}\left[\textbf{1}_{\lbrace\tilde{\tau}_{k+1} >t-u_k\rbrace}  \varphi(\bold{u}_{(m)},\cdot)\right]\right)\bigg\vert_{\bold{u}_{(m)}=\boldsymbol{\tau}_{(m)}} \nonumber \\
	&\ +\textbf{1}_{\lbrace \tau_{N} \leq t \rbrace} \mathcal{E}_t(\varphi(\bold{u}_{(m)},\cdot))\big \vert_{\bold{u}_{(m)}=\boldsymbol{\tau}_{(m)}} \\
	&= \textbf{1}_{\lbrace t < \tau_1\rbrace} \mathcal{E}_t\left( e^{\int_0^{t}\tilde{\lambda}_s^1 ds} \mathbb{E}^{\hat{P}} \left [ \textbf{1}_{\lbrace t < \tau_1\rbrace}  Y \right] \right) \\
	&\ +\sum_{k=1}^{{m-1}} \textbf{1}_{\lbrace \tau_k \leq t < \tau_{k+1}\rbrace}  \mathcal{E}_t\left ( e^{\int_0^{t-u_k}\tilde{\lambda}_s^{k+1} ds} \mathbb{E}^{\hat{P}}\left[\textbf{1}_{\lbrace\tilde{\tau}_{k+1} >t-u_k\rbrace}  \varphi(\textbf{u}_{(k)},u_{k,{k+1}}^{\tilde{\tau}},...,u_{k,m}^{\tilde{\tau}},\cdot)\right]\right)\bigg\vert_{\textbf{u}_{(k)}=\boldsymbol{\tau}_{(k)}} \nonumber \\
	&\ +\sum_{k={m}}^{{N-1}} \textbf{1}_{\lbrace \tau_k \leq t < \tau_{k+1}\rbrace}  \mathcal{E}_t\left ( \varphi(\bold{u}_{(m)},\cdot)\right)\bigg\vert_{\bold{u}_{(m)}=\boldsymbol{\tau}_{(m)}} +\textbf{1}_{\lbrace \tau_{N} \leq t \rbrace} \mathcal{E}_t(\varphi(\bold{u}_{(m)},\cdot))\big \vert_{\bold{u}_{(m)}=\boldsymbol{\tau}_{(m)}} \\
	&=\textbf{1}_{\lbrace t < \tau_1\rbrace} \mathcal{E}_t\left( e^{\int_0^{t}\tilde{\lambda}_s^1 ds} \mathbb{E}^{\hat{P}} \left [ \textbf{1}_{\lbrace t < \tau_1\rbrace}  Y \right] \right) \\
	&\ +\sum_{k=1}^{{m-1}} \textbf{1}_{\lbrace \tau_k \leq t < \tau_{k+1}\rbrace}  \mathcal{E}_t\left ( e^{\int_0^{t-u_k}\tilde{\lambda}_s^{k+1} ds} \mathbb{E}^{\hat{P}}\left[\textbf{1}_{\lbrace\tilde{\tau}_{k+1} >t-u_k\rbrace}  \varphi(\textbf{u}_{(k)},u_{k,{k+1}}^{\tilde{\tau}},...,u_{k,m}^{\tilde{\tau}},\cdot)\right]\right)\bigg\vert_{\textbf{u}_{(k)}=\boldsymbol{\tau}_{(k)}} \nonumber \\
	&\ + \textbf{1}_{\lbrace \tau_{m} \leq t \rbrace} \mathcal{E}_t(\varphi(\bold{u}_{(m)},\cdot))\big \vert_{\bold{u}_{(m)}=\boldsymbol{\tau}_{(m)}} \\
	& ={\tilde{\mathcal{E}}^m_t(Y)}.
	\end{align*}}}  
	\endgroup
\end{proof}

We now prove some consistency properties.

\begin{prop}	\label{prop:ConsistencyAssumptions}
Let $t \geq 0$ and $Y$ satisfy the assumptions in Proposition \ref{prop:WellDefinedConsistency}. Then the following holds:
\begin{enumerate}
	\item If $X$ is an upper semianalytic function on $\tilde{\Omega}$ such that $X \in L^1(\tilde{\Omega})$ and 
	\begin{equation*}
		\esssupT_{\tilde{P}' \in \tilde{\mathcal{P}}(t;\tilde{P})} \mathbb{E}^{\tilde{P}'}\left[Y| \mathcal{G}_{t}^{(N)}\right] = \esssupT_{\tilde{P}' \in \tilde{\mathcal{P}}(t;\tilde{P})} \mathbb{E}^{\tilde{P}'}\left[X| \mathcal{G}_{t}^{(N)}\right]\quad \tilde{P}\text{-a.s. for all } \tilde{P} \in \tilde{\mathcal{P}},
	\end{equation*}
	then $\tilde{\mathcal{E}}_t^N(X)=\tilde{\mathcal{E}}_t^N(Y)$ $\tilde{P}$-a.s. for all $\tilde{P} \in \tilde{\mathcal{P}}$.
	\item If $A \in \mathcal{G}_t^{(N)}$, then $\tilde{\mathcal{E}}_t^{N}(\textbf{1}_A Y)= \textbf{1}_{A} \tilde{\mathcal{E}}_t^{N}( Y)$.
	\item The following pathwise equalities hold:
	\begin{align}
		\tilde{\mathcal{E}}_t^N(\textbf{1}_{\lbrace{ \tau_1 > t \rbrace}}X)&=\textbf{1}_{\lbrace{ \tau_1 >  t \rbrace}} \tilde{\mathcal{E}}_t^N(X), \nonumber \\
		\tilde{\mathcal{E}}_t^N(\textbf{1}_{\lbrace{ \tau_k \leq t < \tau_{k+1}  \rbrace}}X)&=\textbf{1}_{\lbrace{ \tau_k \leq t < \tau_{k+1}  \rbrace}}\tilde{\mathcal{E}}_t^N(X), \quad k=1,...,N-1, \label{eq:PathwiseEqualities}\\
		\tilde{\mathcal{E}}_t^N(\textbf{1}_{\lbrace{ \tau_N \leq t \rbrace}}X)&=\textbf{1}_{\lbrace{ \tau_N \leq  t \rbrace}} \tilde{\mathcal{E}}_t^N(X). \nonumber 
	\end{align}
\end{enumerate}

\end{prop}
\begin{proof}
1. The first property follows directly by the representation in \eqref{eq:ConsistencyCondition}.

2. From now on, we fix a given time $t \ge 0$ and an event $A \in \mathcal{G}_t^{(N)}$. By \eqref{eq:defiSublinearOperatorMulti} we have

\begin{align}
	\tilde{\mathcal{E}}^N_t( \textbf{1}_A Y)&= \textbf{1}_{\lbrace \tau_1 >t \rbrace} \mathcal{E}_t\left( e^{\int_0^t \tilde{\lambda}_v^1 dv }\mathbb{E}^{\hat{P}}[\textbf{1}_{\lbrace t < \tau_1 \rbrace}\textbf{1}_A Y] \right)\label{eq:prop1Aout_1} \\
	&+ \sum_{k=1}^{N-1}\textbf{1}_{\lbrace \tau_k \leq t < \tau_{k+1} \rbrace}\mathcal{E}_t\left(e^{\int_0^{t-u_k} \tilde{\lambda}_v^{k+1} dv} \mathbb{E}^{\hat{P}} [\textbf{1}_{\lbrace \tilde{\tau}_{k+1} >t-u_k \rbrace}  {\varphi}(\boldsymbol{u}_{(k)},\boldsymbol{u}_{k}^{\tilde{\tau}}, \cdot) {\bar{\varphi}}(\boldsymbol{u}_{(k)},\boldsymbol{u}_{k}^{\tilde{\tau}}, \cdot)]\right) \bigg \vert_{\boldsymbol{u}_{(k)}=\boldsymbol{\tau}_{(k)}} \label{eq:prop1Aout_2} \\
	&+\textbf{1}_{\lbrace \tau_N \leq t \rbrace} \mathcal{E}_t\left({\varphi}(\boldsymbol{u}_{(N)},\cdot) {\bar{\varphi}}(\boldsymbol{u}_{(N)},\cdot)\right) \vert_{\boldsymbol{u}_{(N)}=\boldsymbol{\tau}_{(N)}}, \label{eq:prop1Aout}
\end{align}
where 
$$
\varphi:(\mathbb{R}_+^N  \times \Omega, \mathcal{B}(\mathbb{R}_+^N) \otimes \mathcal{F}_{\infty}^{{\mathcal{P}}} ) \to (\mathbb{R},  \mathcal{B}(\mathbb{R}))
$$ 
and
$$
\overline{\varphi}:(\mathbb{R}_+^N \times \Omega, \mathcal{B}(\mathbb{R}_+^N)  \otimes \mathcal{F}_{t} ) \to (\mathbb{R},  \mathcal{B}(\mathbb{R}))
$$
such that
$$
Y(\omega, \hat{\omega})=\varphi(\boldsymbol{\tau}_{(N)}(\omega, \hat{\omega}), \omega)
$$ 
and 
$$
\textbf{1}_A(\omega, \hat{\omega})=\overline{\varphi}(\boldsymbol{\tau}_{(N)}(\omega, \hat{\omega}), \omega).
$$
Fix now $\boldsymbol{u}_{(k)} \in \mathbb{R}_+^k$ and $u_{k,l}^{\tilde{\tau}}:=u_k+ \sum_{m=k+1}^l \tilde{\tau}_m$ for $l=k,...,N$ with $u_{k,k}^{\tilde{\tau}}:=u_k$. For every $l=k,...N$, we let $\overline{\mathbb{H}}^{k,l}=(\overline{\mathcal{H}}^{k,l}_t)_{t \geq 0}$ be the filtration generated by the random variable $u_{k,l}^{\tilde{\tau}}$. Note here that $\overline{\mathbb{H}}^{k,k}$ is the trivial filtration. Moreover, we introduce the filtration $\overline{\mathbb{H}}^{k,(l)}=(\overline{\mathcal{H}}^{k,(l)}_t)_{t \geq 0}$ as given by $\overline{\mathcal{H}}^{k,(l)}_t:= \bigvee_{j=k+1}^{l} \mathcal{H}^{k,j}_t$ for $t \geq 0$ and $l=k+1,...,N$. We now show that there exists $\tilde{A}^{k,N} \in \mathcal{F}_t \vee \overline{\mathcal{H}}^{k,(N)}_t$ such that
\begin{equation} \label{eq:RepresentationGeneralBarPhi}
	\overline{\varphi}(\boldsymbol{u}_{(k)},\boldsymbol{u}_k^{\tilde{\tau}}, \cdot)= \textbf{1}_{\tilde{A}^{k,N}}.
\end{equation}
Since $A \in \mathcal{G}^{(N)}_t$, there exists a measurable function $f$, constants $(s_n)_{n =1,...,N}$ such that $s_n \leq t$ and an $\mathcal{F}_t$-measurable random variable $\mu_t$ such that
\begin{equation*}
	\textbf{1}_A= f(\mu_t, \textbf{1}_{\lbrace \tau_1 \leq s_1\rbrace},..., \textbf{1}_{\lbrace \tau_N \leq s_N\rbrace}). 
\end{equation*}
Thus, as $\textbf{1}_A(\omega, \hat{\omega})=\overline{\varphi}(\boldsymbol{\tau}_{(N)}(\omega, \hat{\omega}), \omega)$ for any $(\omega, \hat{\omega}) \in \Omega \times \hat{\Omega}$, we have
\begin{align*}
	&\overline{\varphi}(\boldsymbol{\tau}_{(N)}(\omega, \hat{\omega}), \omega)
	\\&=f(\mu_t(\omega), \textbf{1}_{\lbrace\tau_1(\omega, \hat{\omega}) \leq s_1\rbrace},..., \textbf{1}_{\lbrace \tau_N(\omega, \hat{\omega}) \leq s_N\rbrace}) \\
	&=f\big(\mu_t(\omega), \textbf{1}_{\lbrace\tau_1(\omega, \hat{\omega}) \leq s_1\rbrace},...,\textbf{1}_{\lbrace\tau_k(\omega, \hat{\omega}) \leq s_k\rbrace},  \textbf{1}_{\lbrace \tau_k(\omega, \hat{\omega}) + \tilde{\tau}_{k+1}(\omega, \hat{\omega}) \leq s_{k+1}\rbrace},...,\textbf{1}_{\lbrace \tau_k(\omega, \hat{\omega})+ \sum_{j=k+1}^N\tilde{\tau}_j(\omega, \hat{\omega}) \leq s_N\rbrace}\big)
\end{align*}

for any $(\omega, \hat{\omega}) \in \Omega \times \hat{\Omega}$. Then it holds
\begin{align}
	\overline{\varphi}(\boldsymbol{u}_{(k)}, \boldsymbol{u}_{k}^{\tilde{\tau}}(\omega, \hat{\omega}), \omega)&=f\big(\mu_t(\omega), \textbf{1}_{\lbrace u_1 \leq s_1\rbrace},...,\textbf{1}_{\lbrace u_k(\omega, \hat{\omega}) \leq s_k\rbrace},  \textbf{1}_{\lbrace u_{k,k+1}^{\tilde{\tau}}(\omega, \hat{\omega}) \leq s_{k+1}\rbrace},...,\textbf{1}_{\lbrace u_{k,N}^{\tilde{\tau}}(\omega, \hat{\omega}) \leq s_N\rbrace}\big)\\
	&= \textbf{1}_{\tilde{A}^{k,N}} \nonumber
\end{align}
where the right-hand side of the first equality is measurable with respect to $\mathcal{F}_t \vee \overline{\mathcal{H}}^{k,(N)}_t$, so that $\tilde{A}^{k,N} \in \mathcal{F}_t \vee \overline{\mathcal{H}}^{k,(N)}_t$.

By Lemma 5.1.1 in \cite{bielecki_rutkowski_2004}, for any $\tilde{A}^{k,l} \in \mathcal{F}_t \vee \overline{\mathcal{H}}^{k,(l)}_t$ there exists $\tilde{A}^{k,l-1} \in \mathcal{F}_t \vee \overline{\mathcal{H}}^{k,(l-1)}_t$ such that 
\begin{align} \label{eq:ResultBielecki}
	 \tilde{A}^{k,l} \cap \lbrace u_{k,l}^{\tilde{\tau}} >t \rbrace= \tilde{A}^{k,l-1} \cap \lbrace u_{k,l}^{\tilde{\tau}} >t \rbrace,
\end{align}
where $l=k+1,...,N$. With the notation introduced above we have that $\tilde{A}^{k,k} \in \mathcal{F}_t$.  
Then we get 
\red{\small{
\allowdisplaybreaks{
\begin{align}
	& \textbf{1}_{\lbrace \tau_k \leq t < \tau_{k+1} \rbrace}\mathcal{E}_t\left(e^{\int_0^{t-u_k}\tilde{\lambda}_v^{k+1} dv} \mathbb{E}^{\hat{P}} \left[{\textbf{1}_{\lbrace \tilde{\tau}_{k+1} >t-u_{k} \rbrace}}  \varphi(\boldsymbol{u}_{(k)},\boldsymbol{u}_k^{\tilde{\tau}}, \cdot)\overline{\varphi}(\boldsymbol{u}_{(k)},\boldsymbol{u}_k^{\tilde{\tau}}, \cdot)\right]\right)\bigg \vert_{\boldsymbol{u}_{(k)}=\boldsymbol{\tau}_{(k)}} \nonumber \\
	&= \textbf{1}_{\lbrace \tau_k \leq t < \tau_{k+1} \rbrace}\mathcal{E}_t\left(e^{\int_0^{t-u_k}\tilde{\lambda}_v^{k+1} dv} \mathbb{E}^{\hat{P}} \left[\prod_{l=k+2}^N\textbf{1}_{\lbrace u_{k,l}^{\tilde{\tau}}>t \rbrace}{\textbf{1}_{\lbrace \tilde{\tau}_{k+1} >t-u_{k} \rbrace}}  \varphi(\boldsymbol{u}_{(k)},\boldsymbol{u}_k^{\tilde{\tau}}, \cdot)\textbf{1}_{ \tilde{A}^{k,N}}\right]\right)\bigg \vert_{\boldsymbol{u}_{(k)}=\boldsymbol{\tau}_{(k)}}  \label{eq:ConsistencyGeneral1}    \\
	&= \textbf{1}_{ \lbrace \tau_k \leq t < \tau_{k+1} \rbrace}\mathcal{E}_t\left(e^{\int_0^{t-u_k}\tilde{\lambda}_v^{k+1} dv} \mathbb{E}^{\hat{P}} \left[{\textbf{1}_{\lbrace \tilde{\tau}_{k+1} >t-u_{k} \rbrace}}  \varphi(\boldsymbol{u}_{(k)},\boldsymbol{u}_k^{\tilde{\tau}}, \cdot)\textbf{1}_{ \tilde{A}^{k,k+1}}\right]\right)\bigg \vert_{\boldsymbol{u}_{(k)}=\boldsymbol{\tau}_{(k)}} \nonumber  \\
	&= \textbf{1}_{ \lbrace \tau_k \leq t\rbrace} \left[\textbf{1}_{ \lbrace t < u_k +\tilde{\tau}_{k+1} \rbrace}\textbf{1}_{ \tilde{A}^{k,k+1}}\mathcal{E}_t\left(e^{\int_0^{t-u_k}\tilde{\lambda}_v^{k+1} dv}  \mathbb{E}^{\hat{P}} \left[{\textbf{1}_{\lbrace \tilde{\tau}_{k+1} >t-u_{k} \rbrace}}  \varphi(\boldsymbol{u}_{(k)},\boldsymbol{u}_k^{\tilde{\tau}}, \cdot)\right]\right)\right]\bigg \vert_{\boldsymbol{u}_{(k)}=\boldsymbol{\tau}_{(k)}} \nonumber  \\
	&= \textbf{1}_{ \lbrace \tau_k \leq t\rbrace} \Bigg[\textbf{1}_{ \lbrace t < u_k +\tilde{\tau}_{k+1} \rbrace}\notag \\ &\qquad \qquad \quad \prod_{l=k+2}^N\textbf{1}_{\lbrace u_{k,l}^{\tilde{\tau}}>t \rbrace}\textbf{1}_{ \tilde{A}^{k,k+1}}\mathcal{E}_t\left(e^{\int_0^{t-u_k}\tilde{\lambda}_v^{k+1} dv}  \mathbb{E}^{\hat{P}} \left[{\textbf{1}_{\lbrace \tilde{\tau}_{k+1} >t-u_{k} \rbrace}}  \varphi(\boldsymbol{u}_{(k)},\boldsymbol{u}_k^{\tilde{\tau}}, \cdot)\right]\right)\Bigg]\bigg \vert_{\boldsymbol{u}_{(k)}=\boldsymbol{\tau}_{(k)}} \nonumber  \\
	&= \textbf{1}_{ \lbrace \tau_k \leq t\rbrace} \left[\textbf{1}_{ \lbrace t < u_k +\tilde{\tau}_{k+1} \rbrace}\textbf{1}_{ \tilde{A}^{k,N}}\mathcal{E}_t\left(e^{\int_0^{t-u_k}\tilde{\lambda}_v^{k+1} dv}  \mathbb{E}^{\hat{P}} \left[{\textbf{1}_{\lbrace \tilde{\tau}_{k+1} >t-u_{k} \rbrace}}  \varphi(\boldsymbol{u}_{(k)},\boldsymbol{u}_k^{\tilde{\tau}}, \cdot)\right]\right)\right]\bigg \vert_{\boldsymbol{u}_{(k)}=\boldsymbol{\tau}_{(k)}} \label{eq:ConsistencyGeneral5}  \\
	&= \textbf{1}_{ \lbrace \tau_k \leq t\rbrace} \left[\textbf{1}_{ \lbrace t < u_k +\tilde{\tau}_{k+1} \rbrace}\overline{\varphi}(\boldsymbol{u}_{(k)},\boldsymbol{u}_k^{\tilde{\tau}}, \cdot)\mathcal{E}_t\left(e^{\int_0^{t-u_k}\tilde{\lambda}_v^{k+1} dv}  \mathbb{E}^{\hat{P}} \left[{\textbf{1}_{\lbrace \tilde{\tau}_{k+1} >t-u_{k} \rbrace}}  \varphi(\boldsymbol{u}_{(k)},\boldsymbol{u}_k^{\tilde{\tau}}, \cdot)\right]\right)\right]\bigg \vert_{\boldsymbol{u}_{(k)}=\boldsymbol{\tau}_{(k)}} \nonumber  \\
	&= \textbf{1}_{ \lbrace \tau_k \leq t < \tau_{k+1}\rbrace} \textbf{1}_A \mathcal{E}_t\left(e^{\int_0^{t-u_k}\tilde{\lambda}_v^{k+1} dv}  \mathbb{E}^{\hat{P}} \left[{\textbf{1}_{\lbrace \tilde{\tau}_{k+1} >t-u_{k} \rbrace}}  \varphi(\boldsymbol{u}_{(k)},\boldsymbol{u}_k^{\tilde{\tau}}, \cdot)\right]\right)\bigg \vert_{\boldsymbol{u}_{(k)}=\boldsymbol{\tau}_{(k)}}.  \label{eq:ConsistencyGeneral6}
\end{align}}
}}

\red{Here, \eqref{eq:ConsistencyGeneral1} is implied by \eqref{eq:RepresentationGeneralBarPhi}, whereas \eqref{eq:ConsistencyGeneral5} follows  by the recurrent use of \eqref{eq:ResultBielecki}}. \\
Note now that all the arguments used in the derivation of  \eqref{eq:ConsistencyGeneral6} also work for $k=0$ and $k=N+1$ with $\tau_0=0$ and $\tau_{N+1}=+\infty$. Hence we have
\begin{equation}
\tilde{\mathcal{E}}_t^{N}(\textbf{1}_A Y)= \textbf{1}_A\tilde{\mathcal{E}}_t^{N}(Y). \label{eq:UsingPathwiseEquality}
\end{equation}

The pathwise equalities in \eqref{eq:PathwiseEqualities} follow directly by \eqref{eq:UsingPathwiseEquality}. 
\end{proof}

\begin{remark}
	Let $Y$ be an upper semianalytic function on $\tilde{\Omega}$ and $k=2,...,N$. If $Y$ is $\mathcal{G}^{\mathcal{P},\red{(k-1)}}$-measurable as well as non-negative or $Y \in L^1(\tilde{\Omega})$ such that for all $t \geq 0$  also $\tilde{\mathcal{E}}^{k-1}_t(Y)$ in $L^1(\tilde{\Omega})$, then it holds
	\begin{equation} \label{eq:Composition}
		\tilde{\mathcal{E}}^{k}_t(\tilde{\mathcal{E}}^{k-1}_t(Y))=\tilde{\mathcal{E}}^{k-1}_t(Y) 
	\end{equation}
	for all $t \geq 0$. First note that the left-hand side in \eqref{eq:Composition} is well-defined due to the assumptions on $Y$ and as $\tilde{\mathcal{E}}^{k-1}_t(Y) $ is again upper semianalytic and $\mathcal{G}^{\mathcal{P}, \red{(k-1)}}$-measurable by Proposition \ref{prop:EvaluatedOperatorUpperSemianalytic}. Then equation \eqref{eq:Composition} follows by a generalization of Proposition \ref{prop:OperatorCoincides}. \\
	In general, it is not possible to consider $\tilde{\mathcal{E}}^{k-2}_t(\tilde{\mathcal{E}}^{k-1}_t(Y))$ because this expression is not well-defined, as $\tilde{\mathcal{E}}^{k-1}(Y)$ is not $\mathcal{G}^{ \mathcal{P},\red{(k-2)}}$-measurable.
\end{remark}

\section{Weak dynamic programming principle} \label{section:DynamicProgramming}
In this section we investigate dynamic programming for the operator $\tilde{\mathcal{E}}^{k}$ for $k=1,...,N$. More precisely, consider
a $\mathcal{G}^{\mathcal{P},\red{(k)}}$-measurable, upper semianalytic and non-negative random variable $Y$ on $\tilde{\Omega}$. Under Assumption \ref{assumptionnutzNew} we prove the dynamic programming principle 
\begin{equation} \label{eq:WeakDynamicProgramming}
		\tilde{\mathcal{E}}^{k}_s(\tilde{\mathcal{E}}_t^{k}(Y)) \geq \tilde{\mathcal{E}}_s^{k}(Y) \quad \tilde{P} \text{-a.s. for all } \tilde{P} \in \tilde{\mathcal{P}},
	\end{equation}
 for $0 \leq s \leq t,$ $k=1,...,N$.
Note that this is a generalization of Theorem 2.22 in \cite{bz_2019}, where \eqref{eq:WeakDynamicProgramming} is showed for $k=1$. \\
Fix $t \ge 0, \red{k=1,...,N}$ and consider an upper semianalytic function $\tilde{X}$ on $\tilde{\Omega}$ such that $\tilde{X} \in L^{1,\red{(k)}}(\tilde{\Omega})$ or $\tilde{X}$ is $\mathcal{G}^{\mathcal{P},\red{(k)}}$-measurable and non-negative. As in \eqref{eq:NotationI}, with a slight notational abuse we introduce the notation 
\begin{equation} \label{eq:NotationII}
\mathcal{E}_t(\tilde X):=\mathcal{E}_t(\tilde{X}(\cdot, \hat{\omega}))(\omega), \quad (\omega, \hat{\omega}) \in \tilde{\Omega}, 
\end{equation}
to denote that we fix $\hat{\omega}$ and compute the operator $\mathcal{E}_t$ on the function $\tilde{X}(\omega, \hat{\omega}), \omega \in \Omega$.
We start with a lemma that is extensively used in the following analysis.
\begin{lemma}\label{lem:francescayinglin}
For any $t \ge 0, \red{k=1,...,N}$ and any $\mathcal{G}^{\mathcal{P}, \red{(k)}}$-measurable, upper semianalytic and non-negative random variable $Z$ on $\tilde{\Omega}$, it holds 
$$
\mathbb{E}^{\hat{P}}[\mathcal{E}_t(Z)]\ge\mathcal{E}_t (\mathbb{E}^{\hat{P}}[Z]). 
$$
\end{lemma}
\begin{proof}
We recall that by \eqref{eq:NotationI} and \eqref{eq:NotationII} we have that 
\begin{align} 
\mathbb{E}^{\hat{P}}[\mathcal{E}_t(Z)]=\int_{\hat \Omega}\mathcal{E}_t(Z(\cdot, \hat \omega))(\omega)d\hat P (\hat\omega), \quad \omega \in \Omega.
\end{align}The proof follows by \eqref{repesssup} and the conditional Fubini-Tonelli Theorem.
\end{proof} 
For brevity reasons, we prove the result for the special case $k=2$. In this case, the left-hand side in \eqref{eq:WeakDynamicProgramming}, which is well-defined by Proposition \ref{prop:EvaluatedOperatorUpperSemianalytic}, can be rewritten as
\begin{align}
	\tilde{\mathcal{E}}^2_s(\tilde{\mathcal{E}}_t^2(Y))&=\textbf{1}_{\lbrace s < \tau_1 \rbrace} \mathcal{E}_s \left(e^{\int_0^s \tilde{\lambda}_v^1 dv } \mathbb{E}^{\hat{P}}\left[\textbf{1}_{\lbrace s < \tau_1 \rbrace} \tilde{\mathcal{E}}_t^{2}(Y)\right]\right) \nonumber \\
	&+\textbf{1}_{\lbrace \tau_1 \leq s < \tau_2 \rbrace} \mathcal{E}_s\left(e^{\int_0^{s-u_1} \tilde{\lambda}_v^2 dv } \mathbb{E}^{\hat{P}}\left[ \textbf{1}_{\lbrace \tilde{\tau}_2 > s-u_1\rbrace} \bar{\varphi}_t(u_1,u_1+\tilde{\tau}_2, \cdot)\right]\right) \bigg\vert_{u_1=\tau_1} \nonumber \\
	&+ \textbf{1}_{\lbrace \tau_2 \leq s\rbrace} \mathcal{E}_s\left(\bar{\varphi}_t(u_1,u_2, \cdot)\right)\vert_{(u_1,u_2)=(\tau_1, \tau_2)} \label{eq:WeakDynamicRewritten}
\end{align}
with $\bar{\varphi}_t:(\mathbb{R}_+^2 \times \Omega, \mathcal{B}(\mathbb{R}_+^2) \otimes \mathcal{F}_{\infty}^{\mathcal{P}} ) \to (\mathbb{R}, \mathcal{B}(\mathbb{R}))$
such that
\begin{equation*}
	\tilde{\mathcal{E}}_t^2(Y)(\omega, \hat{\omega})=\bar{\varphi}_t(\tau_1(\omega, \hat{\omega}), \tau_2(\omega, \hat{\omega}), \omega).
\end{equation*}
 By the definition of $\tilde{\mathcal{E}}_t^2$ in \eqref{eq:defiSublinearOperatorMulti} it follows that for any 
 $(u_1,u_2, \omega) \in \mathbb{R}_+^2 \times \Omega$ it holds
\begin{align}
	\bar{\varphi}_t(u_1,u_2, \omega)&=\textbf{1}_{\lbrace t < u_1 \rbrace} \mathcal{E}_t \left(e^{\int_0^t \tilde{\lambda}_v^1 dv } \mathbb{E}^{\hat{P}}[\textbf{1}_{\lbrace t < \tau_1 \rbrace} Y]\right)(\omega)\notag \\
	&+ \textbf{1}_{\lbrace u_1 \leq t < u_2 \rbrace} \mathcal{E}_t\left(e^{\int_0^{t-u_1} \tilde{\lambda}_v^2 dv } \mathbb{E}^{\hat{P}}[ \textbf{1}_{\lbrace \tilde{\tau}_2 > t-u_1\rbrace} {\varphi}(u_1,u_1+\tilde{\tau}_2, \cdot)]\right)(\omega) \notag \\
	&+\textbf{1}_{\lbrace u_2 \leq t\rbrace} \mathcal{E}_t({\varphi}(u_1,u_2, \cdot))(\omega),\label{eq:defbarvarphi}
\end{align}
where $\varphi$ is defined in \eqref{eq:DefiVarphi1} and \eqref{eq:DefiVarphi2}. 

With this expression in mind, we now proceed to analyze the three terms on the right-hand side of equation \eqref{eq:WeakDynamicRewritten}. We start from the third one.

\begin{prop}\label{prop:dynamicfirst}
For any $t \geq s \ge 0$, it holds
\begin{equation*}
	\textbf{1}_{\lbrace \tau_2 \leq s\rbrace} \mathcal{E}_s(\bar{\varphi}_t(u_1,u_2, \cdot))\vert_{(u_1,u_2)=(\tau_1, \tau_2)} = \textbf{1}_{\lbrace \tau_2 \leq s\rbrace} \mathcal{E}_s({\varphi}(u_1,u_2, \cdot))\vert_{(u_1,u_2)=(\tau_1, \tau_2)},
\end{equation*}
where $\varphi$ is defined in \eqref{eq:DefiVarphi1} and \eqref{eq:DefiVarphi2} and $\bar{\varphi}_t$ given in \eqref{eq:defbarvarphi}.
\end{prop}
\begin{proof}
The proof uses similar arguments as in Theorem 2.22 in \cite{bz_2019}. In particular, we have
\begin{small}
\begin{align*}
	&\textbf{1}_{\lbrace \tau_2 \leq s\rbrace} \mathcal{E}_s\left(\bar{\varphi}_t(u_1,u_2, \cdot)\right)\vert_{(u_1,u_2)=(\tau_1,\tau_2)} \\
	&=\textbf{1}_{\lbrace \tau_2 \leq s\rbrace} \mathcal{E}_s\bigg( \textbf{1}_{\lbrace t < u_1 \rbrace} \mathcal{E}_t \left(e^{\int_0^t \tilde{\lambda}_v^1 dv } \mathbb{E}^{\hat{P}}[\textbf{1}_{\lbrace t < \tau_1 \rbrace} Y]\right) \\
	&\ \qquad \qquad \qquad + \textbf{1}_{\lbrace u_1 \leq t < u_2 \rbrace} \mathcal{E}_t\left(e^{\int_0^{t-u_1} \tilde{\lambda}_v^2 dv } \mathbb{E}^{\hat{P}}[ \textbf{1}_{\lbrace \tilde{\tau}_2 > t-u_1\rbrace} {\varphi}(u_1,u_1+\tilde{\tau}_2, \cdot)]\right) \notag \\
	&\ \qquad \qquad \qquad +\textbf{1}_{\lbrace u_2 \leq t\rbrace} \mathcal{E}_t\left({\varphi}(u_1,u_2, \cdot)\right) \bigg)\bigg\vert_{(u_1,u_2)=(\tau_1,\tau_2)} \\
	&=\textbf{1}_{\lbrace \tau_2 \leq s\rbrace} \bigg[  \textbf{1}_{\lbrace t < u_1 \rbrace} \mathcal{E}_s\left( \mathcal{E}_t \left(e^{\int_0^t \tilde{\lambda}_v^1 dv } \mathbb{E}^{\hat{P}}[\textbf{1}_{\lbrace t < \tau_1 \rbrace} Y]\right)\right)+ \\
	&\ \qquad \qquad \quad + \textbf{1}_{\lbrace u_1 \leq t < u_2 \rbrace} \mathcal{E}_s\left(\mathcal{E}_t\left(e^{\int_0^{t-u_1} \tilde{\lambda}_v^2 dv } \mathbb{E}^{\hat{P}}[ \textbf{1}_{\lbrace \tilde{\tau}_2 > t-u_1\rbrace} {\varphi}(u_1,u_1+\tilde{\tau}_2, \cdot)]\right)\right) \notag \\
	&\  \qquad \qquad \quad+\textbf{1}_{\lbrace u_2 \leq t\rbrace} \mathcal{E}_s\left(\mathcal{E}_t({\varphi}(u_1,u_2, \cdot)) \right)\bigg]\bigg\vert_{(u_1,u_2)=(\tau_1,\tau_2)} \\
	&=\textbf{1}_{\lbrace \tau_2 \leq s\rbrace} \textbf{1}_{\lbrace u_2 \leq t\rbrace} \mathcal{E}_s({\varphi}(u_1,u_2, \cdot) )\vert_{(u_1,u_2)=(\tau_1,\tau_2)} \\
	&=\textbf{1}_{\lbrace \tau_2 \leq s\rbrace} \mathcal{E}_s({\varphi}(u_1,u_2, \cdot) )\big\vert_{(u_1,u_2)=(\tau_1,\tau_2)}.
\end{align*}
\end{small}
Note that the second equality holds since we deal with disjoint sets, whereas the third equality follows from the tower property of $\mathcal{E}_t$ stated in \eqref{towerNutz} and since 
$$
\textbf{1}_{\lbrace \tau_2 \leq s\rbrace} \textbf{1}_{\lbrace t < \tau_1 \rbrace}=0, \ 
\textbf{1}_{\lbrace \tau_2 \leq s\rbrace} \textbf{1}_{\lbrace \tau_1 \leq t < \tau_2 \rbrace}=0
$$ 
for $0 \leq s \leq t$.
\end{proof}

We now consider the first term of the right-hand side of \eqref{eq:WeakDynamicRewritten}. We start with three lemmas. 
\begin{lemma}\label{lem:firstlemmasecondpiece}
Let $Y$ be a $\mathcal{G}^{\mathcal{P}, \red{(2)}}$-measurable, upper seminanalytic and non-negative random variable on $\tilde{\Omega}$. For any $t \ge 0$, it holds
\begin{equation} \notag
	\textbf{1}_{\lbrace \tau_2 \leq t \rbrace} \mathcal{E}_t \left( \varphi(u_1,u_2, \cdot )\right)\vert_{(u_1,u_2)=(\tau_1 , \tau_2 )}=\textbf{1}_{\lbrace \tau_2 \leq t \rbrace} \mathcal{E}_t(Y)	\end{equation}
where $\varphi$ is given in \eqref{eq:DefiVarphi1} and \eqref{eq:DefiVarphi2}. \end{lemma}
\begin{proof}
On the event $\lbrace \tau_2 \leq t \rbrace$ we have
\begin{align*}
	&\tau_{1}(\omega \otimes_t \omega', \hat{\omega})=\tau_1(\omega, \hat{\omega}) \quad \text{ for all } \omega' \in \Omega \\
	&\tau_{2}(\omega \otimes_t \omega', \hat{\omega})=\tau_2(\omega, \hat{\omega}) \quad \text{ for all } \omega' \in \Omega,
\end{align*}
as $\lbrace \tau_2 \leq t \rbrace \subseteq \lbrace \tau_1 \leq t \rbrace$. Therefore, on the event $\lbrace \tau_2 \leq t \rbrace$ we get by \eqref{definitionOperator} and \eqref{eq:defiSublinearOperatorMulti} that
\begin{align*}
	\mathcal{E}_t(Y) =\mathcal{E}_t(Y (\cdot, \hat{\omega}))(\omega)&=\sup_{P \in \mathcal{P}} \int_{\Omega} Y(\omega \otimes_t \omega', \hat{\omega}) P(d\omega')\\
	&=\sup_{P \in \mathcal{P}} \int_{\Omega} \varphi(\tau_1(\omega \otimes_t \omega',\hat{\omega}),\tau_2(\omega \otimes_t \omega',\hat{\omega}),\omega \otimes_t \omega' ) P(d\omega') \\
	&=\sup_{P \in \mathcal{P}} \int_{\Omega} \varphi(\tau_1(\omega ,\hat{\omega}),\tau_2(\omega ,\hat{\omega}),\omega \otimes_t \omega') P(d\omega') \\
	&=\sup_{P \in \mathcal{P}} \int_{\Omega} \varphi(u_1,u_2, \omega \otimes_t \omega') P(d\omega') \big \vert_{(u_1,u_2)=(\tau_1, \tau_2)}\\
	&=\mathcal{E}_t(\varphi(u_1,u_2, \cdot))(\omega)\big \vert_{(u_1,u_2)=(\tau_1(\omega, \hat{\omega}), \tau_2(\omega, \hat{\omega}))},
\end{align*}
for all $\omega \in \Omega$. \end{proof}

\begin{lemma}\label{lem:secondlemmasecondpiece}
For every non-negative measurable function 
$$
\Psi: (\mathbb{R}_+^2 \times \Omega, \mathcal{B}(\mathbb{R}^2) \otimes \mathcal{F}^*_{\infty}) \to (\mathbb{R}, \mathcal{B}(\mathbb{R}))
$$
 we have that
\begin{equation}\notag
	\mathbb{E}^{\hat{P}}[\Psi(\tau_1,\tilde{\tau}_2,\cdot)]=\mathbb{E}^{\hat{P}}\left[ \left(\mathbb{E}^{\hat{P}}[\Psi(u_1,\tilde{\tau}_2,\cdot)]\right)\big \vert_{u_1=\tau_1} \right] 
\end{equation}
and
\begin{equation}
	\mathbb{E}^{\hat{P}}[\Psi(\tau_1,\tilde{\tau}_2,\cdot)]=\mathbb{E}^{\hat{P}}\left[ \left(\mathbb{E}^{\hat{P}}[\Psi(\tau_1,\tilde{u}_2,\cdot)]\right)\big \vert_{\tilde{u}_2=\tilde{\tau}_2} \right]. \label{eq:WeakDynamicProgrammingThird4_a}
\end{equation}
\end{lemma}
\begin{proof}
 Fix first $s_1,s_2 \geq 0$ and $A^* \in \mathcal{F}^*_{\infty}$. Then since $E_1$ and $E_2$ are independent under $\hat{P}$, it holds
\begin{equation*}
\mathbb{E}^{\hat{P}}[\textbf{1}_{\lbrace \tau_1 \leq s_1 \rbrace} \textbf{1}_{\lbrace \tilde{\tau}_2 \leq s_2 \rbrace} \textbf{1}_{A^*}]=\textbf{1}_{A^*}  \mathbb{E}^{\hat{P}}[\textbf{1}_{\lbrace \tau_1 \leq s_1 \rbrace} \textbf{1}_{\lbrace \tilde{\tau}_2 \leq s_2 \rbrace}]=\textbf{1}_{A^*} \mathbb{E}^{\hat{P}}[\textbf{1}_{\lbrace \tau_1 \leq s_1 \rbrace} ] \mathbb{E}^{\hat{P}}[ \textbf{1}_{\lbrace \tilde{\tau}_2 \leq s_2 \rbrace}].
\end{equation*}
 On the other hand, we have
\begin{align*}
	\mathbb{E}^{\hat{P}}\left[ \left(\mathbb{E}^{\hat{P}}[\textbf{1}_{\lbrace u_1 \leq s_1 \rbrace} \textbf{1}_{\lbrace \tilde{\tau}_2 \leq s_2 \rbrace} \textbf{1}_{A^*} ]\right)\big \vert_{u_1=\tau_1} \right]&= \textbf{1}_{A^*}  \mathbb{E}^{\hat{P}}\left[\textbf{1}_{\lbrace u_1 \leq \tau_1 \rbrace} \mathbb{E}^{\hat{P}}[ \textbf{1}_{\lbrace \tilde{\tau}_2 \leq s_2 \rbrace} ] \right] \\
	&=\textbf{1}_{A^*}  \mathbb{E}^{\hat{P}}[\textbf{1}_{\lbrace \tau_1 \leq s_1 \rbrace} ] \mathbb{E}^{\hat{P}}[ \textbf{1}_{\lbrace \tilde{\tau}_2 \leq s_2 \rbrace}]
\end{align*}
and by the same arguments it follows
\begin{align*}
	\mathbb{E}^{\hat{P}}\left[ \left(\mathbb{E}^{\hat{P}}[\textbf{1}_{\lbrace \tau_1 \leq s_1 \rbrace} \textbf{1}_{\lbrace \tilde{u}_2 \leq s_2 \rbrace} \textbf{1}_{A^*} ]\right)\big \vert_{\tilde{u}_2=\tilde{\tau}_2} \right]=\textbf{1}_{A^*}  \mathbb{E}^{\hat{P}}[\textbf{1}_{\lbrace \tau_1 \leq s_1 \rbrace} ] \mathbb{E}^{\hat{P}}[ \textbf{1}_{\lbrace \tilde{\tau}_2 \leq s_2 \rbrace}].
\end{align*}

Then the result follows by a monotone class argument.
\end{proof}

\begin{lemma}\label{lem:thirdlemmasecondpiece}
Let $Y$ be a $\mathcal{G}^{\mathcal{P}, \red{(2)}}$-measurable, upper seminanalytic and non-negative random variable on $\tilde{\Omega}$ and $\tilde{\varphi}$ be the unique non-negative, measurable function $\tilde{\varphi}: (\mathbb{R}_+^2 \times \Omega, \mathcal{B}(\mathbb{R}_+^2) \otimes \mathcal{F}_{\infty}^{\mathcal{P}}) \to (\mathbb{R}, \mathcal{B}(\mathbb{R}))$ such that 
\begin{equation}
			Y(\omega, \hat{\omega})=\tilde{\varphi}(\tau_1(\omega,\hat{\omega}),\tilde{\tau}_2(\omega,\hat{\omega}),\omega), \quad (\omega,\hat{\omega}) \in \Omega \times \hat{\Omega}. \label{eq:DecompositionDifferent}		
\end{equation}
Then for each $0 \leq s \leq t$ it holds 
\begin{equation}
	\mathbb{E}^{\hat{P}}\left[\mathcal{E}_t\left(\mathbb{E}^{\hat{P}}[\textbf{1}_{\lbrace s< u_1 \leq t \rbrace} \tilde{\varphi}(u_1, \tilde{\tau}_2,\cdot)] \right) \bigg \vert_{u_1 = \tau_1} \right] \geq \mathcal{E}_t\left(\mathbb{E}^{\hat{P}}\left[\mathbb{E}^{\hat{P}}[\textbf{1}_{\lbrace s < u_1 \leq t \rbrace} \tilde{\varphi}(u_1, \tilde{\tau}_2,\cdot)] \big \vert_{u_1 = \tau_1}\right] \right) .\label{eq:WeakDynamicProgrammingThird9} 
\end{equation}
\end{lemma}

\begin{proof}
First note that the existence of such a function $\tilde{\varphi}$ follows by Lemma \ref{lemma:DecompositionGeneraln}.  
Let $0 \leq s \leq t $. Then for each $P \in \mathcal{P}$ we have 
\begin{align}
	&\mathbb{E}^{\hat{P}}\left[\mathcal{E}_t\left(\mathbb{E}^{\hat{P}}[\textbf{1}_{\lbrace s< u_1 \leq t \rbrace} \tilde{\varphi}(u_1, \tilde{\tau}_2,\cdot)] \right) \bigg \vert_{u_1 = \tau_1} \right] \notag\\
	&=\mathbb{E}^{\hat{P}}\left[\esssup_{P' \in \mathcal{P}(t;P)}\mathbb{E}^{P'}\left[\mathbb{E}^{\hat{P}}[\textbf{1}_{\lbrace s< u_1 \leq t \rbrace} \tilde{\varphi}(u_1, \tilde{\tau}_2,\cdot)] \bigg\vert \mathcal{F}_t\right] \bigg \vert_{u_1 = \tau_1} \right] \label{eq:PreparationIntermediateStep1} \\
	&\geq \mathbb{E}^{\hat{P}}\left[\mathbb{E}^{P}\left[\mathbb{E}^{\hat{P}}[\textbf{1}_{\lbrace s< u_1 \leq t \rbrace} \tilde{\varphi}(u_1, \tilde{\tau}_2,\cdot)] \bigg\vert \mathcal{F}_t\right] \bigg \vert_{u_1 = \tau_1} \right]. \label{eq:PreparationIntermediateStep2} 
\end{align}
Here, we apply representation \eqref{repesssup} in \eqref{eq:PreparationIntermediateStep1}.
We now prove that for each non-negative and measurable function $\tilde{\varphi}: (\mathbb{R}_+^2 \times \Omega, \mathcal{B}(\mathbb{R}_+^2) \otimes \mathcal{F}_{\infty}^{\mathcal{P}}) \to (\mathbb{R}, \mathcal{B}(\mathbb{R}))$ it holds
\begin{align}
	&\mathbb{E}^{\hat{P}}\left[\mathbb{E}^{P}\left[\mathbb{E}^{\hat{P}}[\textbf{1}_{\lbrace s< u_1 \leq t \rbrace} \tilde{\varphi}(u_1, \tilde{\tau}_2,\cdot)] \bigg\vert \mathcal{F}_t\right] \bigg \vert_{u_1 = \tau_1} \right] \notag \\
	&{=} \mathbb{E}^P\left[\mathbb{E}^{\hat{P}}\left[\mathbb{E}^{\hat{P}}[\textbf{1}_{\lbrace s< u_1 \leq t \rbrace} \tilde{\varphi}(u_1, \tilde{\tau}_2,\cdot)]  \big \vert_{u_1 = \tau_1}\right] \bigg\vert \mathcal{F}_t  \right] \quad P \text{-a.s. for all } P \in \mathcal{P}. \label{eq:IntermediateStepMonotoneClass} 
\end{align}
 
Fix first $s_1, s_2 \geq 0$ and $A \in \mathcal{F}_{\infty}^{{\mathcal{P}}}$. For every $P \in \mathcal{P}$ we have that
\begin{align}
	&\mathbb{E}^{\hat{P}}\left[\mathbb{E}^{P}\left[\mathbb{E}^{\hat{P}}[\textbf{1}_{\lbrace s< u_1 \leq t \rbrace} \textbf{1}_{\lbrace  u_1 \leq s_1 \rbrace} \textbf{1}_{\lbrace  \tilde{\tau}_2 \leq s_2 \rbrace} \textbf{1}_A ]\bigg\vert \mathcal{F}_t\right] \bigg \vert_{u_1 = \tau_1} \right] \notag \\
	&=\mathbb{E}^{\hat{P}}\left[\mathbb{E}^{P}\left[\mathbb{E}^{\hat{P}}[ \textbf{1}_{\lbrace  \tilde{\tau}_2 \leq s_2 \rbrace} \textbf{1}_A ]\big\vert \mathcal{F}_t\right] \textbf{1}_{\lbrace s< \tau_1 \leq t \wedge s_1 \rbrace} \right]\notag \\
	&=\mathbb{E}^{\hat{P}}\left[\textbf{1}_{\lbrace s< \tau_1 \leq t \wedge s_1 \rbrace} \right] \mathbb{E}^{P}\left[\mathbb{E}^{\hat{P}}[ \textbf{1}_{\lbrace  \tilde{\tau}_2 \leq s_2 \rbrace} \textbf{1}_A ]\big\vert \mathcal{F}_t\right] \label{eq:IntermediateMonotoneClass1} \\
	&= \mathbb{E}^{P}\left[\mathbb{E}^{\hat{P}}\left[\textbf{1}_{\lbrace s< \tau_1 \leq t \wedge s_1 \rbrace} \right]\mathbb{E}^{\hat{P}}[ \textbf{1}_{\lbrace  \tilde{\tau}_2 \leq s_2 \rbrace} \textbf{1}_A ]\big\vert \mathcal{F}_t\right] \label{eq:IntermediateMonotoneClass2} \\
	&= \mathbb{E}^{P}\left[\mathbb{E}^{\hat{P}}\left[\textbf{1}_{\lbrace s< \tau_1 \leq t \wedge s_1 \rbrace}\mathbb{E}^{\hat{P}}[ \textbf{1}_{\lbrace  \tilde{\tau}_2 \leq s_2 \rbrace} \textbf{1}_A ] \right]\bigg\vert \mathcal{F}_t\right], \notag
\end{align}
where we use in \eqref{eq:IntermediateMonotoneClass1} that $ \mathbb{E}^{P}\left[\mathbb{E}^{\hat{P}}[ \textbf{1}_{\lbrace  \tilde{\tau}_2 \leq s_2 \rbrace} \textbf{1}_A ]\big\vert \mathcal{F}_t\right]$ is independent of $\hat{\omega}$. Moreover, \eqref{eq:IntermediateMonotoneClass2} follows as $\mathbb{E}^{\hat{P}}\left[\textbf{1}_{\lbrace s< \tau_1 \leq t \wedge s_1 \rbrace} \right]=-e^{-\int_0^{t \wedge s_1} \tilde{\lambda}_v^1 dv} + e^{-\int_0^s \tilde{\lambda}_v^1 dv}$ is $\mathcal{F}_t$-measurable. Then we get \eqref{eq:IntermediateStepMonotoneClass} by using a monotone class argument.
Note now that \eqref{eq:IntermediateStepMonotoneClass} implies 
\begin{align}
	&\mathbb{E}^{\hat{P}}\left[\mathbb{E}^{P}\left[\mathbb{E}^{\hat{P}}[\textbf{1}_{\lbrace s< u_1 \leq t \rbrace} \tilde{\varphi}(u_1, \tilde{\tau}_2,\cdot)] \bigg\vert \mathcal{F}_t\right] \bigg \vert_{u_1 = \tau_1} \right] \notag \\
	&\geq \esssup_{P' \in \mathcal{P}(t;P)}\mathbb{E}^{P'}\left[\mathbb{E}^{\hat{P}}\left[\mathbb{E}^{\hat{P}}[\textbf{1}_{\lbrace s< u_1 \leq t \rbrace} \tilde{\varphi}(u_1, \tilde{\tau}_2,\cdot)]  \big \vert_{u_1 = \tau_1}\right] \bigg\vert \mathcal{F}_t  \right] \notag\\  
	&= \mathcal{E}_t\left(\mathbb{E}^{\hat{P}}\left[\mathbb{E}^{\hat{P}}[\textbf{1}_{\lbrace s< u_1 \leq t \rbrace} \tilde{\varphi}(u_1, \tilde{\tau}_2,\cdot)]  \big \vert_{u_1 = \tau_1}\right] \right)  \label{eq:IntermediateStepMonotoneClass3} 
\end{align}
by using \eqref{repesssup} in \eqref{eq:IntermediateStepMonotoneClass3}. Putting together the inequalities \eqref{eq:PreparationIntermediateStep2} and \eqref{eq:IntermediateStepMonotoneClass3} yields \eqref{eq:WeakDynamicProgrammingThird9}. 

\end{proof}

The next lemma provides a fundamental step in the computations on the second term of the right-hand side of \eqref{eq:WeakDynamicRewritten}. 
\begin{lemma}\label{lem:secondtermsecondpiece}
Let $Y$ be a $\mathcal{G}^{\mathcal{P}, \red{(2)}}$-measurable, upper seminanlytic and non-negative random variable on $\tilde{\Omega}$. For any $0 \le s \le t$ it holds
\small{\begin{align}
	 \mathbb{E}^{\hat{P}}\left[  \textbf{1}_{\lbrace s<\tau_1 \leq t < \tau_2 \rbrace} \mathcal{E}_t\left(e^{\int_0^{t-u_1} \tilde{\lambda}_v^2dv} \mathbb{E}^{\hat{P}} [\textbf{1}_{\lbrace \tilde{\tau}_2 >t-u_1  \rbrace}\varphi(u_1,u_1+\tilde{\tau}_2, \cdot )]\right)\bigg\vert_{u_1=\tau_1} \right] \ge  \mathcal{E}_t\left( \mathbb{E}^{\hat{P}}\left[ \textbf{1}_{\lbrace s < \tau_1 \leq t < \tau_2 \rbrace} Y\right]\right), \label{eq:WeakDynamicProgrammingThird3}
\end{align}}
where $\varphi$ is given in \eqref{eq:DefiVarphi1} and \eqref{eq:DefiVarphi2}.
\end{lemma}
\begin{proof}
Since the operator $\mathcal{E}_t$ is $\mathcal{F}^*_{\infty}$-measurable, and considering a generalized version of Lemma \ref{lemma:DecompositionGeneraln}, where the reference filtration at final time is now $\mathcal{F}^*_{\infty}$, we can apply Lemma \ref{lem:secondlemmasecondpiece} to the function $\Psi:(\mathbb{R}_+^2 \times \Omega, \mathcal{B}(\mathbb{R}^2) \otimes \mathcal{F}_{\infty}^*) \to (\mathbb{R},\mathcal{B}(\mathbb{R}))$ given by
\small{
\begin{align} \label{eq:DefinitionPsiWeakDynamic}
	\Psi(\tau_1, \tilde{\tau}_2, \omega)&:= \textbf{1}_{\lbrace s < \tau_1 \leq t \rbrace}  \textbf{1}_{\lbrace \tilde{\tau}_2>t-\tau_1 \rbrace} \mathcal{E}_t\left(e^{\int_0^{t-u_1} \tilde{\lambda}_v^2dv} \mathbb{E}^{\hat{P}} [\textbf{1}_{\lbrace \tilde{\tau}_2 >t-u_1  \rbrace}\varphi(u_1,u_1+\tilde{\tau}_2, \cdot )]\right)(\omega)  \bigg\vert_{u_1=\tau_1}.
\end{align}}
Then we get
\begingroup
\allowdisplaybreaks
\red{\begin{align}
	 &\mathbb{E}^{\hat{P}}\left[  \textbf{1}_{\lbrace s < \tau_1 \rbrace}  \textbf{1}_{\lbrace \tau_1 \leq t < \tau_2 \rbrace} \mathcal{E}_t\left(e^{\int_0^{t-u_1} \tilde{\lambda}_v^2dv} \mathbb{E}^{\hat{P}} [\textbf{1}_{\lbrace \tilde{\tau}_2 >t-u_1  \rbrace}\varphi(u_1,u_1+\tilde{\tau}_2, \cdot )]\right)\bigg\vert_{u_1=\tau_1} \right] \nonumber \\
	& =\mathbb{E}^{\hat{P}}\left[  \left(  \mathbb{E}^{\hat{P}} \left[  \mathcal{E}_t(\textbf{1}_{\lbrace s < u_1 \leq t \rbrace}  \textbf{1}_{\lbrace \tilde{\tau}_2 >t-u_1 \rbrace}  e^{\int_0^{t-u_1} \tilde{\lambda}_v^2dv} \mathbb{E}^{\hat{P}} [\textbf{1}_{\lbrace \tilde{\tau}_2 >t-u_1  \rbrace}\varphi(u_1,u_1+\tilde{\tau}_2, \cdot )])\right] \right) \bigg\vert_{u_1=\tau_1}\right]\label{eq:WeakDynamicProgrammingThird7*}\\
	& \geq \mathbb{E}^{\hat{P}}\left[  \mathcal{E}_t \left( \mathbb{E}^{\hat{P}}\left[\textbf{1}_{\lbrace s < u_1 \leq t \rbrace}  \textbf{1}_{\lbrace \tilde{\tau}_2 >t-u_1 \rbrace}e^{\int_0^{t-u_1} \tilde{\lambda}_v^2dv} \mathbb{E}^{\hat{P}} [\textbf{1}_{\lbrace \tilde{\tau}_2 >t-u_1  \rbrace}\varphi(u_1,u_1+\tilde{\tau}_2, \cdot )]\right]\right)  \bigg\vert_{u_1=\tau_1}\right] \label{eq:WeakDynamicProgrammingThird7} \\
	&= \mathbb{E}^{\hat{P}}\left[ \left(\textbf{1}_{\lbrace s < u_1 \leq t \rbrace}e^{\int_0^{t-u_1} \tilde{\lambda}_v^2dv}  \mathcal{E}_t \left( \mathbb{E}^{\hat{P}} \left[  \textbf{1}_{\lbrace \tilde{\tau}_2 >t-u_1 \rbrace} \mathbb{E}^{\hat{P}} [\textbf{1}_{\lbrace \tilde{\tau}_2 >t-u_1  \rbrace}\varphi(u_1,u_1+\tilde{\tau}_2, \cdot )]\right]\right) \right) \bigg\vert_{u_1=\tau_1}\right] \nonumber \\
	&= \mathbb{E}^{\hat{P}}\left[ \left(\textbf{1}_{\lbrace s < u_1 \leq t \rbrace}e^{\int_0^{t-u_1} \tilde{\lambda}_v^2dv} e^{-\int_0^{t-u_1} \tilde{\lambda}_v^2dv}  \mathcal{E}_t \left(\mathbb{E}^{\hat{P}} \left[\textbf{1}_{\lbrace \tilde{\tau}_2 >t-u_1  \rbrace}\varphi(u_1,u_1+\tilde{\tau}_2, \cdot )\right]\right) \right) \bigg\vert_{u_1=\tau_1}\right] \nonumber \\
	&= \mathbb{E}^{\hat{P}}\left[ \mathcal{E}_t \left(\mathbb{E}^{\hat{P}} \left[\textbf{1}_{\lbrace s < u_1 \leq t \rbrace}\textbf{1}_{\lbrace \tilde{\tau}_2 >t-u_1  \rbrace}\varphi(u_1,u_1+\tilde{\tau}_2, \cdot )\right]\right) \bigg\vert_{u_1=\tau_1}\right] \nonumber \\
	& \geq \mathcal{E}_t \left(\mathbb{E}^{\hat{P}}\left[ \mathbb{E}^{\hat{P}} [\textbf{1}_{\lbrace s < u_1 \leq t \rbrace} \textbf{1}_{\lbrace \tilde{\tau}_2 >t- u_1 \rbrace}\varphi(u_1, u_1+\tilde{\tau}_2, \cdot) ]\vert_{u_1=\tau_1}\right] \right) \label{eq:lemma2} \\
&=\mathcal{E}_t \left(\mathbb{E}^{\hat{P}}\left[ \textbf{1}_{\lbrace s < \tau_1 \leq t \rbrace} \textbf{1}_{\lbrace \tilde{\tau}_2 >t- \tau_1 \rbrace}\varphi(\tau_1, \tau_1+ \tilde{\tau}_2, \cdot)\right] \right) \label{eq:WeakDynamicProgrammingThird10} \\
&=\mathcal{E}_t \left(\mathbb{E}^{\hat{P}}\left[ \textbf{1}_{\lbrace s < \tau_1 \rbrace} \textbf{1}_{\lbrace \tau_1 \leq t <\tau_2  \rbrace} Y\right] \right) \notag.
\end{align}}
\endgroup
\red{Note that} \eqref{eq:WeakDynamicProgrammingThird7*} follows from the fact that for fixed $\hat{\omega} \in \hat{\Omega}$ the indicator function $\textbf{1}_{\lbrace \tilde{\tau}_2 >t-u_1 \rbrace}$ is $\mathcal{F}_t$-measurable and from Remark 2.4 (iv) in \cite{nh_2013}. Inequality \eqref{eq:WeakDynamicProgrammingThird7} comes directly from Lemma \ref{lem:francescayinglin}, whereas \eqref{eq:lemma2} follows by applying Lemma \ref{lem:thirdlemmasecondpiece} to $\tilde{Y}:= \textbf{1}_{\lbrace \tilde{\tau}_2 >t-u_1  \rbrace} Y $, which is again a $\mathcal{G}^{\mathcal{P}}$-measurable, upper semianalytic and non-negative random variable on $\tilde{\Omega}$. 
Finally, \eqref{eq:WeakDynamicProgrammingThird10} by Lemma  \ref{lem:secondlemmasecondpiece} applied to $Y\textbf{1}_{\lbrace s < \tau_1 < t < \tau_2 \rbrace}$.
\end{proof}

We are now ready to give the next proposition, which takes into account the first term of the right-hand side of \eqref{eq:WeakDynamicRewritten}.

\begin{prop}\label{prop:dynamicsecond}
Let $Y$ be a $\mathcal{G}^{\mathcal{P}, \red{(2)}}$-measurable, upper semianalytic and non-negative random variable on $\tilde{\Omega}$. For any $0 \le s \le t$ it holds
\begin{small}
\begin{align}
	\textbf{1}_{\lbrace s < \tau_1 \rbrace} \mathcal{E}_s \left(e^{\int_0^s \tilde{\lambda}_v^1 dv } \mathbb{E}^{\hat{P}}[\textbf{1}_{\lbrace s < \tau_1 \rbrace} \tilde{\mathcal{E}}_t^{2}(Y)]\right) \geq \textbf{1}_{\lbrace s < \tau_1 \rbrace} \mathcal{E}_s\left( e^{\int_0^s \tilde{\lambda}_v^1 dv} \mathbb{E}^{\hat{P}}\left [ \textbf{1}_{\lbrace s < \tau_1 \rbrace} Y \right] \right). \label{eq:WeakDynamicProgrammingThirdPart}
\end{align}
\end{small}
\end{prop}
\begin{proof}
From Definition \ref{def:defiSublinearOperatorMulti} we get
\begin{align}
&\textbf{1}_{\lbrace s < \tau_1 \rbrace} \mathcal{E}_s \left(e^{\int_0^s \tilde{\lambda}_v^1 dv } \mathbb{E}^{\hat{P}}[\textbf{1}_{\lbrace s < \tau_1 \rbrace} \tilde{\mathcal{E}}_t^{2}(Y)]\right) \notag \\
	&=\textbf{1}_{\lbrace s < \tau_1 \rbrace} \mathcal{E}_s\bigg( e^{\int_0^s \tilde{\lambda}_v^1 dv} \mathbb{E}^{\hat{P}}\bigg [ \textbf{1}_{\lbrace s < \tau_1 \rbrace} \bigg \lbrace \textbf{1}_{\lbrace t < \tau_1 \rbrace}\mathcal{E}_t(e^{\int_0^t \tilde{\lambda}_v^1 dv} \mathbb{E}^{\hat{P}}[\textbf{1}_{\lbrace t < \tau_1 \rbrace}Y]) \nonumber \\
	&\ + \textbf{1}_{\lbrace \tau_1 \leq t < \tau_2 \rbrace} \mathcal{E}_t\left(e^{\int_0^{t-u_1} \tilde{\lambda}_v^2dv} \mathbb{E}^{\hat{P}} [\textbf{1}_{\lbrace \tilde{\tau}_2 >t-u_1  \rbrace}\varphi(u_1,u_1+\tilde{\tau}_2, \cdot )]\right)\bigg\vert_{u_1=\tau_1} \notag \\
	&\ +	\textbf{1}_{\lbrace \tau_2 \leq t \rbrace} \mathcal{E}_t(\varphi(u_1,u_2, \cdot))\vert_{(u_1,u_2)=(\tau_1,\tau_2)} \bigg \rbrace \bigg ] \bigg) \label{eq:thingfromdefinition}.
	\end{align} 
Moreover, following the same arguments as in the proof of Theorem 2.22 of \cite{bz_2019} we get 
\begin{align}
	\mathbb{E}^{\hat{P}}\left [ \textbf{1}_{\lbrace s < \tau_1 \rbrace} \textbf{1}_{\lbrace t < \tau_1 \rbrace}\mathcal{E}_t\left(e^{\int_0^t \tilde{\lambda}_v^1 dv} \mathbb{E}^{\hat{P}}[\textbf{1}_{\lbrace t < \tau_1 \rbrace}Y]\right)\right]&= \mathcal{E}_t\left(e^{\int_0^t \tilde{\lambda}_v^1 dv} \mathbb{E}^{\hat{P}}[\textbf{1}_{\lbrace t < \tau_1 \rbrace}Y]\right) \mathbb{E}^{\hat{P}}\left [\textbf{1}_{\lbrace t < \tau_1 \rbrace}\right] \nonumber \\
	&=e^{\int_0^t \tilde{\lambda}_v^1 dv} \mathcal{E}_t( \mathbb{E}^{\hat{P}}[\textbf{1}_{\lbrace t < \tau_1 \rbrace}Y]) e^{-\int_0^t \tilde{\lambda}_v^1 dv} \nonumber \\
	&=\mathcal{E}_t( \mathbb{E}^{\hat{P}}[\textbf{1}_{\lbrace t < \tau_1 \rbrace}Y]). \label{eq:WeakDynamicProgrammingThird2}
\end{align}
By Lemma \ref{lem:secondtermsecondpiece} we have
\begin{small}
\begin{align}
	& \mathbb{E}^{\hat{P}}\left[  \textbf{1}_{\lbrace s<\tau_1 \leq t < \tau_2 \rbrace} \mathcal{E}_t\left(e^{\int_0^{t-u_1} \tilde{\lambda}_v^2dv} \mathbb{E}^{\hat{P}} [\textbf{1}_{\lbrace \tilde{\tau}_2 >t-u_1  \rbrace}\varphi(u_1,u_1+\tilde{\tau}_2, \cdot )]\right)\bigg\vert_{u_1=\tau_1} \right]\notag \\
	 & \quad \geq \mathcal{E}_t\left( \mathbb{E}^{\hat{P}}\left[ \textbf{1}_{\lbrace s < \tau_1 \leq t < \tau_2 \rbrace} Y\right]\right),\label{eq:WeakDynamicProgrammingThird2_*}
\end{align}
\end{small}
By Lemma \ref{lem:firstlemmasecondpiece} it follows that
\begin{align}
	\mathbb{E}^{\hat{P}}\left[ \textbf{1}_{\lbrace s < \tau_1 \rbrace}  \textbf{1}_{\lbrace \tau_2 \leq t \rbrace} \mathcal{E}_t \left( \varphi(u_1,u_2, \cdot )\right)\vert_{(u_1,u_2)=(\tau_1, \tau_2)}\right]= \mathbb{E}^{\hat{P}}\left[ \textbf{1}_{\lbrace s < \tau_1 \rbrace}  \textbf{1}_{\lbrace \tau_2 \leq t \rbrace} \mathcal{E}_t (Y)\right]. \label{eq:WeakDynamicProgrammingThird1}
\end{align}

Putting together \eqref{eq:thingfromdefinition},\eqref{eq:WeakDynamicProgrammingThird2},  \eqref{eq:WeakDynamicProgrammingThird2_*}, \eqref{eq:WeakDynamicProgrammingThird1} and Remark 2.4 (iii) in \cite{nh_2013}, with the notation \eqref{eq:NotationII}, we get that
\red{\small{
\begin{align}
&\textbf{1}_{\lbrace s < \tau_1 \rbrace} \mathcal{E}_s (e^{\int_0^s \tilde{\lambda}_v^1 dv } \mathbb{E}^{\hat{P}}[\textbf{1}_{\lbrace s < \tau_1 \rbrace} \tilde{\mathcal{E}}_t^{2}(Y)]) \notag \\
	&\geq \textbf{1}_{\lbrace s < \tau_1 \rbrace}  e^{\int_0^s \tilde{\lambda}_v^1 dv}  \mathcal{E}_s\left( \mathcal{E}_t( \mathbb{E}^{\hat{P}}[\textbf{1}_{\lbrace t < \tau_1 \rbrace}Y])+\mathcal{E}_t \left(\mathbb{E}^{\hat{P}}\left[ \textbf{1}_{\lbrace s < \tau_1 \rbrace} \textbf{1}_{\lbrace \tau_1 \leq t <\tau_2  \rbrace} Y\right] \right) +\mathbb{E}^{\hat{P}}\left[ \textbf{1}_{\lbrace s < \tau_1 \rbrace}  \textbf{1}_{\lbrace \tau_2 \leq t \rbrace} \mathcal{E}_t (Y)\right] \right) \label{eq:FirstEquality*}\\
	&=\textbf{1}_{\lbrace s < \tau_1 \rbrace}  e^{\int_0^s \tilde{\lambda}_v^1 dv}  \mathcal{E}_s\left( \mathcal{E}_t( \mathbb{E}^{\hat{P}}[\textbf{1}_{\lbrace t < \tau_1 \rbrace}Y])+ \mathcal{E}_t \left(\mathbb{E}^{\hat{P}}\left[ \textbf{1}_{\lbrace s < \tau_1 \rbrace} \textbf{1}_{\lbrace \tau_1 \leq t <\tau_2  \rbrace} Y\right] \right) +\mathbb{E}^{\hat{P}}[ \mathcal{E}_t ( \textbf{1}_{\lbrace s < \tau_1 \rbrace}  \textbf{1}_{\lbrace \tau_2 \leq t \rbrace}  Y)] \right) \label{eq:FirstEquality} \\
	&\geq \textbf{1}_{\lbrace s < \tau_1 \rbrace}  e^{\int_0^s \tilde{\lambda}_v^1 dv}  \mathcal{E}_s\left( \mathcal{E}_t( \mathbb{E}^{\hat{P}}[\textbf{1}_{\lbrace t < \tau_1 \rbrace}Y])+  \mathcal{E}_t\left(\mathbb{E}^{\hat{P}}[    \textbf{1}_{\lbrace s < \tau_1 \rbrace}   \textbf{1}_{\lbrace \tau_1 \leq t < \tau_2 \rbrace}Y ]\right) +\mathcal{E}_t(\mathbb{E}^{\hat{P}}[ \textbf{1}_{\lbrace s < \tau_1 \rbrace}  \textbf{1}_{\lbrace \tau_2 \leq t \rbrace}  Y]) \right) \label{eq:SecondInequality}\\
	&\geq \textbf{1}_{\lbrace s < \tau_1 \rbrace}  e^{\int_0^s \tilde{\lambda}_v^1 dv}  \mathcal{E}_s\left( \mathcal{E}_t( \mathbb{E}^{\hat{P}}[\textbf{1}_{\lbrace t < \tau_1 \rbrace}Y]+\mathbb{E}^{\hat{P}}[    \textbf{1}_{\lbrace s < \tau_1 \rbrace}   \textbf{1}_{\lbrace \tau_1 \leq t < \tau_2 \rbrace}Y ] +\mathbb{E}^{\hat{P}}[ \textbf{1}_{\lbrace s < \tau_1 \rbrace}  \textbf{1}_{\lbrace \tau_2 \leq t \rbrace}  Y]) \right) \label{eq:ThirdInequality} \\
	&=\textbf{1}_{\lbrace s < \tau_1 \rbrace}   \mathcal{E}_s\left( e^{\int_0^s \tilde{\lambda}_v^1 dv} \mathbb{E}^{\hat{P}}[\textbf{1}_{\lbrace s < \tau_1 \rbrace}Y] \right), \nonumber
\end{align}}}
where we use Lemma  \ref{lem:secondtermsecondpiece} in \eqref{eq:FirstEquality*}. Equality \eqref{eq:FirstEquality} holds by Remark 2.4 (iv) in \cite{nh_2013} together with the fact that  
\begin{equation} \label{eq:MeasurabilityFixedOmegaHat}
\lbrace{ s < \tau_1(\cdot, \hat{\omega})\rbrace},  \lbrace{ \tau_2(\cdot, \hat{\omega}) \leq t\rbrace} \in \mathcal{F}_t
\end{equation}
for fixed $\hat{\omega} \in \hat{\Omega}$. Inequality \eqref{eq:SecondInequality} follows from Lemma \ref{lem:francescayinglin}, and \eqref{eq:ThirdInequality} by the sublinearity of $\mathcal{E}_t$.
\end{proof}

We finally consider the second term of the right-hand side of \eqref{eq:WeakDynamicRewritten}, i.e.
$$
\textbf{1}_{\lbrace \tau_1 \leq s < \tau_2 \rbrace} \mathcal{E}_s\left(e^{\int_0^{s-u_1} \tilde{\lambda}_v^2 dv } \mathbb{E}^{\hat{P}}[ \textbf{1}_{\lbrace \tilde{\tau}_2 > s-u_1\rbrace} \bar{\varphi}_t(u_1,u_1+\tilde{\tau}_2, \cdot)]\right) \bigg\vert_{u_1=\tau_1}.
$$
\begin{prop}\label{prop:dynamicthird}
For any $0 \le s \le t$ it holds 
\begin{align}
	&\textbf{1}_{\lbrace \tau_1 \leq s < \tau_2 \rbrace} \mathcal{E}_s\left(e^{\int_0^{s-u_1} \tilde{\lambda}_v^2 dv } \mathbb{E}^{\hat{P}}[ \textbf{1}_{\lbrace \tilde{\tau}_2 > s-u_1\rbrace} \bar{\varphi}_t(u_1,u_1+\tilde{\tau}_2, \cdot)]\right) \bigg\vert_{u_1=\tau_1} \nonumber \\
	& \geq \textbf{1}_{\lbrace \tau_1 \leq s < \tau_2 \rbrace} \mathcal{E}_s\left (e^{\int_0^{s-u_1} \tilde{\lambda}_v^2 dv} \mathbb{E}^{\hat{P}}\left[\textbf{1}_{\lbrace \tilde{\tau}_2 > s-u_1\rbrace} \varphi(u_1,u_1+\tilde{\tau}_2, \cdot)\right]\right)\bigg \vert_{u_1=\tau_1}, \label{eq:WeakDynamicMiddleTerm}
\end{align}
where $\bar{\varphi}_t$ is defined in \eqref{eq:defbarvarphi} and $\varphi$ is given in \eqref{eq:DefiVarphi1} and \eqref{eq:DefiVarphi2}.
\end{prop}
\begin{proof}
From \eqref{eq:defbarvarphi} we get
\begin{align}
&\textbf{1}_{\lbrace \tau_1 \leq s < \tau_2 \rbrace} \mathcal{E}_s\left(e^{\int_0^{s-u_1} \tilde{\lambda}_v^2 dv } \mathbb{E}^{\hat{P}}[ \textbf{1}_{\lbrace \tilde{\tau}_2 > s-u_1\rbrace}{\bar{\varphi}_t}(u_1,u_1+\tilde{\tau}_2, \cdot)]\right) \bigg\vert_{u_1=\tau_1}\notag \\
&=\textbf{1}_{\lbrace \tau_1 \leq s < \tau_2 \rbrace} \mathcal{E}_s\bigg (e^{\int_0^{s-u_1} \tilde{\lambda}_v^2 dv} \mathbb{E}^{\hat{P}}\bigg[\textbf{1}_{\lbrace \tilde{\tau}_2 > s-u_1\rbrace} \bigg \lbrace \textbf{1}_{\lbrace u_1 > t \rbrace} \mathcal{E}_t(e^{\int_0^t \tilde{\lambda}_v^1 dv} \mathbb{E}^{\hat{P}}[\textbf{1}_{\lbrace t < \tau_1 \rbrace}Y]) \nonumber \\
	&\ +\textbf{1}_{\lbrace u_1 \leq t \leq u_1+\tilde{\tau}_2 \rbrace} \mathcal{E}_t(e^{\int_0^{t-u_1} \tilde{\lambda}_v^2 dv} \mathbb{E}^{\hat{P}}[\textbf{1}_{\lbrace \tilde{\tau}_2 >t-u_1 \rbrace} \varphi(u_1,u_1+\tilde{\tau}_2, \cdot)]) \notag \\ \ &\ + \textbf{1}_{\lbrace u_1+\tilde{\tau}_2 \leq t \rbrace} \mathcal{E}_t(\varphi(u_1,u_1+\tilde{\tau}_2,\cdot))\bigg\rbrace \bigg]\bigg) \bigg \vert_{u_1=\tau_1} \label{eq:initialtermlastprop}.
\end{align}
Since $\textbf{1}_{\lbrace u_1 > t \rbrace} \mathcal{E}_t\left(e^{\int_0^t \tilde{\lambda}_v^1 dv} \mathbb{E}^{\hat{P}}[\textbf{1}_{\lbrace t < \tau_1 \rbrace}Y]\right)$ is independent of $\hat{\omega}$, we have
\begin{small}
\begin{align}
	&\mathbb{E}^{\hat{P}}\left[ \textbf{1}_{\lbrace \tilde{\tau}_2 > s-u_1\rbrace} \textbf{1}_{\lbrace u_1 > t \rbrace} \mathcal{E}_t\left(e^{\int_0^t \tilde{\lambda}_v^1 dv} \mathbb{E}^{\hat{P}}[\textbf{1}_{\lbrace t < \tau_1 \rbrace}Y]\right) \right]\nonumber \\
	&= \mathbb{E}^{\hat{P}}\left[ \textbf{1}_{\lbrace \tilde{\tau}_2 > s-u_1\rbrace}\right]  \textbf{1}_{\lbrace u_1 > t \rbrace}\mathcal{E}_t\left(e^{\int_0^t \tilde{\lambda}_v^1 dv} \mathbb{E}^{\hat{P}}[\textbf{1}_{\lbrace t < \tau_1 \rbrace}Y]\right)\notag \\
	&=e^{\int_0^{s-u_1} \tilde{\lambda}_v^2 dv} \textbf{1}_{\lbrace u_1 > t \rbrace} e^{\int_0^t \tilde{\lambda}_v^1 dv} \mathcal{E}_t(\mathbb{E}^{\hat{P}}[\textbf{1}_{\lbrace t < \tau_1 \rbrace}Y]). \label{eq:WeakDynamicMiddleTerm1}
\end{align}
\end{small}

Also note that $\textbf{1}_{\lbrace u_1 \leq   t \rbrace} \mathcal{E}_t \left(e^{\int_0^{t-u_1} \tilde{\lambda}_v^2 dv} \mathbb{E}^{\hat{P}}[\textbf{1}_{\lbrace \tilde{\tau}_2 >t-u_1 \rbrace} \varphi(u_1,u_1+\tilde{\tau}_2, \cdot) ]\right)$ is independent of $\hat{\omega}$, so that 
\begin{align}
	&\mathbb{E}^{\hat{P}}\left[ \textbf{1}_{\lbrace \tilde{\tau}_2 > s-u_1\rbrace} \textbf{1}_{\lbrace u_1 \leq   t <u_1 +\tilde{\tau}_2\rbrace} \mathcal{E}_t \left(e^{\int_0^{t-u_1} \tilde{\lambda}_v^2 dv} \mathbb{E}^{\hat{P}}[\textbf{1}_{\lbrace \tilde{\tau}_2 >t-u_1 \rbrace} \varphi(u_1,u_1+\tilde{\tau}_2, \cdot) ]\right)\right] \nonumber \\
	&=\mathbb{E}^{\hat{P}}\left[ \textbf{1}_{\lbrace \tilde{\tau}_2 > s-u_1\rbrace} \textbf{1}_{\lbrace \tilde{\tau}_2 > t-u_1\rbrace}\right] \textbf{1}_{\lbrace u_1 \leq   t \rbrace} e^{\int_0^{t-u_1} \tilde{\lambda}_v^2 dv} \mathcal{E}_t ( \mathbb{E}^{\hat{P}}[\textbf{1}_{\lbrace \tilde{\tau}_2 >t-u_1 \rbrace} \varphi(u_1,u_1+\tilde{\tau}_2, \cdot) ]) \nonumber \\
	&=\mathbb{E}^{\hat{P}}\left[  \textbf{1}_{\lbrace \tilde{\tau}_2 > t-u_1\rbrace}\right] \textbf{1}_{\lbrace u_1 \leq   t \rbrace} e^{\int_0^{t-u_1} \tilde{\lambda}_v^2 dv} \mathcal{E}_t ( \mathbb{E}^{\hat{P}}[\textbf{1}_{\lbrace \tilde{\tau}_2 >t-u_1 \rbrace} \varphi(u_1,u_1+\tilde{\tau}_2, \cdot) ]) \nonumber  \\
	&=e^{-\int_0^{t-u_1} \tilde{\lambda}_v^2 dv}\textbf{1}_{\lbrace u_1 \leq   t \rbrace} e^{\int_0^{t-u_1} \tilde{\lambda}_v^2 dv} \mathcal{E}_t ( \mathbb{E}^{\hat{P}}[\textbf{1}_{\lbrace \tilde{\tau}_2 >t-u_1 \rbrace} \varphi(u_1,u_1+\tilde{\tau}_2, \cdot) ]) \nonumber \\
	&=\textbf{1}_{\lbrace u_1 \leq   t \rbrace}\mathcal{E}_t ( \mathbb{E}^{\hat{P}}[\textbf{1}_{\lbrace \tilde{\tau}_2 >t-u_1 \rbrace} \varphi(u_1,u_1+\tilde{\tau}_2, \cdot) ]).\label{eq:WeakDynamicMiddleTerm2}
\end{align}
\allowdisplaybreaks
Putting together \eqref{eq:WeakDynamicMiddleTerm1}, \eqref{eq:WeakDynamicMiddleTerm2} and Remark 2.4 (iii) in \cite{nh_2013}, we have that the right-hand side of \eqref{eq:initialtermlastprop} is equal to 
\red{\begin{align}
	&\textbf{1}_{\lbrace \tau_1 \leq s < \tau_2 \rbrace} e^{\int_0^{s-u_1} \tilde{\lambda}_v^2 dv} \mathcal{E}_s\bigg ( e^{\int_0^{s-u_1} \tilde{\lambda}_v^2 dv} \textbf{1}_{\lbrace u_1 > t \rbrace} e^{\int_0^t \tilde{\lambda}_v^1 dv} \mathcal{E}_t(\mathbb{E}^{\hat{P}}[\textbf{1}_{\lbrace t < \tau_1 \rbrace}Y]) \nonumber \\
	&+ \textbf{1}_{\lbrace u_1 \leq   t \rbrace}\mathcal{E}_t ( \mathbb{E}^{\hat{P}}[\textbf{1}_{\lbrace \tilde{\tau}_2 >t-u_1 \rbrace} \varphi(u_1,u_1+\tilde{\tau}_2, \cdot) ])  \nonumber \\
	&+ \mathbb{E}^{\hat{P}}\left[ \textbf{1}_{\lbrace \tilde{\tau}_2 > s-u_1\rbrace} \textbf{1}_{\lbrace u_1+\tilde{\tau}_2 \leq t \rbrace} \mathcal{E}_t(\varphi(u_1,u_1+\tilde{\tau}_2,\cdot)) \right]\bigg) \bigg \vert_{u_1=\tau_1} \nonumber\\
	&=\textbf{1}_{\lbrace \tau_1 \leq s < \tau_2 \rbrace} e^{\int_0^{s-u_1} \tilde{\lambda}_v^2 dv} \mathcal{E}_s\bigg(\mathcal{E}_t ( \textbf{1}_{\lbrace u_1 \leq   t \rbrace}\mathbb{E}^{\hat{P}}[\textbf{1}_{\lbrace \tilde{\tau}_2 >t-u_1 \rbrace} \varphi(u_1,u_1+\tilde{\tau}_2, \cdot) ])  \nonumber\\
	&+\mathbb{E}^{\hat{P}}\left[ \mathcal{E}_t(\textbf{1}_{\lbrace \tilde{\tau}_2 > s-u_1\rbrace} \textbf{1}_{\lbrace u_1+\tilde{\tau}_2 \leq t \rbrace} \varphi(u_1,u_1+\tilde{\tau}_2,\cdot)) \right]\bigg) \bigg \vert_{u_1=\tau_1} \label{eq:WeakDynamicMiddleTerm4*} \\
	&\geq \textbf{1}_{\lbrace \tau_1 \leq s < \tau_2 \rbrace} e^{\int_0^{s-u_1} \tilde{\lambda}_v^2 dv} \mathcal{E}_s\bigg(\mathcal{E}_t \big( \textbf{1}_{\lbrace u_1 \leq   t \rbrace}\mathbb{E}^{\hat{P}}[\textbf{1}_{\lbrace \tilde{\tau}_2 >t-u_1 \rbrace} \varphi(u_1,u_1+\tilde{\tau}_2, \cdot) ])  \nonumber\\
	&+\mathcal{E}_t\left(\mathbb{E}^{\hat{P}}\left[ \textbf{1}_{\lbrace \tilde{\tau}_2 > s-u_1\rbrace} \textbf{1}_{\lbrace u_1+\tilde{\tau}_2 \leq t \rbrace} \varphi(u_1,u_1+\tilde{\tau}_2,\cdot) \right]\right) \bigg) \bigg \vert_{u_1=\tau_1} \label{eq:WeakDynamicMiddleTerm3}\\
	& \geq \textbf{1}_{\lbrace \tau_1 \leq s < \tau_2 \rbrace} e^{\int_0^{s-u_1} \tilde{\lambda}_v^2 dv} \mathcal{E}_s\bigg(\mathcal{E}_t \big( \mathbb{E}^{\hat{P}}[\textbf{1}_{\lbrace u_1 \leq   t \rbrace}\textbf{1}_{\lbrace \tilde{\tau}_2 >t-u_1 \rbrace} \varphi(u_1,u_1+\tilde{\tau}_2, \cdot) ] \nonumber \\
	&+\mathbb{E}^{\hat{P}}\left[ \textbf{1}_{\lbrace \tilde{\tau}_2 > s-u_1\rbrace} \textbf{1}_{\lbrace u_1+\tilde{\tau}_2 \leq t \rbrace} \varphi(u_1,u_1+\tilde{\tau}_2,\cdot) \right]\big) \bigg) \bigg \vert_{u_1=\tau_1} \label{eq:WeakDynamicMiddleTerm4}\\
	& = \textbf{1}_{\lbrace \tau_1 \leq s < \tau_2 \rbrace} e^{\int_0^{s-u_1} \tilde{\lambda}_v^2 dv} \notag \\
	& \quad \cdot \mathcal{E}_s\bigg(\mathcal{E}_t \big( \mathbb{E}^{\hat{P}}\left[\big(\textbf{1}_{\lbrace u_1 \leq   t \rbrace}\textbf{1}_{\lbrace \tilde{\tau}_2 >t-u_1 \rbrace} +\textbf{1}_{\lbrace \tilde{\tau}_2 > s-u_1\rbrace} \textbf{1}_{\lbrace u_1+\tilde{\tau}_2 \leq t \rbrace} \big)\varphi(u_1,u_1+\tilde{\tau}_2,\cdot) \right]\big) \bigg) \bigg \vert_{u_1=\tau_1} \nonumber  \\
	&=\textbf{1}_{\lbrace \tau_1 \leq s < \tau_2 \rbrace} e^{\int_0^{s-u_1} \tilde{\lambda}_v^2 dv} \mathcal{E}_s\big( \mathbb{E}^{\hat{P}}[\textbf{1}_{\lbrace \tilde{\tau}_2 >s-u_1 \rbrace} \varphi(u_1,u_1+\tilde{\tau}_2, \cdot) ] \big)  \big \vert_{u_1=\tau_1}. \label{eq:WeakDynamicMiddleTerm4**}
\end{align}}
Note that \red{\eqref{eq:WeakDynamicMiddleTerm4*} comes from $\textbf{1}_{\lbrace u_1 \leq s\rbrace} \textbf{1}_{\lbrace u_1 >t \rbrace}=0$ for $0 \leq s \leq t$ and Remark 2.4 (iv) in \cite{nh_2013}, together with \eqref{eq:MeasurabilityFixedOmegaHat}}. The inequality in \eqref{eq:WeakDynamicMiddleTerm3} follows by Lemma \ref{lem:francescayinglin} and the one in \eqref{eq:WeakDynamicMiddleTerm4} by the sublinearity of $\mathcal{E}_t$. 
As $ \textbf{1}_{\lbrace \tilde{\tau}_2 > t-u_1 \rbrace}\textbf{1}_{\lbrace \tilde{\tau}_2 > s-u_1 \rbrace}=\textbf{1}_{\lbrace \tilde{\tau}_2 > t-u_1 \rbrace}$, we have
\red{\begin{align}
	&\textbf{1}_{\lbrace u_1 \leq   t \rbrace}\textbf{1}_{\lbrace \tilde{\tau}_2 >t-u_1 \rbrace} +\textbf{1}_{\lbrace \tilde{\tau}_2 > s-u_1\rbrace} \textbf{1}_{\lbrace u_1+\tilde{\tau}_2 \leq t \rbrace} 
	=\textbf{1}_{\lbrace \tilde{\tau}_2 > s-u_1\rbrace} \textbf{1}_{\lbrace u_1 \leq   t \rbrace} \label{eq:Indicator2}.
\end{align}}
Then $\eqref{eq:WeakDynamicMiddleTerm4**}$ \red{follows by \eqref{eq:Indicator2} and the tower property of $\mathcal{E}_t$, together with the fact that} for $0 \leq s \leq t $ \red{we have}
\begin{equation*}
	\textbf{1}_{\lbrace \tau_1 \leq s < \tau_2 \rbrace} \textbf{1}_{\lbrace \tau_1 \leq t \rbrace}= \textbf{1}_{\lbrace \tau_1 \leq s < \tau_2 \rbrace}.
	\end{equation*} 
\end{proof}

Putting together Definition \ref{def:defiSublinearOperatorMulti} with Propositions \ref{prop:dynamicfirst}, \ref{prop:dynamicsecond} and \ref{prop:dynamicthird}, we then get the following theorem.
\begin{theorem}\label{thm:weakdynamicprogrammingprinciple}
	Let Assumption \ref{assumptionnutzNew} hold and $Y$ be an upper semianalytic function on $\tilde{\Omega}$ such that $Y$ is $\mathcal{G}^{\mathcal{P}, \red{(k)}}$-measurable and non-negative. Then for any $0 \leq s \leq t,k=1,...,N$ it holds
	\begin{equation} \notag
		\tilde{\mathcal{E}}^{k}_s(\tilde{\mathcal{E}}_t^{k}(Y)) \geq \tilde{\mathcal{E}}_s^{k}(Y) \quad \tilde{P} \text{-a.s. for all } \tilde{P} \in \tilde{\mathcal{P}}. 
	\end{equation}
\end{theorem}
\begin{proof}
	For $k =3,...,N$ the result follows by similar arguments as in the Propositions \ref{prop:dynamicfirst}, \ref{prop:dynamicsecond} and \ref{prop:dynamicthird}.
\end{proof}

\red{
\begin{remark}
	Note that the strong dynamic programming principle is not any longer valid for the operator $\tilde{\mathcal{E}}^k$ for $k=1,...,N$. This follows immediately by combining point (i) of Proposition \ref{prop:OperatorCoincides} with Appendix A in \cite{bz_2019}, where a counterexample for a function $Y$ in the case $k=1$ is provided. We also refer to Remark 2.24 in \cite{bz_2019} for some further comments. \\Furthermore, whether the strong dynamic programming principle is satisfied or not, it depends not only on the derivative $Y$ but also on the set of priors $\mathcal{P}$ with respect to which we define the operator $\tilde{\mathcal{E}}^k$. It is then difficult to assess this property in full generality.  However, in the next section we study some basket credit derivatives and derive some sufficient conditions under which the strong tower property holds for these financial instruments. 
	\end{remark}
}

\section{Valuation of credit portfolio products in a multiple default setting under model uncertainty} \label{sec:valuation}
 In this section we study the valuation of basket credit derivatives in a multiple default setting under model uncertainty. For more details on these insurance products we refer to Section 9.1 in \cite{bielecki_rutkowski_2004}. \red{Moreover, there also exists an extensive literature on the modeling of CDOs, see e.g. \cite{cont_minca_2011}, \cite{Ehlers_Schoenbucher_2009}, \cite{filipovic_overbeck_schmidt_2010}, \cite{giesecke_goldberg_ding_2009}, \cite{sidenius_piterbarg_andersen_2008}.} Standard products in this class of derivatives are given by the collection of the so called $i$-th to default contingent claim CCT$^{(i)}$ with maturity $T>0$ for any $i=1,...,N$. In the framework outlined in Section \ref{sec:coxmodel}, we assume that the multiple defaults are represented by a collection of ordered stopping times $(\tau_i)_{1 \le i \le n}$ constructed as in \eqref{eq:DefiTauFixedProbability} and \eqref{eq:DefiTildeTauFixedProbability}. In this setting, for any $i=1,\dots,N$, the payoff of the claim CCT$^{(i)}$ is defined as follows:
\begin{itemize}
\item If $\tau_i \leq T$, then the claim pays the amount $Z^i_{\tau_i}$ at time $\tau_i$, where $Z^i$ is given by a $\mathbb{G}^{(N)}$-predictable process, and a $\mathcal{G}_T^{(N)}$-measurable amount $X^i$ at time $T$.
\item If $\tau_i >T$, the holder receives an amount $X$ at time $T$, where $X$ is given by a non-negative and $\mathcal{G}_T^{(N)}$-measurable random variable.
 \end{itemize}
  We now evaluate a special type of these contracts by using the operator $\tilde{\mathcal{E}}^N$ in the following setting. For the financial interpretation of $\tilde{\mathcal{E}}^N$ as pricing operator, we refer to Section \ref{sec:Superhedging}.\\
Let $T < \infty$ be the maturity time. We define the filtration 
$\mathbb{F}^{\mathcal{P}}:=(\mathcal{F}_t^{\mathcal{P}})_{t \in [0,T]}$ by
\begin{equation}
	\mathcal{F}_t^{\mathcal{P}}:=\mathcal{F}_t^* \vee \mathcal{N}_{T}^{\mathcal{P}}, \quad t \in [0,T], \label{filtrationP}
\end{equation}
where $\mathcal{N}_{T}^{\mathcal{P}}$ is the collection of sets which are $(P, \mathcal{F}_{T})$-null for all $P \in \mathcal{P}$. For fixed $i=1,...,N$  we consider a product associated to the $i$-th default event, defined in particular by the following payoff:
\begin{enumerate}
	\item $\textbf{1}_{\lbrace \tau_i >T\rbrace}Y$, where $Y$ is a $\mathcal{F}_T^{\mathcal{P}}$-measurable, non-negative and upper semianalytic function on $\Omega$, such that $\mathcal{E}(Y)< \infty$;
	\item $\textbf{1}_{\lbrace 0 < \tau_i \leq T\rbrace}Z_{\tau_i}$, where $Z=(Z_t)_{t \in [0,T]}$ is an $\mathbb{F}^{\mathcal{P}}$-predictable non-negative process on $\Omega$, such that the function $Z(t,\omega):=Z_t(\omega), (t,\omega) \in [0,T] \times \Omega$, is upper semianalytic and $\sup_{t \in [0,T]} \mathcal{E}(Z_t) < \infty$.
\end{enumerate}
The payoff defined above can be seen as a particular case of a CCT$^{(i)}$, for any fixed $i=1,\dots,N$, and as a generalization for multiple default times  of the insurance products studied in Section 2.4 of \cite{bz_2019}.
Before we evaluate this payoff for $N=2$, we state an auxiliary lemma which we need in the following. 
\begin{lemma} \label{lemma:ExpectationCalculations}
Fix $\omega \in \Omega$ and let $h: (\mathbb{R}_+ \times \Omega, \mathcal{B}(\mathbb{R}_+) \otimes \mathcal{F}^{\mathcal{P}}_{T}) \to (\mathbb{R}, \mathcal{B}(\mathbb{R}))$ be a function such that 
$$
 \mathbb{E}^{\hat{P}}\left[ \vert h(\tilde{\tau}_i(\omega, \hat{\omega}), \omega) \vert \right]<\infty, 
$$
$i=1,2$.
Then we have 
	\begin{equation} \label{eq:Expectation1}
		\mathbb{E}^{\hat{P}}\left[\textbf{1}_{\lbrace s <  \tilde{\tau}_{i}(\omega, \hat{\omega}) \leq t\rbrace} h(\tilde{\tau}_{i}(\omega, \hat{\omega}), \omega) \right]=\int_t^s h(x, \omega) e^{-\int_0^x\tilde{\lambda}^{i}_y(\omega)dy} \tilde{\lambda}^{i}_x(\omega) dx
	\end{equation}
	for all $0 \leq s \leq t \leq T$ and $i=1,2$.
\end{lemma}

\begin{proof}
	Let $i =1,2$. For fixed $\omega \in \Omega$ and $u \geq 0$ it holds
	\begin{equation*}
		\hat{P}(\tilde{\tau}_{i}(\omega, \hat{\omega}) \leq u )= 1-\hat{P}(\tilde{\tau}_{i} (\omega, \hat{\omega}) > u )= 1-e^{-\int_0^u \tilde{\lambda}_y^{i}(\omega) dy}.
	\end{equation*}
	It follows that for fixed $\omega \in \Omega$ the density of $\tilde{\tau}_{i}(\cdot, \omega)$ is given by the function $u \to e^{-\int_0^u \tilde{\lambda}_y^{i}(\omega) dy}\tilde{\lambda}_u^{i}(\omega)$.
	 Therefore we have
	\begin{align*}
		\mathbb{E}^{\hat{P}}\left[\textbf{1}_{\lbrace s <  \tilde{\tau}_{i}(\omega,\hat{\omega}) \leq t\rbrace} h(\tilde{\tau}_{i}(\omega, \hat{\omega}),\omega) \right]	= \int_t^s h(x,\omega) e^{-\int_0^x\tilde{\lambda}^{i}_y(\omega)dy} \tilde{\lambda}^{i}_x(\omega) dx
	\end{align*}
	for $0 \leq s \leq t \leq T$.
\end{proof}

We start by valuating the payoff payed in case of no default.

\begin{prop}
	Let $Y$ be an $\mathcal{F}_T^{\mathcal{P}}$-measurable upper semianalytic function on $\Omega$ such that $\mathcal{E}(\vert Y \vert)< \infty$. Then for every $t \in [0,T]$ the random variables
	\begin{gather*}
		\textbf{1}_{\lbrace \tau_1 > T \rbrace}Y, \quad Y e^{-\int_t^T \tilde{\lambda}_s^1 ds}, \quad \textbf{1}_{\lbrace \tau_2 > T \rbrace}Y,  \\
		Y \left(\int_t^T e^{-\int_0^{T-x} \tilde{\lambda}_s^2 ds}e^{-\int_{t}^x \tilde{\lambda}_y^1 dy}\tilde{\lambda}_x^1dx+e^{-\int_{t}^{T}\tilde{\lambda}_s^1 ds}\right), \quad  Ye^{-\int_{t-u_1}^{T-u_1} \tilde{\lambda}_s^2 ds}
	\end{gather*}
	are upper semianalytic. Moreover, $\textbf{1}_{\lbrace \tau_1 >T \rbrace}Y$ and $\textbf{1}_{\lbrace \tau_2 >T \rbrace}Y$ belong to $L^1(\tilde{\Omega})$. If $\mathcal{P}$ satisfies Assumption \ref{assumptionnutzNew}, we have that
	\begin{equation} \label{eq:FirstPayoffi=1}
		\tilde{\mathcal{E}}_t^2\left(\textbf{1}_{\lbrace \tau_1 > T \rbrace}Y\right)= \textbf{1}_{\lbrace \tau_1 > t \rbrace} \mathcal{E}_t\left(Y e^{-\int_t^T \tilde{\lambda}_s^1 ds}\right)
	\end{equation} 
	and
	
	\begin{align} \label{eq:FirstPayoffi=2}
		\tilde{\mathcal{E}}_t^2\left(\textbf{1}_{\lbrace \tau_2 > T \rbrace}Y\right)&= \textbf{1}_{\lbrace \tau_1 > t \rbrace} \mathcal{E}_t \left( Y \left[\int_t^T e^{-\int_0^{T-x} \tilde{\lambda}_s^2 ds}e^{-\int_{t}^x \tilde{\lambda}_y^1 dy}\tilde{\lambda}_x^1dx+e^{-\int_{t}^{T}\tilde{\lambda}_s^1 ds}\right]\right) \notag \\
		&\ + \textbf{1}_{\lbrace \tau_1 \leq t < \tau_2\rbrace} \mathcal{E}_t \left( Ye^{-\int_{t-u_1}^{T-u_1} \tilde{\lambda}_s^2 ds}\right)\bigg \vert_{u_1= \tau_1}
	\end{align}

	for every $t \in [0,T]$.
\end{prop}

\begin{proof}
Clearly $\textbf{1}_{\lbrace \tau_1 > T \rbrace}Y$ and $Y e^{-\int_t^T \tilde{\lambda}_s^1 ds}$ are upper semianalytic and belong to $L^1(\tilde{\Omega})$.
	As $\textbf{1}_{\lbrace \tau_1 > T \rbrace}Y$ is measurable with respect to $\mathcal{G}_{\infty}^{(1)}$, Proposition \ref{prop:OperatorCoincides} implies that
	\begin{equation} \label{eq:FirstPayoffi=1Coincides}
		\tilde{\mathcal{E}}_t^2\left(\textbf{1}_{\lbrace \tau_1 > T \rbrace}Y\right)= \tilde{\mathcal{E}}_t^1\left(\textbf{1}_{\lbrace \tau_1 > T \rbrace}Y\right)
	\end{equation}
	 for all $t \in [0,T]$. Thus \eqref{eq:FirstPayoffi=1} follows directly by \eqref{eq:FirstPayoffi=1Coincides} and Lemma 2.26 in \cite{bz_2019}. \\
	 By noting that $\textbf{1}_{\lbrace \tau_2 > T \rbrace}, e^{\int_0^t \tilde{\lambda}_s^1 ds}, \int_t^T e^{-\int_0^{T-x} \tilde{\lambda}_s^2 ds}e^{-\int_{t}^x \tilde{\lambda}_y^1 dy}\tilde{\lambda}_x^1dx, e^{-\int_{t}^{T}\tilde{\lambda}_s^1 ds}$ and $e^{-\int_{t-u_1}^{T-u_1} \tilde{\lambda}_s^2 ds}$ are non-negative Borel-measurable functions, it follows by point 4 and 5 in Lemma \ref{lemma:propertiesUpperSemianalytic} that the random variables
	 \begin{equation*}
	 	 \textbf{1}_{\lbrace \tau_2 > T \rbrace}Y, \quad Y \left(\int_t^T e^{-\int_0^{T-x} \tilde{\lambda}_s^2 ds}e^{-\int_{t}^x \tilde{\lambda}_y^1 dy}\tilde{\lambda}_x^1dx+e^{-\int_{t}^{T}\tilde{\lambda}_s^1 ds}\right), \quad Ye^{-\int_{t-u_1}^{T-u_1} \tilde{\lambda}_s^2 ds}
	 \end{equation*} are upper semianalytic. Moreover, it holds
	\begin{align*}
	 	\tilde{\mathcal{E}}(\vert Y \textbf{1}_{\lbrace \tau_2 > T\rbrace}\vert)=  \sup_{\tilde{P} \in \tilde{\mathcal{P}}} \mathbb{E}^{\tilde{P}} \left[ \vert Y \textbf{1}_{\lbrace \tau_2 > T\rbrace}\vert \right] \leq \sup_{\tilde{P} \in \tilde{\mathcal{P}}} \mathbb{E}^{\tilde{P}} \left[ \vert Y \vert \right]=  \sup_{{P} \in {\mathcal{P}}} \mathbb{E}^{{P}} \left[ \vert Y \vert \right]= \mathcal{E}(\vert Y \vert ) < \infty,
	 \end{align*}
	 which proves that $Y \textbf{1}_{\lbrace \tau_2 > T\rbrace}$ is in $L^1(\tilde{\Omega})$.
	From Definition \ref{def:defiSublinearOperatorMulti} we get
	\begin{align}
		\tilde{\mathcal{E}}_t^2\left(\textbf{1}_{\lbrace \tau_2 > T \rbrace}Y\right)&= \textbf{1}_{\lbrace \tau_1 > t \rbrace} \mathcal{E}_t \left( e^{\int_0^t \tilde{\lambda}^1_s ds} \mathbb{E}^{\hat{P}}\left[ \textbf{1}_{\lbrace \tau_1 > t \rbrace}  \textbf{1}_{\lbrace \tau_2 >T \rbrace}Y \right]\right) \nonumber \\
		&\ + \textbf{1}_{\lbrace \tau_1 \leq  t < \tau_2 \rbrace} \mathcal{E}_t \left( e^{\int_0^{t-u_1} \tilde{\lambda}^2_s ds} \mathbb{E}^{\hat{P}}\left[  \textbf{1}_{\lbrace \tilde{\tau}_2 >T-u_1 \rbrace}Y \right]\right) \bigg \vert_{u_1=\tau_1}\label{eq:FirstPayoffBreak}
	\end{align}
	 for $t \in [0,T]$. It holds
	\begin{align}
		\mathbb{E} ^{\hat{P}}\left[ \textbf{1}_{\lbrace \tau_1 > t \rbrace}  \textbf{1}_{\lbrace \tau_2 >T \rbrace}Y \right]&=Y\mathbb{E}^{\hat{P}}\left[ \textbf{1}_{\lbrace \tau_1 > t \rbrace}  \textbf{1}_{\lbrace \tau_2 >T \rbrace}\right] \nonumber\\
		&=Y\left(\mathbb{E}^{\hat{P}}\left[ \textbf{1}_{\lbrace t < \tau_1 \leq T  \rbrace}  \textbf{1}_{\lbrace \tau_2 >T \rbrace}\right] + \mathbb{E}^{\hat{P}}\left[ \textbf{1}_{\lbrace \tau_1 > T  \rbrace}  \textbf{1}_{\lbrace \tau_2 >T \rbrace}\right]\right) \notag \\
		 &=Y\left(\mathbb{E} ^{\hat{P}}\left[ \mathbb{E}^{\hat{P}}\left[ \textbf{1}_{\lbrace  t < u_1 \leq T\rbrace}  \textbf{1}_{\lbrace \tilde{\tau}_2 >T-u_1 \rbrace}\right]\big \vert_{u_1=\tau_1} \right] + \mathbb{E}^{\hat{P}}\left[ \textbf{1}_{\lbrace \tau_1 > T  \rbrace} \right]\right) \label{eq:FirstPayoffApplicationLemmaA} \\
		  &=Y\left(\mathbb{E}^{\hat{P}}\left[  \textbf{1}_{\lbrace t < \tau_1 \leq T \rbrace}  e^{-\int_0^{T-\tau_1} \tilde{\lambda}_s^2 ds} \right]+e^{-\int_0^{T}\tilde{\lambda}_s^1 ds}\right)  \notag \\
		  &=Y\left(\int_t^T e^{-\int_0^{T-x} \tilde{\lambda}_s^2 ds}e^{-\int_0^x \tilde{\lambda}_y^1 dy}\tilde{\lambda}_x^1dx+e^{-\int_0^{T}\tilde{\lambda}_s^1 ds}\right), \quad 0 \le t \le T. \label{eq:FirstPayoffApplicationLemma}
	\end{align}
Here we use Lemma \ref{lem:secondlemmasecondpiece} in \eqref{eq:FirstPayoffApplicationLemmaA} and Lemma \ref{lemma:ExpectationCalculations} in \eqref{eq:FirstPayoffApplicationLemma} with $h(x, \omega)=e^{-\int_0^{T-x} \tilde{\lambda}^2_s(\omega) ds}$.

Moreover, for any fixed $u_1>0$ and  $t \in [0,T]$ we have 
 \begin{equation}\label{eq:exphatPsecond}
 Y \mathbb{E}^{\hat{P}}\left[ \textbf{1}_{\lbrace \tilde{\tau}_2 >T-u_1 \rbrace} \right]= Y e^{-\int_0^{T-\tau_1} \tilde{\lambda}_s^2 ds} .
 \end{equation} 
 Then the result follows by putting together \eqref{eq:FirstPayoffBreak}, \eqref{eq:FirstPayoffApplicationLemma} and \eqref{eq:exphatPsecond}.
	\end{proof}
\begin{remark}
	\begin{enumerate}
	\item \red{Under suitable conditions on the process $\tilde{\lambda}^1$ it is possible to compute \eqref{eq:FirstPayoffi=1}. For example by following \cite{bo_2020} and \cite{fns_2019}, we can assume that $\tilde{\lambda}^1$ is given by a non-linear affine process and compute \eqref{eq:FirstPayoffi=1} via generalized Riccati equations.}
	\item \red{Furthermore, it is also possible to include non-linear affine intensities into the self-exciting framework we consider in Example \ref{example:SelfExciting}. To do so, following the same approach as in \cite{bo_2020} we assume that the process $\mu$ in \eqref{eq:TildeLambdaSelfExciting} is a non-linear affine CIR model in order to guarantee that $\mu$ is positive. For a detailed study of the CIR model under parameter uncertainty we refer to \cite{fns_2019}.}
	\end{enumerate}
\end{remark}
We now turn to the case when the default happens before time $T$. We start with the following lemma.
\begin{lemma} \label{lemma:SecondPayOffFixedProbability}
	Let $Z:=(Z_t)_{t \in [0,T]}$ be an $\mathbb{F}^{\mathcal{P}}$-predictable and non-negative or bounded process on $\Omega$, and fix a measure $\tilde{P} \in \tilde{\mathcal{P}}$ such that $\tilde{P}=P \otimes \hat{P}$ for a probability measure $P \in \mathcal{P}$. Then we have
	\begin{equation} \label{eq:secondPayoffi=1}
		\mathbb{E}^{\tilde{P}}\left[ \textbf{1}_{\lbrace t < \tau_1 \leq s \rbrace} Z_{\tau_1} \big \vert \mathcal{G}_t^{(2)}\right]= \textbf{1}_{\lbrace \tau_1 > t \rbrace} \mathbb{E}^{{{P}}}\left[ \int_t^s Z_u e^{-\int_t^u {\tilde{\lambda}_v^1} dv } {\tilde{\lambda}_u^1} du \bigg \vert \mathcal{F}_t \right] \quad \tilde{P} \text{-a.s.}
	\end{equation}
	and 
	\begin{align}
		&\mathbb{E}^{\tilde{P}}\left[ \textbf{1}_{\lbrace t < \tau_2 \leq s \rbrace} Z_{\tau_2} \big \vert \mathcal{G}_t^{(2)}\right]\notag \\&=
		\textbf{1}_{\lbrace \tau_1 > t \rbrace} e^{\int_0^t \tilde{\lambda}_v^1 dv }{\mathbb{E}}^{{P}}\left[\int_t^s \left(\int_{0}^{s-y} Z_{y+x}e^{-\int_0^x \tilde{
	\lambda}_w^2 dw }\tilde{\lambda}_x^2 dx\right)e^{-\int_0^y\tilde{\lambda}_w^1dw} \tilde{\lambda}_y^1dy\bigg\vert \mathcal{F}_t\right] \nonumber \\
		&\quad +\textbf{1}_{\lbrace \tau_1 \leq t < \tau_2 \rbrace} e^{\int_0^{t-\tau_1} \tilde{\lambda}_v^2 dv } 	\mathbb{E}^{{P}}\left[ \int_{t-u_1}^{s-u_1} Z_{u_1+x} e^{-\int_0^x \tilde{
	\lambda}_w^2 dw }\tilde{\lambda}_x^2 dx \bigg\vert \mathcal{F}_t \right]
\quad \tilde{P} \text{-a.s.}, \label{eq:secondPayoffi=2}
	\end{align}
	 for any $0 \leq t \leq s \leq T$. 
	\end{lemma}
\begin{proof}
	Fix $t \in [0,T]$. Equation \eqref{eq:secondPayoffi=1} follows as
	\begin{equation*}
		\mathbb{E}^{\tilde{P}}\left[ \textbf{1}_{\lbrace t < \tau_1 \leq s \rbrace} Z_{\tau_1} \big \vert \mathcal{G}_t^{(2)}\right]=\mathbb{E}^{\tilde{P}}\left[ \textbf{1}_{\lbrace t < \tau_1 \leq s \rbrace} Z_{\tau_1} \big \vert \mathcal{G}_t^{(1)}\right]
	\end{equation*} 
	and by Lemma 2.27 in \cite{bz_2019}. By Theorem \ref{prop:NTimesModelUncertainty} it holds 
	\begin{align}
		\mathbb{E}^{\tilde{P}}\left[ \textbf{1}_{\lbrace t < \tau_2 \leq s \rbrace} Z_{\tau_2} \big \vert \mathcal{G}_t^{(2)}\right]&=\textbf{1}_{\lbrace \tau_1 > t \rbrace} e^{\int_0^t \tilde{\lambda}_v^1 dv } {\mathbb{E}}^{{P}}\left[\mathbb{E}^{\hat{P}} \left[ \textbf{1}_{\lbrace t < \tau_1 \rbrace} \textbf{1}_{\lbrace t < \tau_2 \leq s \rbrace} Z_{\tau_2}\right] \vert \mathcal{F}_t\right] \nonumber\\
		&+\textbf{1}_{\lbrace  \tau_1 \leq t < \tau_2 \rbrace} e^{\int_0^{t-\tau_1} \tilde{\lambda}_v^2 dv} \mathbb{E}^{\tilde{P}} \left[ \textbf{1}_{\lbrace t < u_1 + \tilde{\tau}_2 \leq s\rbrace} Z_{u_1+\tilde{\tau}_2} \vert \mathcal{F}_t \right] \vert_{u_1=\tau_1}.\label{eq:secondPayoffi=2,1}
	\end{align} 
We have
\red{\small{
\begin{align}
	{\mathbb{E}}^{{P}} \left[\mathbb{E}^{\hat{P}} \left[ \textbf{1}_{\lbrace t < \tau_1 \rbrace} \textbf{1}_{\lbrace t < \tau_2 \leq s \rbrace} Z_{\tau_2}\right] \big \vert \mathcal{F}_t\right] 
	&
	={\mathbb{E}}^{{P}} \left[\mathbb{E}^{\hat{P}} \left[\textbf{1}_{\lbrace t < \tau_1 \leq s  \rbrace} \mathbb{E}^{\hat{P}}\left[ \textbf{1}_{\lbrace  \tilde{\tau}_2 \leq s-u_1 \rbrace} Z_{u_1+\tilde{\tau}_2}\right]\vert_{u_1=\tau_1}\right] \bigg\vert \mathcal{F}_t\right] \notag \\
	&={\mathbb{E}}^{{P}} \left[\mathbb{E}^{\hat{P}} \left[\textbf{1}_{\lbrace t < \tau_1 \leq s \rbrace} \int_{0}^{s-\tau_1} Z_{\tau_1+x}e^{-\int_0^x \tilde{
	\lambda}_w^2 dw }\tilde{\lambda}_x^2 dx\right] \bigg\vert \mathcal{F}_t\right] \label{eq:secondPayoffi=2,1b*} \\
	&={\mathbb{E}}^{{P}}\left[\int_t^s \left(\int_{0}^{s-y} Z_{y+x}e^{-\int_0^x \tilde{\lambda}_w^2 dw }\tilde{\lambda}_x^2 dx\right)e^{-\int_0^y\tilde{\lambda}_w^1dw} \tilde{\lambda}_y^1dy\bigg\vert \mathcal{F}_t\right]. \label{eq:secondPayoffi=2,1c}
	\end{align}}}
	
\red{We use twice Lemma \ref{lemma:ExpectationCalculations}, namely in \eqref{eq:secondPayoffi=2,1b*} and \eqref{eq:secondPayoffi=2,1c}. }
Analogously we can compute the second term in $\eqref{eq:secondPayoffi=2,1}$ by using  $\eqref{eq:ConditionalExpectationFrancescaYinglin}$ and Lemma \ref{lemma:ExpectationCalculations}.
\end{proof}

We are now ready to give the following proposition.
\begin{prop}
	Let $Z:=(Z_t)_{t \in [0,T]}$ be an $\mathbb{F}^{\mathcal{P}}$-predictable process on $\Omega$ such that the function $Z(t,\omega):=Z_t(\omega)$, $(t, \omega) \in [0,T] \times \Omega$, is upper semianalytic and there exists $M \in \mathbb{R}_+$ such that $\sup_{t \in [0,T]} \vert Z_t \vert < M$ $\mathcal{P}$-q.s..
	Then for all $0 \leq t \leq s \leq T$ the random variables
	\begin{gather*}
		\textbf{1}_{\lbrace t < \tau_1 \leq s \rbrace} Z_{\tau_1}, \quad  \int_t^s Z_u e^{-\int_t^u \tilde{\lambda}^1_v dv } \tilde{\lambda}^1_u du, \quad  \textbf{1}_{\lbrace t < \tau_2 \leq s \rbrace} Z_{\tau_2},  \\
		\int_t^s \left(\int_{0}^{s-y} Z_{y+x}e^{-\int_0^x \tilde{\lambda}_w^2 dw }\tilde{\lambda}_x^2 dx\right)e^{-\int_0^y\tilde{\lambda}_w^1dw} \tilde{\lambda}_y^1dy, \quad \int_{t-u_1}^{s-u_1} Z_{u_1+x} e^{-\int_0^x \tilde{\lambda}_w^2 dw }\tilde{\lambda}_x^2 dx
	\end{gather*}
	are upper semianalytic. Moreover, $\textbf{1}_{\lbrace t < \tau_1 \leq s \rbrace} Z_{\tau_1}$ and $\textbf{1}_{\lbrace t < \tau_2 \leq s \rbrace} Z_{\tau_2}$ belong to $L^1(\tilde{\Omega})$.
		If $\mathcal{P}$ satisfies Assumption \ref{assumptionnutzNew}, we have
	\begin{equation} \label{eq:secondPayoffi=1Operator}
		\tilde{\mathcal{E}}_t^2\left( \textbf{1}_{\lbrace t < \tau_1 \leq s \rbrace} Z_{\tau_1} \right)= \textbf{1}_{\lbrace \tau_1 > t \rbrace} \mathcal{E}_t\left( \int_t^s Z_u e^{-\int_t^u {\tilde{\lambda}_v^1} dv } {\tilde{\lambda}^1}_u du \right) \quad \tilde{P} \text{-a.s. for all } \tilde{P} \in \tilde{\mathcal{P}}
	\end{equation}  
	and 
	\begin{small}
	\begin{align} \label{eq:secondPayoffi=2Operator}
		\tilde{\mathcal{E}}_t^2\left( \textbf{1}_{\lbrace t < \tau_2 \leq s \rbrace} Z_{\tau_2} \right)&=\textbf{1}_{\lbrace \tau_1 > t \rbrace} e^{\int_0^t \tilde{\lambda}_v^1 dv }\mathcal{E}_t \left(\int_t^s \left(\int_{0}^{s-y} Z_{y+x}e^{-\int_0^x \tilde{\lambda}_w^2 dw }\tilde{\lambda}_x^2 dx\right)e^{-\int_0^y\tilde{\lambda}_w^1dw} \tilde{\lambda}_y^1dy\right) \nonumber \\
		&+ \textbf{1}_{\lbrace \tau_1 \leq t < \tau_2 \rbrace}  e^{\int_0^{t-\tau_1} \tilde{\lambda}_v^2 dv } \mathcal{E}_t\left(\int_{t-u_1}^{s-u_1} Z_{u_1+x} e^{-\int_0^x \tilde{\lambda}_w^2 dw }\tilde{\lambda}_x^2 dx\right)\bigg \vert_{u_1=\tau_1}, 
	\end{align}
	\end{small}
	$ \tilde{P} \text{-a.s. for all } \tilde{P} \in \tilde{\mathcal{P}},$ for all $0 \leq t \leq T$. 
\end{prop}

\begin{proof}
Note that $\textbf{1}_{\lbrace t < \tau_1 \leq s \rbrace} Z_{\tau_1}$ and $\int_t^s Z_u e^{-\int_t^u \tilde{\lambda}^1_v dv } \tilde{\lambda}^1_u du$ are upper semianalytic and in $L^1(\tilde{\Omega})$ by Corollary 2.28 in \cite{bz_2019}.
	Equation \eqref{eq:secondPayoffi=1Operator} follows by Corollary 2.28 in \cite{bz_2019}, as
	\begin{equation*}
		\tilde{\mathcal{E}}_t^2\left( \textbf{1}_{\lbrace t < \tau_1 \leq s \rbrace} Z_{\tau_1} \right)=\tilde{\mathcal{E}}_t^1\left( \textbf{1}_{\lbrace t < \tau_1 \leq s \rbrace} Z_{\tau_1} \right)
	\end{equation*}
	by Proposition \ref{prop:OperatorCoincides}. \\
 Points 5 and 6 in Lemma \ref{lemma:propertiesUpperSemianalytic} imply that
	\begin{equation*}
		\textbf{1}_{\lbrace t < \tau_2 \leq s \rbrace} Z_{\tau_2}, \quad 
		\int_t^s \left(\int_{0}^{s-y} Z_{y+x}e^{-\int_0^x \tilde{\lambda}_w^2 dw }\tilde{\lambda}_x^2 dx\right)e^{-\int_0^y\tilde{\lambda}_w^1dw} \tilde{\lambda}_y^1dy, \quad \int_{t-u_1}^{s-u_1} Z_{u_1+x} e^{-\int_0^x \tilde{\lambda}_w^2 dw }\tilde{\lambda}_x^2 dx
	\end{equation*} 
	are upper semianalytic. Clearly, $\textbf{1}_{\lbrace t < \tau_2 \leq s \rbrace} Z_{\tau_2}$ is in $L^1(\tilde{\Omega})$. Moreover, equation \eqref{eq:secondPayoffi=2Operator} follows by Lemma \ref{lemma:SecondPayOffFixedProbability}.
\end{proof}

\subsection{Sufficient conditions for the classical tower property}
In this section we show that the operator $\tilde{\mathcal{E}}_t^2$ applied to CCT$^{(i)}$ maps $L^1(\tilde{\Omega})$ on to $L^1(\tilde{\Omega})$, and derive some sufficient conditions under which the classical tower property holds for such payoffs. 
\begin{prop} \label{prop:PayoffClassicalTowerII}
	Let $Z:=(Z_t)_{t \in [0,T]}$ be an $\mathbb{F}^{\mathcal{P}}$-predictable non-negative process on $\Omega$ such that the function $Z(t,\omega):=Z_t(\omega)$, $(t, \omega) \in [0,T] \times \Omega$, is upper semianalytic and there exists $M \in \mathbb{R}_+$ such that $\sup_{t \in [0,T]} Z_t  < M$ $P$-a.s. for all $P \in \mathcal{P}$.  If $\mathcal{P}$ satisfies Assumption \ref{assumptionnutzNew}, then we have
	\begin{equation*}
		\tilde{\mathcal{E}}_t^2(\textbf{1}_{\lbrace 0< {\tau}_2 < T \rbrace}Z_{\tau_2}) \in L^1(\tilde{\Omega}).
	\end{equation*} 
	If in addition
	\begin{align}
		&\int_s^t \mathcal{E}_t  \left ( \int_{t-x}^{T-x} Z_{x+v}e^{-\int_0^{v} \tilde{\lambda}_w^2 dw} \tilde{\lambda}_{v}^2 dv\right)e^{-\int_0^x \tilde{\lambda}^1_y dy}\tilde{\lambda}^1_x dx \notag\\
	&=\mathcal{E}_t \left( \int_s^t \left( \int_{t-x}^{T-x} Z_{x+v}e^{-\int_0^{v} \tilde{\lambda}_w^2 dw} \tilde{\lambda}_{v}^2 dv\right)e^{-\int_0^x \tilde{\lambda}^1_y dy}\tilde{\lambda}^1_x dx\right)\label{eq:Payoff2ConditionI}
	\end{align}
	and
	\small{
	\begin{align}
		&\mathcal{E}_t\left(\int_t^T\left( \int_0^{T-x} Z_{v+x}e^{-\int_0^{v} \tilde{\lambda}_w^2 dw }\tilde{\lambda}_{v}^2 dv\right)e^{-\int_0^x \tilde{\lambda}_y^1 dy}\tilde{\lambda}_x^1 dx \right) \notag\\ 
		&\quad + \mathcal{E}_t \left(\int_s^t \left( \int_{t-x}^{T-x} Z_{v+x} e^{-\int_0^{v} \tilde{\lambda}_w^2 dw }\tilde{\lambda}_{v}^2 dv \right) e^{-\int_0^x \tilde{\lambda}_y^1dy}\tilde{\lambda}_x^1dx \right) \nonumber \\
		&=\mathcal{E}_t \Bigg( \int_t^T\left( \int_0^{T-x} Z_{v+x}e^{-\int_0^{v} \tilde{\lambda}_w^2 dw }\tilde{\lambda}_{v}^2 dv\right)e^{-\int_0^x \tilde{\lambda}_y^1 dy}\tilde{\lambda}_x^1 dx  \notag \\ 
		& \qquad \quad + \int_t^T\left( \int_0^{T-x} Z_{v+x}e^{-\int_0^{v} \tilde{\lambda}_w^2 dw }\tilde{\lambda}_{v}^2 dv\right)e^{-\int_0^x \tilde{\lambda}_y^1 dy}\tilde{\lambda}_x^1 dx  \Bigg) \label{eq:Payoff2ConditionII}
	\end{align}}
	for all $0 \leq s \leq t \leq T$, then  
	\begin{equation}\label{eq:towerprop}
	\tilde{\mathcal{E}}_s^2\left(\tilde{\mathcal{E}}_t^2 \left(\textbf{1}_{\lbrace 0< {\tau}_2 < T \rbrace}Z_{\tau_2}\right)\right)	= \tilde{\mathcal{E}}_s^2\left( \textbf{1}_{\lbrace 0< {\tau}_2 < T \rbrace}Z_{\tau_2}\right)	 \quad \tilde{P} \text{-a.s. for all } \tilde{P} \in \tilde{\mathcal{P}}
	\end{equation}
	for all $0 \leq s \leq t \leq T$. 
\end{prop}

\begin{proof}
	Let $t \in [0,T]$. We start by proving that
\begin{equation} \label{eq:Payoff2GoalIntegrability}
	\tilde{\mathcal{E}}\left(  \tilde{\mathcal{E}}_t^2\left( \textbf{1}_{\lbrace 0< {\tau}_2 < T \rbrace}Z_{\tau_2}  \right)  \right) < \infty.
\end{equation}
We have
\begin{align}
	\tilde{\mathcal{E}}\left(  \tilde{\mathcal{E}}^2_t(\textbf{1}_{\lbrace 0< {\tau}_2 < T \rbrace}Z_{\tau_2}) \right)
	&=\sup_{\tilde{P} \in \tilde{\mathcal{P}}} \mathbb{E}^{\tilde{P}} \left[ \vert \tilde{\mathcal{E}}^2_t( \textbf{1}_{\lbrace 0< {\tau}_2 < T \rbrace}Z_{\tau_2})\vert \right] \notag \\
	&\leq \sup_{\tilde{P} \in \tilde{\mathcal{P}}} \mathbb{E}^{\tilde{P}} \left[  \tilde{\mathcal{E}}^2_t( \vert Z_{\tau_2}\vert) \right] \notag \\
	&\leq \sup_{\tilde{P} \in \tilde{\mathcal{P}}} \mathbb{E}^{\tilde{P}} \left[  \tilde{\mathcal{E}}^2_t( M) \right] \notag \\
	&=M < \infty.
	\end{align}
To derive the classical tower property for the payoff function $\textbf{1}_{\lbrace 0 < \tau_2 \leq T\rbrace }Z_{\tau_2}$ we need to prove that in this case the inequalities which we use in Section \ref{section:DynamicProgramming} to derive the weak dynamic programming principle are indeed equalities. Let $0 \leq s \leq t \leq T$. First, we prove that \eqref{eq:WeakDynamicProgrammingThird7} in Lemma \ref{lem:secondtermsecondpiece} is indeed an equality, i.e.,
\begin{align}
	&\mathbb{E}^{\hat{P}} \left[ \mathcal{E}_t \left( \textbf{1}_{\lbrace s < u_1 \leq t\rbrace} \textbf{1}_{\lbrace \tilde{\tau}_2 > t-u_1\rbrace}  e^{\int_0^{t-u_1} \tilde{\lambda}_v^2 dv } \mathbb{E}^{\hat{P}} \left[ \textbf{1}_{\lbrace \tilde{\tau}_2>t-u_1 \rbrace} \textbf{1}_{\lbrace 0 < u_1 + \tilde{\tau}_2 \leq T \rbrace} Z_{u_1+\tilde{\tau}_2}\right]\right) \right] \notag \\
	&=\mathcal{E}_t \left( \mathbb{E}^{\hat{P}} \left[ \textbf{1}_{\lbrace s < u_1 \leq t\rbrace} \textbf{1}_{\lbrace \tilde{\tau}_2 > t-u_1\rbrace}  e^{\int_0^{t-u_1} \tilde{\lambda}_v^2 dv } \mathbb{E}^{\hat{P}} \left[ \textbf{1}_{\lbrace \tilde{\tau}_2>t-u_1 \rbrace} \textbf{1}_{\lbrace 0 < u_1 + \tilde{\tau}_2 \leq T \rbrace} Z_{u_1+\tilde{\tau}_2}\right]\right] \right). \label{eq:PayoffII1}
\end{align}
By using Lemma \ref{lemma:ExpectationCalculations} we can rewrite the left-hand side of \eqref{eq:PayoffII1} as
\red{\begin{align}
	&\mathbb{E}^{\hat{P}} \left[ \mathcal{E}_t \left( \textbf{1}_{\lbrace s < u_1 \leq t\rbrace} \textbf{1}_{\lbrace \tilde{\tau}_2 > t-u_1\rbrace}  e^{\int_0^{t-u_1} \tilde{\lambda}_v^2 dv } \mathbb{E}^{\hat{P}} \left[ \textbf{1}_{\lbrace \tilde{\tau}_2>t-u_1 \rbrace} \textbf{1}_{\lbrace 0 < u_1 + \tilde{\tau}_2 \leq T \rbrace} Z_{u_1+\tilde{\tau}_2}\right]\right) \right] \notag \\
	&=\mathbb{E}^{\hat{P}} \left[ \textbf{1}_{\lbrace \tilde{\tau}_2 > t-u_1\rbrace}   \right] \mathcal{E}_t \left( \textbf{1}_{\lbrace s < u_1 \leq t\rbrace}  e^{\int_0^{t-u_1} \tilde{\lambda}_v^2 dv }  \int_{t-u_1}^{T-u_1} Z_{u_1+v}e^{-\int_0^{v} \tilde{\lambda}_w^2 dw} \tilde{\lambda}_{v}^2 dv \right) \label{eq:PayoffII3_2} \\
	&=\mathcal{E}_t \left( \textbf{1}_{\lbrace s < u_1 \leq t\rbrace} \int_{t-u_1}^{T-u_1} Z_{u_1+v}e^{-\int_0^{v} \tilde{\lambda}_w^2 dw} \tilde{\lambda}_{v}^2 dv\right),\label{eq:PayoffII3} 
\end{align}}
where \eqref{eq:PayoffII3_2} follows since $\mathcal{E}_t \left( \textbf{1}_{\lbrace s < u_1 \leq t\rbrace}  e^{\int_0^{t-u_1} \tilde{\lambda}_v^2 dv }  \int_{t-u_1}^{T-u_1} Z_{u_1+v}e^{-\int_0^{v} \tilde{\lambda}_w^2 dw} \tilde{\lambda}_{v}^2 dv \right)$ is independent of $\hat{\omega}$. We now turn to the right-hand side of \eqref{eq:PayoffII1}. By Lemma \ref{lemma:ExpectationCalculations} and similar calculations as for the left-hand side of \eqref{eq:PayoffII1}, we have that 
\red{\begin{align}
	&\mathcal{E}_t \left( \mathbb{E}^{\hat{P}} \left[ \textbf{1}_{\lbrace s < u_1 \leq t\rbrace} \textbf{1}_{\lbrace \tilde{\tau}_2 > t-u_1\rbrace}  e^{\int_0^{t-u_1} \tilde{\lambda}_v^2 dv } \mathbb{E}^{\hat{P}} \left[ \textbf{1}_{\lbrace \tilde{\tau}_2>t-u_1 \rbrace} \textbf{1}_{\lbrace 0 < u_1 + \tilde{\tau}_2 \leq T \rbrace} Z_{u_1+\tilde{\tau}_2}\right]\right] \right) \notag \\
	&=\mathcal{E}_t \left(\textbf{1}_{\lbrace s < u_1 \leq t\rbrace} \int_{t-u_1}^{T-u_1} Z_{u_1+v}e^{-\int_0^{v} \tilde{\lambda}_w^2 dw} \tilde{\lambda}_{v}^2 dv \right) \label{eq:PayoffII4}. 
\end{align}}
At this point, equation \eqref{eq:PayoffII1} follows directly by \eqref{eq:PayoffII3} and \eqref{eq:PayoffII4}. 
Coming now to equation \eqref{eq:lemma2}, from Lemma \ref{lemma:ExpectationCalculations} it can be seen that
\begin{align}
	&\mathbb{E}^{\hat{P}} \left[ \mathcal{E}_t \left( \mathbb{E}^{\hat{P}} \left[ \textbf{1}_{\lbrace s < u_1 \leq t\rbrace} \textbf{1}_{\lbrace \tilde{\tau}_2 > t-u_1\rbrace}\textbf{1}_{\lbrace 0 < u_1 + \tilde{\tau}_2 \leq T \rbrace} Z_{u_1+\tilde{\tau}_2} \right]\right) \bigg\vert_{u_1=\tau_1}\right] \notag \\
	&=\mathcal{E}_t \left( \mathbb{E}^{\hat{P}} \left[ \mathbb{E}^{\hat{P}} \left[ \textbf{1}_{\lbrace s < u_1 \leq t\rbrace} \textbf{1}_{\lbrace \tilde{\tau}_2 > t-u_1\rbrace}\textbf{1}_{\lbrace 0 < u_1 + \tilde{\tau}_2 \leq T \rbrace} Z_{u_1+\tilde{\tau}_2} \right] \big\vert_{u_1=\tau_1}\right] \right) \label{eq:PayoffII2}
\end{align}
is equivalent to  
\small{
\begin{align*}
	&\mathbb{E}^{\hat{P}} \left[ \textbf{1}_{\lbrace s < \tau_1 \leq t \rbrace} \mathcal{E}_t  \left ( \int_{t-u_1}^{T-u_1} Z_{u_1+v}e^{-\int_0^{v} \tilde{\lambda}_w^2 dw} \tilde{\lambda}_{v}^2 dv\right)\bigg\vert_{u_1=\tau_1}\right]\notag \\ 
	&\qquad =\mathcal{E}_t\left( \mathbb{E}^{\hat{P}} \left [ \textbf{1}_{\lbrace s < \tau_1 \leq t \rbrace} \int_{t-u_1}^{T-u_1} Z_{u_1+v}e^{-\int_0^{v} \tilde{\lambda}_w^2 dw} \tilde{\lambda}_{v}^2 dv\right]\right),
\end{align*}}
\\which is exactly condition \eqref{eq:Payoff2ConditionI} by using Lemma \ref{lemma:ExpectationCalculations}.
Next, we show that the inequalities 
\eqref{eq:SecondInequality} and \eqref{eq:ThirdInequality} are indeed equalities, i.e.,
\begin{align}
	\mathbb{E}^{\hat{P}} \left[ \mathcal{E}_t \left ( \textbf{1}_{\lbrace s < \tau_1\rbrace}  \textbf{1}_{\lbrace \tau_2 \leq t  \rbrace} \textbf{1}_{\lbrace 0 < \tau_2 \leq T  \rbrace} Z_{\tau_2}\right) \right]= \mathcal{E}_t \left( \mathbb{E}^{\hat{P}}  \left [ \textbf{1}_{\lbrace s < \tau_1\rbrace}  \textbf{1}_{\lbrace \tau_2 \leq t  \rbrace} \textbf{1}_{\lbrace 0 < \tau_2 \leq T  \rbrace} Z_{\tau_2}\right] \right) \label{eq:PayoffII5}
\end{align}
and
\begin{align}
&\mathcal{E}_t\left( \mathbb{E}^{\hat{P}} \left[ \textbf{1}_{\lbrace t < \tau_1\rbrace} \textbf{1}_{\lbrace 0  < \tau_2 \leq T \rbrace}Z_{\tau_2}\right]\right)	+ \mathcal{E}_t\left( \mathbb{E}^{\hat{P}} \left[ \textbf{1}_{\lbrace s < \tau_1\rbrace}  \textbf{1}_{\lbrace \tau_1  \leq t < \tau_2  \rbrace}\textbf{1}_{\lbrace 0  < \tau_2 \leq T \rbrace}Z_{\tau_2}\right]\right) \notag \\
&+  \mathcal{E}_t\left( \mathbb{E}^{\hat{P}} \left[ \textbf{1}_{\lbrace s < \tau_1\rbrace}  \textbf{1}_{\lbrace \tau_1  \leq t < \tau_2  \rbrace}\textbf{1}_{\lbrace 0  < \tau_2 \leq T \rbrace}Z_{\tau_2}\right]\right) \notag \\
&= \mathcal{E}_t \bigg(\mathbb{E}^{\hat{P}} \left[ \textbf{1}_{\lbrace t < \tau_1\rbrace} \textbf{1}_{\lbrace 0  < \tau_2 \leq T \rbrace}Z_{\tau_2}\right]+ \mathbb{E}^{\hat{P}} \left[ \textbf{1}_{\lbrace s < \tau_1\rbrace}  \textbf{1}_{\lbrace \tau_1  \leq t < \tau_2  \rbrace}\textbf{1}_{\lbrace 0  < \tau_2 \leq T \rbrace}Z_{\tau_2}\right] \notag\\
&+\mathbb{E}^{\hat{P}} \left[ \textbf{1}_{\lbrace s < \tau_1\rbrace}  \textbf{1}_{\lbrace \tau_1  \leq t < \tau_2  \rbrace}\textbf{1}_{\lbrace 0  < \tau_2 \leq T \rbrace}Z_{\tau_2}\right] \bigg). \label{eq:PayoffII6}
\end{align}
We first consider the right-hand side of \eqref{eq:PayoffII5}. Similar arguments used in the derivation of \eqref{eq:secondPayoffi=2,1c} imply that
\begin{align}
	\mathcal{E}_t\left(\mathbb{E}^{\hat{P}}  \left [ \textbf{1}_{\lbrace s < \tau_1\rbrace}  \textbf{1}_{\lbrace \tau_2 \leq t  \rbrace} \textbf{1}_{\lbrace 0 < \tau_2 \leq T  \rbrace} Z_{\tau_2}\right]\right)	&=\mathcal{E}_t \left( \int_s^t \left( \int_0^{t-x} Z_{v+x} e^{-\int_0^{v} \tilde{\lambda}_w^2 dw} \tilde{\lambda}_{v}^2 dv\right)e^{-\int_0^x \tilde{\lambda}_y^1 dy} \tilde{\lambda}_x^1 dx\right)  \label{eq:PayoffII7} \\
	&=\int_s^t \left( \int_0^{t-x} Z_{v+x} e^{-\int_0^{v} \tilde{\lambda}_w^2 dw} \tilde{\lambda}_{v}^2 dv\right)e^{-\int_0^x \tilde{\lambda}_y^1 dy} \tilde{\lambda}_x^1 dx \label{eq:PayoffII8},
\end{align}
since $ \int_s^t \left( \int_0^{t-x} Z_{v+x} e^{-\int_0^{v} \tilde{\lambda}_w^2 dw} \tilde{\lambda}_{v}^2 dv\right)e^{-\int_0^x \tilde{\lambda}_y^1 dy} \tilde{\lambda}_x^1 dx$ is non-negative and $\mathcal{F}_t$-measurable.
	
Moreover, note that for fixed $\omega \in \hat{\Omega}$ the random variable $\textbf{1}_{\lbrace s < \tau_1\rbrace}  \textbf{1}_{\lbrace \tau_2 \leq t  \rbrace} \textbf{1}_{\lbrace 0 < \tau_2 \leq T  \rbrace} Z_{\tau_2}$ is $\mathcal{F}_t$-measurable and non-negative. For this reason, the left-hand side of \eqref{eq:PayoffII5} can be rewritten as
\begin{align}
	\mathbb{E}^{\hat{P}} \left[ \mathcal{E}_t \left ( \textbf{1}_{\lbrace s < \tau_1\rbrace}  \textbf{1}_{\lbrace \tau_2 \leq t  \rbrace} \textbf{1}_{\lbrace 0 < \tau_2 \leq T  \rbrace} Z_{\tau_2}\right) \right]&= \mathbb{E}^{\hat{P}} \left[ \textbf{1}_{\lbrace s < \tau_1\rbrace}  \textbf{1}_{\lbrace \tau_2 \leq t  \rbrace} \textbf{1}_{\lbrace 0 < \tau_2 \leq T  \rbrace} Z_{\tau_2}\right] \notag \\
	&= \int_s^t \left( \int_0^{t-x} Z_{v+x} e^{-\int_0^{v} \tilde{\lambda}_w^2 dw} \tilde{\lambda}_{v}^2 dv\right)e^{-\int_0^x \tilde{\lambda}_y^1 dy} \tilde{\lambda}_x^1 dx,\label{eq:PayoffII9}
\end{align}
where \eqref{eq:PayoffII9} follows by \eqref{eq:PayoffII7}. We now get \eqref{eq:PayoffII5} directly from \eqref{eq:PayoffII8} and \eqref{eq:PayoffII9}.

We now turn to \eqref{eq:PayoffII6}. The left-hand side terms can be written with similar arguments used as in \eqref{eq:secondPayoffi=2,1c} as
\begin{align}
	\mathbb{E}^{\hat{P}} \left[ \textbf{1}_{\lbrace t < \tau_1\rbrace}  \textbf{1}_{\lbrace 0 < \tau_2 \leq T \rbrace} Z_{\tau_2}\right]	&=\int_t^T\left( \int_0^{T-x} Z_{v+x}e^{-\int_0^{v} \tilde{\lambda}_w^2 dw }\tilde{\lambda}_{v}^2 dv\right)e^{-\int_0^x \tilde{\lambda}_y^1 dy}\tilde{\lambda}_x^1 dx, \label{eq:PayoffII10}
\end{align}
\begin{align}
\mathbb{E}^{\hat{P}} \left[ \textbf{1}_{\lbrace s < \tau_1 \leq t \rbrace} \textbf{1}_{\lbrace t < \tau_2 \leq T \rbrace} Z_{\tau_2} \right]	 &=\int_s^t \left( \int_{t-x}^{T-x} Z_{v+x} e^{-\int_0^{v} \tilde{\lambda}_w^2 dw }\tilde{\lambda}_{v}^2 dv \right) e^{-\int_0^x \tilde{\lambda}_y^1dy}\tilde{\lambda}_x^1dx\label{eq:PayoffII11}
\end{align}
and 
\begin{align}
	\mathbb{E}^{\hat{P}} \left[ \textbf{1}_{\lbrace s < \tau_1\rbrace} \textbf{1}_{\lbrace  \tau_2 \leq t \rbrace} Z_{\tau_2}\right]	&= \int_s^t \left( \int_0^{t-x} Z_{v+x} e^{-\int_0^{v} \tilde{\lambda}_w^2 dw} \tilde{\lambda}_{v}^2 dv\right)e^{-\int_0^x \tilde{\lambda}_y^1 dy} \tilde{\lambda}_x^1 dx. \label{eq:PayoffII12}
\end{align}
By \eqref{eq:PayoffII8}, \eqref{eq:PayoffII10}, \eqref{eq:PayoffII11},  \eqref{eq:PayoffII12} and assumption \eqref{eq:Payoff2ConditionII}
we get
\red{\begin{align}
&\mathcal{E}_t\left( \mathbb{E}^{\hat{P}} \left[ \textbf{1}_{\lbrace t < \tau_1\rbrace} \textbf{1}_{\lbrace 0  < \tau_2 \leq T \rbrace}Z_{\tau_2}\right]\right)	+ \mathcal{E}_t\left( \mathbb{E}^{\hat{P}} \left[ \textbf{1}_{\lbrace s < \tau_1\rbrace}  \textbf{1}_{\lbrace \tau_1  \leq t < \tau_2  \rbrace}\textbf{1}_{\lbrace 0  < \tau_2 \leq T \rbrace}Z_{\tau_2}\right]\right) \notag \\
&\quad +  \mathcal{E}_t\left( \mathbb{E}^{\hat{P}} \left[ \textbf{1}_{\lbrace s < \tau_1\rbrace}  \textbf{1}_{\lbrace \tau_1  \leq t < \tau_2  \rbrace}\textbf{1}_{\lbrace 0  < \tau_2 \leq T \rbrace}Z_{\tau_2}\right]\right) \notag \\
&=\mathcal{E}_t\bigg( \int_t^T\left( \int_0^{T-x} Z_{v+x}e^{-\int_0^{v} \tilde{\lambda}_w^2 dw }\tilde{\lambda}_{v}^2 dv\right)e^{-\int_0^x \tilde{\lambda}_y^1 dy}\tilde{\lambda}_x^1 dx \notag \\ & \qquad \quad  + \int_s^t \left( \int_{t-x}^{T-x} Z_{v+x} e^{-\int_0^{v} \tilde{\lambda}_w^2 dw }\tilde{\lambda}_{v}^2 dv \right) e^{-\int_0^x \tilde{\lambda}_y^1dy}\tilde{\lambda}_x^1dx \notag \\
&\quad  \qquad + \int_s^t \left( \int_0^{t-x} Z_{v+x} e^{-\int_0^{v} \tilde{\lambda}_w^2 dw} \tilde{\lambda}_{v}^2 dv\right)e^{-\int_0^x \tilde{\lambda}_y^1 dy} \tilde{\lambda}_x^1 dx \bigg). \notag
\end{align}}
Here, the last equality is implied by the $\mathcal{F}_t \text{-measurability}$ of $ \int_s^t \left( \int_0^{t-x} Z_{v+x} e^{-\int_0^{v} \tilde{\lambda}_w^2 dw} \tilde{\lambda}_{v}^2 dv\right)e^{-\int_0^x \tilde{\lambda}_y^1 dy} \tilde{\lambda}_x^1 dx$. This proves equation \eqref{eq:PayoffII6}. 

In order to conclude the proof, we need to show that inequalities \eqref{eq:WeakDynamicMiddleTerm3} and \eqref{eq:WeakDynamicMiddleTerm4} in the proof of Proposition \ref{prop:dynamicthird} are indeed equalities, that is,
\begin{align} \label{eq:PayoffII16}
	\mathbb{E}^{\hat{P}} \left[ \mathcal{E}_t \left( \textbf{1}_{\lbrace \tilde{\tau}_2 > s-u_1 \rbrace} \textbf{1}_{\lbrace u_1+ \tilde{\tau}_2 \leq t  \rbrace} \textbf{1}_{\lbrace u_1+ \tilde{\tau}_2 \leq T  \rbrace} Z_{u_1+\tilde{\tau}_2} \right)\right]= \mathcal{E}_t\left( \mathbb{E}^{\hat{P}} \left[ \textbf{1}_{\lbrace \tilde{\tau}_2 > s-u_1 \rbrace} \textbf{1}_{\lbrace u_1+ \tilde{\tau}_2 \leq t  \rbrace} \textbf{1}_{\lbrace u_1+ \tilde{\tau}_2 \leq T  \rbrace} Z_{u_1+\tilde{\tau}_2} \right]\right)
\end{align}
and 
\begin{align} 
	&\mathcal{E}_t \left( \textbf{1}_{\lbrace u_1 \leq t \rbrace} \mathbb{E}^{\hat{P}} \left[  \textbf{1}_{\lbrace \tilde{\tau}_2 >t-u_1 \rbrace} \textbf{1}_{\lbrace \tilde{\tau}_2 +u_1 \leq T \rbrace} Z_{u_1 + \tilde{\tau}_2} \right]\right) +\mathcal{E}_t \left(  \mathbb{E}^{\hat{P}} \left[  \textbf{1}_{\lbrace \tilde{\tau}_2 >s-u_1 \rbrace} \textbf{1}_{\lbrace \tilde{\tau}_2 +u_1 \leq t\rbrace} \textbf{1}_{\lbrace \tilde{\tau}_2 +u_1 \leq T \rbrace} Z_{u_1 + \tilde{\tau}_2} \right]\right)  \notag\\
	&=\mathcal{E}_t \left( \textbf{1}_{\lbrace u_1 \leq t \rbrace} \mathbb{E}^{\hat{P}} \left[  \textbf{1}_{\lbrace \tilde{\tau}_2 >t-u_1 \rbrace} \textbf{1}_{\lbrace \tilde{\tau}_2 +u_1 \leq T \rbrace} Z_{u_1 + \tilde{\tau}_2} \right]+  \mathbb{E}^{\hat{P}} \left[  \textbf{1}_{\lbrace \tilde{\tau}_2 >s-u_1 \rbrace} \textbf{1}_{\lbrace \tilde{\tau}_2 +u_1 \leq t\rbrace} \textbf{1}_{\lbrace \tilde{\tau}_2 +u_1 \leq T \rbrace} Z_{u_1 + \tilde{\tau}_2} \right]\right). \label{eq:PayoffII17}
\end{align}
By Lemma \ref{lemma:ExpectationCalculations} we have
\begin{align}
	\mathcal{E}_t\left( \mathbb{E}^{\hat{P}} \left[ \textbf{1}_{\lbrace \tilde{\tau}_2 > s-u_1 \rbrace} \textbf{1}_{\lbrace u_1+ \tilde{\tau}_2 \leq t  \rbrace} \textbf{1}_{\lbrace u_1+ \tilde{\tau}_2 \leq T  \rbrace} Z_{u_1+\tilde{\tau}_2} \right]\right)	&= \mathcal{E}_t\left( \int_{s-u_1}^{t-u_1} Z_{v+u_1} e^{-\int_0^{v} \tilde{\lambda}_w^2 dw} \tilde{\lambda}_{v}^2dv\right) \notag  \\
	&=  \int_{s-u_1}^{t-u_1} Z_{v+u_1} e^{-\int_0^{v} \tilde{\lambda}_w^2 dw} \tilde{\lambda}_{v}^2dv, \label{eq:PayoffII18}
\end{align}
where \eqref{eq:PayoffII18} follows since $\int_{s-u_1}^{t-u_1} Z_{v+u_1} e^{-\int_0^{v} \tilde{\lambda}_w^2 dw} \tilde{\lambda}_{v}^2dv$ is non-negative and $\mathcal{F}_t$-measurable.
By similar arguments as in \eqref{eq:PayoffII9} it follows
\begin{align}
	\mathbb{E}^{\hat{P}} \left[ \mathcal{E}_t \left( \textbf{1}_{\lbrace \tilde{\tau}_2 > s-u_1 \rbrace} \textbf{1}_{\lbrace u_1+ \tilde{\tau}_2 \leq t  \rbrace} \textbf{1}_{\lbrace u_1+ \tilde{\tau}_2 \leq T  \rbrace} Z_{u_1+\tilde{\tau}_2} \right)\right]	&=  \int_{s-u_1}^{t-u_1} Z_{v+u_1} e^{-\int_0^{v} \tilde{\lambda}_w^2 dw} \tilde{\lambda}_{v}^2dv \label{eq:PayoffII19}, 
\end{align}
which implies together with \eqref{eq:PayoffII18} that \eqref{eq:PayoffII16} holds. Moreover, by \eqref{eq:PayoffII18} and the $\mathcal{F}_t$-measurability of $  \int_{s-u_1}^{t-u_1} Z_{v+u_1} e^{-\int_0^{v} \tilde{\lambda}_w^2 dw} \tilde{\lambda}_{v}^2dv$, equation \eqref{eq:PayoffII17} follows. 
\end{proof}

\begin{example}\label{example:PayoffII}
Fix $0 \leq s \leq t\leq T$ such that $T \leq t+s$. Moreover, let the process $Z=(Z_t)_{t \in [0,T]}$ be deterministic and non-negative. We want to show that conditions \eqref{eq:Payoff2ConditionI} and \eqref{eq:Payoff2ConditionII} are satisfied. \\
	Note that 
	\begin{equation} \label{eq:PayoffIIExample1}
		 \int_s^t \left( \int_{t-x}^{T-x} Z_{x+v}e^{-\int_0^{v} \tilde{\lambda}_w^2 dw} \tilde{\lambda}_{v}^2 dv\right)e^{-\int_0^x \tilde{\lambda}^1_y dy}\tilde{\lambda}^1_x dx
	\end{equation}
	is $\mathcal{F}_t$-measurable and non-negative if {$T \leq t+s$}.
	Then we get that 
	\begin{align} 
	&\mathcal{E}_t\left(  \int_s^t \left( \int_{t-x}^{T-x} Z_{x+v}e^{-\int_0^{v} \tilde{\lambda}_w^2 dw} \tilde{\lambda}_{v}^2 dv\right)e^{-\int_0^x \tilde{\lambda}^1_y dy}\tilde{\lambda}^1_x dx\right) \notag\\ 
	&=  \int_s^t \left( \int_{t-x}^{T-x} Z_{x+v}e^{-\int_0^{v} \tilde{\lambda}_w^2 dw} \tilde{\lambda}_{v}^2 dv\right)e^{-\int_0^x \tilde{\lambda}^1_y dy}\tilde{\lambda}^1_x dx.\label{eq:PayoffIIExample2}
	\end{align}
	Moreover, we have 
	\begin{align}
	&\int_s^t \mathcal{E}_t  \left ( \int_{t-x}^{T-x} Z_{x+v}e^{-\int_0^{v} \tilde{\lambda}_w^2 dw} \tilde{\lambda}_{v}^2 dv\right) e^{-\int_0^x \tilde{\lambda}^1_y dy}\tilde{\lambda}^1_x dx \notag \\
	&=\int_s^t  \left(\int_{t-x}^{T-x} Z_{x+v}e^{-\int_0^{v} \tilde{\lambda}_w^2 dw} \tilde{\lambda}_{v}^2 dv\right) e^{-\int_0^x \tilde{\lambda}^1_y dy}\tilde{\lambda}^1_x dx,  \label{eq:PayoffIIExample3}
		\end{align}
	where \eqref{eq:PayoffIIExample3} holds because $ \int_{t-u_1}^{T-u_1} Z_{u_1+v}e^{-\int_0^{v} \tilde{\lambda}_w^2 dw} \tilde{\lambda}_{v}^2 dv$ is $\mathcal{F}_t$-measurable and non-negative for $\int_{t-u_1}^{T-u_1} Z_{u_1+v}e^{-\int_0^{v} \tilde{\lambda}_w^2 dw} \tilde{\lambda}_{v}^2 dv$. Condition \eqref{eq:Payoff2ConditionI} follows now directly from \eqref{eq:PayoffIIExample2} and \eqref{eq:PayoffIIExample3}. 
	
	Moreover, since $\int_s^t \left( \int_{t-x}^{T-x} Z_{x+v}e^{-\int_0^{v} \tilde{\lambda}_w^2 dw} \tilde{\lambda}_{v}^2 dv\right)e^{-\int_0^x \tilde{\lambda}^1_y dy}\tilde{\lambda}^1_x dx$ is
	 $\mathcal{F}_t$-measurable, equation \eqref{eq:PayoffIIExample2} implies that condition \eqref{eq:Payoff2ConditionII} holds.
\end{example}

\begin{prop} \label{prop:PayoffClassicalTowerI}
	Fix $0 \leq s \leq t \leq T$ and let $Y=Y(\omega)$, $\omega \in \Omega$, be an $\mathcal{F}_T^{\mathcal{P}}$-measurable upper semianalytic and non-negative function on $\Omega$ such that $\mathcal{E}( Y)< \infty$. If $\mathcal{P}$ satisfies Assumption \ref{assumptionnutzNew}, then $\tilde{\mathcal{E}}_t^2(\textbf{1}_{\lbrace \tilde{\tau}_2 >T \rbrace}Y) \in L^1(\tilde{\Omega})$ for any $t \in [0,T]$. If in addition
	\begin{equation} \label{eq:Payoff1Condition1}
		\int_s^t  \mathcal{E}_t \left(Y e^{-\int_0^{T-x} \tilde{\lambda}_v^2 dv }\right) e^{-\int_0^x \tilde{\lambda}_y^1 dy}\tilde{\lambda}_x^1dx= \mathcal{E}_t\left( \int_s^t Y e^{-\int_0^{T-x} \tilde{\lambda}_v^2dv}e^{-\int_0^x \tilde{\lambda}_y^1 dy}\tilde{\lambda}_x^1 dx\right)
		\end{equation}
		and
	\begin{align} 
		&\mathcal{E}_t\left( Y\left(\int_t^T e^{-\int_0^{T-x} \tilde{\lambda}_u^2 du}e^{-\int_0^x \tilde{\lambda}_y^1 dy}\tilde{\lambda}_x^1dx+e^{-\int_0^{T}\tilde{\lambda}_u^1 du} \right) \right)+ \mathcal{E}_t \left( Y\int_s^t e^{-\int_0^{T-x} \tilde{\lambda}_u^2du} e^{-\int_0^{x} \tilde{\lambda}_y^1dy} \tilde{\lambda}_x^1dx \right) \notag \\
	&= \mathcal{E}_t\left( Y\left(\int_s^T e^{-\int_0^{T-x} \tilde{\lambda}_u^2 du}e^{-\int_0^x \tilde{\lambda}_y^1 dy}\tilde{\lambda}_x^1dx+e^{-\int_0^{T}\tilde{\lambda}_u^1 du} \right) \right), \label{eq:Payoff1Condition2}
	\end{align}
	then  
	\begin{equation}
	\tilde{\mathcal{E}}_s^2\left(\tilde{\mathcal{E}}_t^2 \left( \textbf{1}_{\lbrace \tau_2 > T\rbrace }Y\right)\right)	= \tilde{\mathcal{E}}_s^2\left( \textbf{1}_{\lbrace \tau_2 > T\rbrace }Y\right)	 \quad \tilde{P} \text{-a.s. for all } \tilde{P} \in \tilde{\mathcal{P}}.
	\end{equation}
\end{prop}

\begin{proof}
To prove this result we have to show that the inequalities \eqref{eq:WeakDynamicProgrammingThird7}, \eqref{eq:lemma2}, \eqref{eq:SecondInequality}, \eqref{eq:ThirdInequality}, \eqref{eq:WeakDynamicMiddleTerm3} and \eqref{eq:WeakDynamicMiddleTerm4} are indeed equalities. This can be done by similar arguments as the ones used in the proof of Proposition \ref{prop:PayoffClassicalTowerI}.
\end{proof} 
We now provide some examples where  $\eqref{eq:Payoff1Condition1}$ and $\eqref{eq:Payoff1Condition2}$ are satisfied.
\begin{example} \label{example:Payoff1}
	Let us consider $Y=1$. Assume that $T \leq t+s$, which also implies that $T \leq 2t$.
	In this case the left-hand side in $\eqref{eq:Payoff1Condition1}$ can be rewritten as 
	\begin{equation} \label{eq:Payoff1Condition1Example1}
		\int_s^t  \mathcal{E}_t \left(  e^{-\int_0^{T-x} \tilde{\lambda}_v^2 dv }\right) e^{-\int_0^x \tilde{\lambda}_y^1 dy}\tilde{\lambda}_x^1dx= \int_s^t   e^{-\int_0^{T-x} \tilde{\lambda}_v^2 dv } e^{-\int_0^x \tilde{\lambda}_y^1 dy}\tilde{\lambda}_x^1dx, 
	\end{equation}
	since $e^{-\int_0^{T-x} \tilde{\lambda}_v^2 dv}$ is $\mathcal{F}_t$-measurable and non-negative  for $s < x \leq t$. Note now that 
	\begin{equation*}
		 \int_s^t   e^{-\int_0^{T-x} \tilde{\lambda}_v^2 dv } e^{-\int_0^x \tilde{\lambda}_y^1 dy}\tilde{\lambda}_x^1dx
	\end{equation*}
 is $\mathcal{F}_t$-measurable and non-negative. Thus, the right-hand side in $\eqref{eq:Payoff1Condition1}$ can be written as
	\begin{equation}
		 \mathcal{E}_t \left( \int_s^t   e^{-\int_0^{T-x} \tilde{\lambda}_v^2 dv } e^{-\int_0^x \tilde{\lambda}_y^1 dy}\tilde{\lambda}_x^1dx\right)= \int_s^t   e^{-\int_0^{T-x} \tilde{\lambda}_v^2 dv } e^{-\int_0^x \tilde{\lambda}_y^1 dy}\tilde{\lambda}_x^1dx, \label{eq:Payoff1Condition1Example2}
	\end{equation}
	so that $\eqref{eq:Payoff1Condition1}$ holds by $\eqref{eq:Payoff1Condition1Example1}$ and $\eqref{eq:Payoff1Condition1Example2}$.
Regarding \eqref{eq:Payoff1Condition2}, we get
	\begin{align} 
	&\mathcal{E}_t\left( \int_t^T e^{-\int_0^{T-x} \tilde{\lambda}_u^2 du}e^{-\int_0^x \tilde{\lambda}_y^1 dy}\tilde{\lambda}_x^1dx+e^{-\int_0^{T}\tilde{\lambda}_u^1 du}  \right)+ \mathcal{E}_t  \left(\int_s^t e^{-\int_0^{T-x} \tilde{\lambda}_u^2du} e^{-\int_0^{x} \tilde{\lambda}_u^1du} \tilde{\lambda}_x^1dx \right)\notag \\
	=&\mathcal{E}_t\left( \int_t^T e^{-\int_0^{T-x} \tilde{\lambda}_u^2 du}e^{-\int_0^x \tilde{\lambda}_y^1 dy}\tilde{\lambda}_x^1dx+e^{-\int_0^{T}\tilde{\lambda}_u^1 du}  \right)+ \int_s^t e^{-\int_0^{T-x} \tilde{\lambda}_u^2du} e^{-\int_0^{x} \tilde{\lambda}_u^1du} \tilde{\lambda}_x^1dx \label{eq:Payoff1Condition2Example1} \\
	=&\mathcal{E}_t\left( \int_t^T e^{-\int_0^{T-x} \tilde{\lambda}_u^2 du}e^{-\int_0^x \tilde{\lambda}_y^1 dy}\tilde{\lambda}_x^1dx+e^{-\int_0^{T}\tilde{\lambda}_u^1 du}  + \int_s^t e^{-\int_0^{T-x} \tilde{\lambda}_u^2du} e^{-\int_0^{x} \tilde{\lambda}_u^1du} \tilde{\lambda}_x^1dx \right) \label{eq:Payoff1Condition2Example2} \\
	=&\mathcal{E}_t\left( \int_s^T e^{-\int_0^{T-x} \tilde{\lambda}_u^2 du}e^{-\int_0^x \tilde{\lambda}_y^1 dy}\tilde{\lambda}_x^1dx+e^{-\int_0^{T}\tilde{\lambda}_u^1 du}  \right) \notag, 	
	\end{align}
	where $\eqref{eq:Payoff1Condition2Example1}$ and $\eqref{eq:Payoff1Condition2Example2}$ come from $\eqref{eq:Payoff1Condition1Example2}$ and from the $\mathcal{F}_t$-measurability of the term 
	$$
	\int_s^t e^{-\int_0^{T-x} \tilde{\lambda}_u^2du} e^{-\int_0^{x} \tilde{\lambda}_u^1du} \tilde{\lambda}_x^1dx ,
	$$
	 respectively. 
\end{example}

\begin{remark}
\begin{enumerate}
	\item Under the assumption $T \leq t+s$, equality $\eqref{eq:Payoff1Condition1}$ can also be proved for a general payoff $Y$ as given in Proposition \ref{prop:PayoffClassicalTowerI} by using the same arguments as in Example \ref{example:Payoff1}. However, it is not possible to prove $\eqref{eq:Payoff1Condition2}$ as the term $Y\int_s^t e^{-\int_0^{T-x} \tilde{\lambda}_u^2du} e^{-\int_0^{x} \tilde{\lambda}_u^1du} \tilde{\lambda}_x^1dx  $ is not any longer $\mathcal{F}_t$-measurable. Thus, we have 
	\begin{equation*}
		\mathcal{E}_t \left( Y\int_s^t e^{-\int_0^{T-x} \tilde{\lambda}_u^2du} e^{-\int_0^{x} \tilde{\lambda}_u^1du} \tilde{\lambda}_x^1dx   \right) = \mathcal{E}_t \left( Y \right) \int_s^t e^{-\int_0^{T-x} \tilde{\lambda}_u^2du} e^{-\int_0^{x} \tilde{\lambda}_u^1du} \tilde{\lambda}_x^1dx,  
	\end{equation*}
	which does not allow to do a similar step as in \eqref{eq:Payoff1Condition2Example2}.
	\item In the $G$-setting, equality $\eqref{eq:Payoff1Condition2}$ holds when
	\begin{small} 
		\begin{align*}
		&-\mathcal{E}_t\left( -Y\left(\int_t^T e^{-\int_0^{T-x} \tilde{\lambda}_u^2 du}e^{-\int_0^x \tilde{\lambda}_y^1 dy}\tilde{\lambda}_x^1dx+e^{-\int_0^{T}\tilde{\lambda}_u^1 du} \right) \right)\\
		&= \mathcal{E}_t\left( Y\left(\int_t^T e^{-\int_0^{T-x} \tilde{\lambda}_u^2 du}e^{-\int_0^x \tilde{\lambda}_y^1 dy}\tilde{\lambda}_x^1dx+e^{-\int_0^{T}\tilde{\lambda}_u^1 du} \right) \right), 
	\end{align*}
	\end{small}
	or 
	\begin{equation*}
		-\mathcal{E}_t \left(-Y\int_s^t e^{-\int_0^{T-x} \tilde{\lambda}_u^2du} e^{-\int_0^{x} \tilde{\lambda}_y^1dy} \tilde{\lambda}_x^1dx \right)
		= \mathcal{E}_t \left( Y\int_s^t e^{-\int_0^{T-x} \tilde{\lambda}_u^2du} e^{-\int_0^{x} \tilde{\lambda}_y^1dy} \tilde{\lambda}_x^1dx \right).
	\end{equation*}
\end{enumerate}
\end{remark}

 \section{Superhedging} \label{sec:Superhedging}
By generalizing the superhedging results for payment stream in Section 3 in \cite{bz_2019} we can prove also a dynamic robust superhedging duality in our extended setting. \\ 
Fix $T >0$. We define the filtration $\mathbb{G}^{\tilde{\mathcal{P}},(N)}:=(\mathcal{G}_t^{\tilde{\mathcal{P}},(N)})_{t \in [0,T]}$ by
 \begin{equation*}
 	\mathcal{G}^{\tilde{\mathcal{P}},(N)}_t:=\mathcal{G}_t^{(N),*} \vee \mathcal{N}_T^{\tilde{\mathcal{P}}}, \quad t \in [0,T],
 \end{equation*}
 where $\mathcal{N}_T^{\tilde{\mathcal{P}}}$ is the collection of sets which are $(\tilde{P},\mathcal{G}_T^{(N)})$-null for all $\tilde{P} \in \tilde{\mathcal{P}}$. Let $\tilde{R}:=(\tilde{R}_t)_{t \in [0,T]}$ be a non-negative $\mathbb{G}^{\tilde{\mathcal{P}},(N)}$-adapted process with nondecreasing paths, such that $\tilde{R}_t$ is upper semianalytic for all $t \in [0,T]$ and $\tilde{R}_0=0$. Moreover, consider a process $S=(S_t)_{t \in [0,T]}$ which is $m$-dimensional, $\mathbb{G}^{\tilde{\mathcal{P}},(N)}$-adapted with c\`{a}dl\`{a}g paths and a $(\tilde{P},\mathbb{G}^{\tilde{{P}},(N)})$-semimartingale for all $\tilde{P} \in \tilde{\mathcal{P}}$. Here, $S$ represents the (discounted) tradable assets on the enlarged market. The money market account is set equal to $1$. 
 Furthermore, the set of $m$-dimensional $\mathbb{G}^{\tilde{\mathcal{P}},(N)}$-predictable processes which are $S$-integrable for all $\tilde{P} \in \tilde{\mathcal{P}}$ is denoted by $\tilde{L}(S,\tilde{\mathcal{P}})$ and the admissible strategies on $\tilde{\Omega}$ are given by
\begin{equation*}
\tilde{\bigtriangleup}:=\bigg\lbrace \tilde{\delta} \in \tilde{L}(S,\tilde{\mathcal{P}}): \int^{(\tilde{P})} \tilde{\delta} dS \text{ is a } (\tilde{P},\mathbb{G}_+^{\tilde{{P}},(N)}) \text{-supermartingale for all } \tilde{P} \in \tilde{\mathcal{P}} \bigg\rbrace.
\end{equation*}
We use the notation $\int^{(\tilde{P})} \tilde{\delta} dS := (\int^{(\tilde{P}),t} \tilde{\delta} dS)_{t \in [0,T]}$ for the usual It\^{o} integral under $\tilde{P}$. 
\begin{asum} \label{sigmaSaturated} 
\begin{enumerate} 
\itemsep0pt
	\item $\tilde{\mathcal{P}}$ is a set of sigma martingale measures for $S$, i.e., $S$ is a $(\tilde{P}, \mathbb{G}_+^{{\tilde{P}},(N)})$-sigma-martingale for all $\tilde{P} \in \tilde{\mathcal{P}}$;
	\item $\tilde{\mathcal{P}}$ is saturated: all equivalent sigma-martingale measures of its elements still belong to $\tilde{\mathcal{P}}$;
	\item $S$ has dominating diffusion under every $\tilde{P} \in \tilde{\mathcal{P}}$.
\end{enumerate}
\end{asum}
Next we state generalized versions of Theorem 3.11 and Theorem 3.12 in \cite{bz_2019}.
\begin{prop} \label{superhedging}
Let Assumption \ref{assumptionnutzNew} hold for $\mathcal{P}$ and Assumption \ref{sigmaSaturated} for $\tilde{\mathcal{P}}$, respectively. Consider a cumulative payment stream $\tilde{R}=(\tilde{R}_{s})_{s \in [0,T]}$ with $\tilde{\mathcal{E}}_{s}^N(\tilde{R}_T) \in L^1(\tilde{\Omega})$ for all $s \in [0,T]$. 
Let $t \in [0,T]$. If there exists a $\mathbb{G}^{{\mathcal{P}},(N)}$-adapted process $\tilde{X}=(\tilde{X}_s)_{s \in [0,T]}$ with c\`{a}dl\`{a}g paths, such that for $s \in [0,t],$
\begin{equation*}
	\tilde{X}_s = \tilde{\mathcal{E}}_s^N(\tilde{R}_t) \quad \tilde{P}\text{-a.s.} \text{ for all } \tilde{P} \in \tilde{\mathcal{P}},
\end{equation*}
and if the tower property holds for $\tilde{R}$, i.e., for all $r,s \in [0,t]$ with  $r \leq s$,
\begin{equation*}
	\tilde{\mathcal{E}}_r^N(\tilde{R}_t)=\tilde{\mathcal{E}}_r^N(\tilde{\mathcal{E}}_s^N(\tilde{R}_t)) \quad \tilde{P}\text{-a.s.} \text{ for all } \tilde{P} \in \tilde{\mathcal{P}},
\end{equation*}
then we have the following equivalent dualities for all $\tilde{P} \in \tilde{\mathcal{P}}$ and $0 \leq s \leq t \leq T$
\begin{align}
	\tilde{\mathcal{E}}_s^N(\tilde{R}_t) 
	=& \essinf \lbrace \tilde{v} \text{ is } \mathcal{G}_s^{\tilde{\mathcal{P}},(N)}\text{-measurable}: \exists \tilde{\delta} \in \tilde{\bigtriangleup} \text{ such that } \tilde{v} + \int_s^{(\tilde{P}'),t} \tilde{\delta}_u dS_u \geq \tilde{R}_t \quad \tilde{P}' \text{-a.s. } \nonumber \\ 
	&\text{ for all } \tilde{P}' \in \tilde{\mathcal{P}} \rbrace  \quad \tilde{P} \text{-a.s.}  \notag \\
	=& \essinf \lbrace \tilde{v} \text{ is } \mathcal{G}_s^{\tilde{\mathcal{P}},(N)}\text{-measurable}: \exists \tilde{\delta} \in \tilde{\bigtriangleup} \text{ such that } \tilde{v} + \int_s^{(\tilde{P}'),t} \tilde{\delta}_u dS_u \geq \tilde{R}_t \quad \tilde{P}' \text{-a.s. } \nonumber \\ 
	&\text{ for all } \tilde{P}' \in \tilde{\mathcal{P}}(s;\tilde{P}) \rbrace  \quad \tilde{P} \text{-a.s.} \notag
\end{align}
and 
\begin{align}
	\tilde{\mathcal{E}}_s^N(\tilde{R}_t-\tilde{R}_s)=& \essinf \lbrace \tilde{v} \text{ is } \mathcal{G}_s^{\tilde{\mathcal{P}},(N)}\text{-measurable}: \exists \tilde{\delta} \in \tilde{\bigtriangleup} \text{ such that } \tilde{v} + \int_s^{(\tilde{P}'),t} \tilde{\delta}_u dS_u \geq \tilde{R}_t - \tilde{R}_s \nonumber \\	&\tilde{P}' \text{-a.s. for all } \tilde{P}' \in \tilde{\mathcal{P}} \rbrace  \quad \tilde{P} \text{-a.s.}  \notag \\
	=& \essinf \lbrace \tilde{v} \text{ is } \mathcal{G}_s^{\tilde{\mathcal{P}},(N)}\text{-measurable}: \exists \tilde{\delta} \in \tilde{\bigtriangleup} \text{ such that } \tilde{v} + \int_s^{(\tilde{P}'),t} \tilde{\delta}_u dS_u \geq \tilde{R}_t  - \tilde{R}_s  \nonumber \\ 
	&\tilde{P}' \text{-a.s. for all } \tilde{P}' \in \tilde{\mathcal{P}}(s;\tilde{P}) \rbrace  \quad \tilde{P} \text{-a.s.} \notag
\end{align}

\end{prop}
\begin{proof}
	The proof follows by the same arguments as the proof of Theorem 3.11 in \cite{bz_2019}.
\end{proof}
For $0 \leq s \leq t \leq T$, we define the following set
\begin{align*}
	\tilde{C}_s^t:= \left \lbrace \tilde{\delta} \in \tilde{\Delta}: \tilde{\mathcal{E}}^N_{s_1}(\tilde{R}_t)+\int_{s_1}^{(\tilde{P}),s_2} \tilde{\delta}_u dS_u \geq \tilde{R}_{s_2} \ \tilde{P} \text{-a.s. for all }\tilde{P} \in \tilde{\mathcal{P}} \text{ for all } s \leq s_1 \leq s_2 \leq t \right \rbrace.
\end{align*}

\begin{prop}
	Under the assumptions of Theorem \ref{superhedging}, for $0 \leq s \leq t \leq T$, we have the following statements:
	\begin{enumerate}
		\item The set $\tilde{C}_0^T$ is not empty.
		\item The robust global superhedging price of $\tilde{R}$ is given by $\tilde{\mathcal{E}}^{N}(\tilde{R}_T)$ and the robust local superhedging price of $\tilde{R}$ on the interval $[s,t]$ is given by $\tilde{\mathcal{E}}_s^{N}(\tilde{R}_t - \tilde{R}_s)$.
		\item Optimal superheding strategies exist. 
	\end{enumerate}
\end{prop}
\begin{proof}
	The proof follows by the same arguments as the proof of Theorem 3.12 in \cite{bz_2019}.
\end{proof}

\appendix
\red{ 
\section{Sufficient conditions for the tower property}
For simplicity we consider as in Section \ref{section:DynamicProgramming}
 the case of two default times, i.e. $k=2$.\begin{lemma}
	 Let $Y$ be an upper semianalytic function on $\tilde{\Omega}$ which is $\mathcal{G}^{\mathcal{P},(2)}$-measurable and non-negative and $\varphi$ be the associated function in \eqref{eq:FunctionVarphiDefinition1} such that \eqref{eq:FunctionVarphiDefinition2} holds. If $Y$ satisfies the following equalities $P$-a.s. for all $P \in \mathcal{P}$ and for all $0 \leq s \leq t$
\begin{align}
&  \mathbb{E}^{\hat{P}} \left[  \mathcal{E}_t(\textbf{1}_{\lbrace s < u_1 \leq t \rbrace}  \textbf{1}_{\lbrace \tilde{\tau}_2 >t-u_1 \rbrace}  \mathbb{E}^{\hat{P}} [\textbf{1}_{\lbrace \tilde{\tau}_2 >t-u_1  \rbrace}\varphi(u_1,u_1+\tilde{\tau}_2, \cdot )])\right] \nonumber \\
& =  \mathcal{E}_t \left( \mathbb{E}^{\hat{P}}\left[\textbf{1}_{\lbrace s < u_1 \leq t \rbrace}  \textbf{1}_{\lbrace \tilde{\tau}_2 >t-u_1 \rbrace}\mathbb{E}^{\hat{P}} [\textbf{1}_{\lbrace \tilde{\tau}_2 >t-u_1  \rbrace}\varphi(u_1,u_1+\tilde{\tau}_2, \cdot )]\right]\right) \label{eq:WeakDynamicProgrammingAppendix1}	
\end{align}
and
\begin{align}
& \mathbb{E}^{\hat{P}} \left[ \mathcal{E}_t \left(\mathbb{E}^{\hat{P}} \left[\textbf{1}_{\lbrace s < u_1 \leq t \rbrace}\textbf{1}_{\lbrace \tilde{\tau}_2 >t-u_1  \rbrace}\varphi(u_1,u_1+\tilde{\tau}_2, \cdot )\right]\right) \bigg\vert_{u_1=\tau_1} \right]\nonumber \\
	& =\mathcal{E}_t \left(\mathbb{E}^{\hat{P}}\left[ \mathbb{E}^{\hat{P}} [\textbf{1}_{\lbrace s < u_1 \leq t \rbrace} \textbf{1}_{\lbrace \tilde{\tau}_2 >t- u_1 \rbrace}\varphi(u_1, u_1+\tilde{\tau}_2, \cdot) ]\vert_{u_1=\tau_1} \right] \right)\label{eq:WeakDynamicProgrammingAppendix2}		
\end{align}
and
\begin{align}
	\mathbb{E}^{\hat{P}}\left[ \mathcal{E}_t \left( \textbf{1}_{\lbrace s < \tau_1 \rbrace}  \textbf{1}_{\lbrace \tau_2 \leq t \rbrace}  Y\right)\right] =\mathcal{E}_t \left(\mathbb{E}^{\hat{P}}\left[ \textbf{1}_{\lbrace s < \tau_1 \rbrace}  \textbf{1}_{\lbrace \tau_2 \leq t \rbrace}  Y\right]\right) \label{eq:WeakDynamicProgrammingAppendix3}
\end{align}
and
\begin{align}
	&  \mathcal{E}_t\left( \mathbb{E}^{\hat{P}}[\textbf{1}_{\lbrace t < \tau_1 \rbrace}Y]\right)+  \mathcal{E}_t\left(\mathbb{E}^{\hat{P}}[    \textbf{1}_{\lbrace s < \tau_1 \rbrace}   \textbf{1}_{\lbrace \tau_1 \leq t < \tau_2 \rbrace}Y ]\right) +\mathcal{E}_t\left(\mathbb{E}^{\hat{P}}[ \textbf{1}_{\lbrace s < \tau_1 \rbrace}  \textbf{1}_{\lbrace \tau_2 \leq t \rbrace}  Y]\right) \nonumber\\
	&=  \mathcal{E}_t\left( \mathbb{E}^{\hat{P}}[\textbf{1}_{\lbrace t < \tau_1 \rbrace}Y]+\mathbb{E}^{\hat{P}}[    \textbf{1}_{\lbrace s < \tau_1 \rbrace}   \textbf{1}_{\lbrace \tau_1 \leq t < \tau_2 \rbrace}Y ] +\mathbb{E}^{\hat{P}}[ \textbf{1}_{\lbrace s < \tau_1 \rbrace}  \textbf{1}_{\lbrace \tau_2 \leq t \rbrace}  Y]\right)\label{eq:WeakDynamicProgrammingAppendix4}	
\end{align}
and
\begin{align}
	&\mathbb{E}^{\hat{P}}\left[ \mathcal{E}_t(\textbf{1}_{\lbrace \tilde{\tau}_2 > s-u_1\rbrace} \textbf{1}_{\lbrace u_1+\tilde{\tau}_2 \leq t \rbrace} \varphi(u_1,u_1+\tilde{\tau}_2,\cdot)) \right] =\mathcal{E}_t\left(\mathbb{E}^{\hat{P}}\left[ \textbf{1}_{\lbrace \tilde{\tau}_2 > s-u_1\rbrace} \textbf{1}_{\lbrace u_1+\tilde{\tau}_2 \leq t \rbrace} \varphi(u_1,u_1+\tilde{\tau}_2,\cdot) \right]\right) \label{eq:WeakDynamicProgrammingAppendix5}
\end{align}
and 
\begin{align}
	& \mathcal{E}_t \left( \mathbb{E}^{\hat{P}}[\textbf{1}_{\lbrace u_1 \leq   t \rbrace}\textbf{1}_{\lbrace \tilde{\tau}_2 >t-u_1 \rbrace} \varphi(u_1,u_1+\tilde{\tau}_2, \cdot) ]\right)  +\mathcal{E}_t\left(\mathbb{E}^{\hat{P}}\left[ \textbf{1}_{\lbrace \tilde{\tau}_2 > s-u_1\rbrace} \textbf{1}_{\lbrace u_1+\tilde{\tau}_2 \leq t \rbrace} \varphi(u_1,u_1+\tilde{\tau}_2,\cdot) \right]\right) \nonumber \\
	& =\mathcal{E}_t \left( \mathbb{E}^{\hat{P}}[\textbf{1}_{\lbrace u_1 \leq   t \rbrace}\textbf{1}_{\lbrace \tilde{\tau}_2 >t-u_1 \rbrace} \varphi(u_1,u_1+\tilde{\tau}_2, \cdot) ] +\mathbb{E}^{\hat{P}}\left[ \textbf{1}_{\lbrace \tilde{\tau}_2 > s-u_1\rbrace} \textbf{1}_{\lbrace u_1+\tilde{\tau}_2 \leq t \rbrace} \varphi(u_1,u_1+\tilde{\tau}_2,\cdot) \right]\right)	,\label{eq:WeakDynamicProgrammingAppendix6} 
\end{align}
then the strong dynamic programming principle holds.
\end{lemma}
\begin{proof}
This follows immediately by identifying all the inequalities in the proofs in Section \ref{section:DynamicProgramming}.
\end{proof}
\begin{remark}
	By using Yan's commutability theorem in \cite{1985commutability} it is possible to find some sufficient criteria which allow  change the ``order or integration'' between $\mathbb{E}^{\hat{P}}$ and $\mathcal{E}_t$, which is what is needed in \eqref{eq:WeakDynamicProgrammingAppendix1}-\eqref{eq:WeakDynamicProgrammingAppendix3} and in \eqref{eq:WeakDynamicProgrammingAppendix5}, see for example Appendix B in \cite{bz_2019}. Moreover, note that the Yan's commutability criterium requires a condition on the set of priors $\mathcal{P}$, which shows that the strong dynamic principle depends on $Y$ as well as on the set of probability measures $\mathcal{P}$.
\end{remark}}

\bibliography{2022_10_13_Multiple_Default_Times.bib}{}
\bibliographystyle{plain}
\end{document}